\documentclass[11pt]{article}
\usepackage{fullpage}

\usepackage{comment,amsfonts,amsmath,amsthm,graphicx,%algorithm,
algorithmic,mathtools}

\newcommand{\alert}[1]{\textbf{\color{red}		[[[#1]]]}\marginpar{\textbf{\color{red}**}}\typeout{ALERT:\@\the\inputlineno: #1}}

\usepackage[colorlinks,linkcolor=blue,filecolor=blue,citecolor=blue,urlcolor
=blue,pdfstartview=FitH,linktocpage=true]{hyperref}

\newtheorem*{theorem*}{Theorem}
\newtheorem{theorem}{Theorem}
\newtheorem*{definition*}{Definition}
\newtheorem{definition}{Definition}
\newtheorem*{lemma*}{Lemma}
\newtheorem{lemma}{Lemma}
\newtheorem*{corollary*}{Corollary}
\newtheorem{corollary}{Corollary}
\newtheorem{remark}{Remark}
\newtheorem*{fact*}{Fact}

\newcommand\blfootnote[1]{%
  \begingroup
  \renewcommand\thefootnote{}\footnote{#1}%
  \addtocounter{footnote}{-1}%
  \endgroup
}

\begin{document}
\title{On the Size Overhead of Pairwise Spanners}
\author{Ofer Neiman\footnote{Ben-Gurion University of the Negev, Israel. E-mail: \texttt{neimano@cs.bgu.ac.il}}\qquad  Idan Shabat\footnote{Ben-Gurion University of the Negev, Israel. Email: \texttt{idansha6@gmail.com}}}

\date{}
\maketitle

\blfootnote{Supported by Lynn and William Frankel Center for Computer Sciences and ISF grant 970/21.}

\pagenumbering{gobble}

\begin{abstract}

Given an undirected possibly weighted $n$-vertex graph $G=(V,E)$ and a set $\mathcal{P}\subseteq V^2$ of pairs, a subgraph $S=(V,E')$ is called a ${\cal P}$-{\em pairwise} $\alpha$-spanner of $G$, if for every pair $(u,v)\in\mathcal{P}$ we have $d_S(u,v)\leq\alpha\cdot d_G(u,v)$. The parameter $\alpha$ is called the {\em stretch} of the spanner, and its {\em size overhead} is define as $\frac{|E'|}{|{\cal P}|}$.

A surprising connection was recently discussed between the additive stretch of $(1+\epsilon,\beta)$-spanners, to the hopbound of $(1+\epsilon,\beta)$-hopsets. A long sequence of works showed that if the spanner/hopset has size $\approx n^{1+1/k}$ for some parameter $k\ge 1$, then $\beta\approx\left(\frac1\epsilon\right)^{\log k}$. In this paper we establish a new connection to the {\em size overhead} of pairwise spanners. In particular, we show that if $|{\cal P}|\approx n^{1+1/k}$, then a ${\cal P}$-pairwise $(1+\epsilon)$-spanner must have size at least $\beta\cdot |{\cal P}|$ with $\beta\approx\left(\frac1\epsilon\right)^{\log k}$ (a near matching upper bound was recently shown in \cite{ES23}). That is, the size overhead of pairwise spanners has similar bounds to the hopbound of hopsets, and to the additive stretch  of spanners.

We also extend the connection between pairwise spanners and hopsets to the large stretch regime, by showing nearly matching upper and lower bounds for ${\cal P}$-pairwise $\alpha$-spanners. In particular, we show that if $|{\cal P}|\approx n^{1+1/k}$, then the size overhead is $\beta\approx\frac k\alpha$.

A {\em source-wise spanner} is a special type of pairwise spanner, for which ${\cal P}=A\times V$ for some $A\subseteq V$. A {\em prioritized} spanner is given also a ranking of the vertices $V=(v_1,\dots,v_n)$, and is required to provide improved stretch for pairs containing higher ranked vertices. By using a sequence of reductions: from pairwise spanners to source-wise spanners to prioritized spanners, we improve on the state-of-the-art results for source-wise and prioritized spanners. Since our spanners can be equipped with a path-reporting mechanism, we also substantially improve the known bounds for path-reporting prioritized distance oracles. Specifically, we provide a path-reporting distance oracle, with size $O(n\cdot(\log\log n)^2)$, that has a constant stretch for any query that contains a vertex ranked among the first $n^{1-\delta}$ vertices (for any constant $\delta>0$). Such a result was known before only for non-path-reporting distance oracles.

\end{abstract}

\newpage
\tableofcontents
\newpage

\pagenumbering{arabic}
\setcounter{page}{1}

\section{Introduction}

Spanners and hopsets are basic graph structures that have been extensively studied, and found numerous applications in graph algorithms, distributed computing, geometric algorithms, and many more. In this work we study pairwise spanners, and prove an intriguing relation between pairwise spanners, general spanners, and hopsets. We then derive several new results on {\em source-wise} and {\em prioritized} spanners and distance oracles.

Let $G=(V,E)$ be an undirected $n$-vertex graph, possibly with nonnegative weights on the edges. An $(\alpha,\beta)$-{\em spanner} is a subgraph $S=(V,E')$ such that for every pair $u,v\in V$
\begin{equation}\label{eq:spanner}
d_S(u,v)\le \alpha\cdot d_G(u,v)+\beta~,
\end{equation}
where $d_G,d_S$ stand for the shortest path distances in $G,S$, respectively. A spanner is called {\em near-additive} if its multiplicative stretch is $\alpha=1+\epsilon$ for some small $\epsilon>0$.\footnote{Spanners with additive stretch $\beta>0$ are usually defined on unweighted graphs. A possible variation for weighted graph has also been studied \cite{EGN19,ABDKS20}, in which $\beta$ is multiplied by the weight of the heaviest edge on some $u-v$ shortest path.}
Given a set ${\cal P}\subseteq V^2$ of pairs, a {\em ${\cal P}$-pairwise} $\alpha$-spanner for $G$ has to satisfy Equation (\ref{eq:spanner}) (with $\beta=0$) only for pairs $(u,v)\in{\cal P}$.  The {\em size overhead} of a pairwise spanner is defined as $\frac{|E'|}{|{\cal P}|}$.
A {\em source-wise spanner}\footnote{Source-wise spanners were called {\em terminal spanners} in \cite{EFN17}.} is a special type of pairwise spanner, in which ${\cal P}=A\times V$ for some set $A\subseteq V$.

An $(\alpha,\beta)-${\em hopset} for a graph $G=(V,E)$ is a set of edges $H\subseteq\binom{V}{2}$ such that for every $u,v\in V$,
\[d^{(\beta)}_{G\cup H}(u,v)\leq\alpha\cdot d_G(u,v)~.\]
Here, $G\cup H$ denotes the graph $G$ with the additional edges of $H$, and the weight function $w(x,y)=d_G(x,y)$ for every $\{x,y\}\in H$. The notation $d^{(\beta)}_{G\cup H}(u,v)$ stands for the weight of the shortest $u-v$ path in this graph, among the ones that contain at most $\beta$ edges. %A {\em near-exact hopset} is a hopset with stretch $\alpha=1+\epsilon$.
A hopset is called {\em near-exact} if its stretch is $\alpha=1+\epsilon$ for some small $\epsilon>0$.

\subsection{Pairwise Spanners, Near-additive Spanners, and Hopsets}
In \cite{EN22,KP22}, a connection was discussed between the additive stretch $\beta$ of near-additive spanners, to the hopbound $\beta$ of near-exact hopsets. Given an integer parameter $k\ge 1$, that governs the size of the spanner/hopset to be $\approx n^{1+1/k}$, a sequence of works on spanners \cite{EP04,TZ06,P08,ABP18,EN16b,EGN19} and on hopsets \cite{C00,B09,HKN14,MPVX15,EN16,HP17,EN19}, culminated in achieving 
$\beta=O\left(\frac{\log k}{\epsilon}\right)^{\log k}$ 
for both spanners and hopsets. In \cite{ABP18} a lower bound of $\beta\ge\Omega_k\left(\frac{1}{\epsilon}\right)^{\log k-1}$ was shown%\footnote{In fact, the lower bound for hopsets has $k$ instead of $\log k$ in the denominator.}
. So whenever $\epsilon$ is sufficiently small, we have that 
$\beta\approx \left(\frac1\epsilon\right)^{\log k}$.
As spanners and hopsets are different objects, and $\beta$ has a very different role for each, this similarity is somewhat surprising (albeit comparable techniques are used for the constructions). 

In this paper we establish an additional connection, to pairwise spanners. Here the parameter $k$ governs the number of pairs, 
$|{\cal P}|\approx n^{1+1/k}$,\footnote{In this paper we focus on ${\cal P}$-pairwise spanners with $|{\cal P}|\ge n$. Note that for $|{\cal P}|<n$ there is a trivial $\Omega(n)$ lower bound on the size.} 
and we show that ${\cal P}$-pairwise $(1+\epsilon)$-spanners must have size at least $\beta\cdot|{\cal P}|$ with $\beta\approx \left(\frac1\epsilon\right)^{\log k}$. So the parameter $\beta$, which is the additive stretch for near-exact spanners, and is the hopbound for hopsets, now plays the role of the {\em size overhead} for pairwise spanners.

An exact version of pairwise spanners (i.e., with $\alpha=1$) was introduced in \cite{CE05}, where they were called {\em distance preservers}. The sparsest distance preservers that are currently known are due to \cite{CE05,B21}. For an $n$-vertex graph and a set of pairs $\mathcal{P}$, they have size $O\left(\min\{n^{2/3}|\mathcal{P}|,n|\mathcal{P}|^{1/2}\}+n\right)$. So whenever $|{\cal P}|\le n^{2-\delta}$ for some constant $\delta>0$, the size overhead $\beta$ is polynomial in $n$.

When considering near-exact pairwise spanners, following \cite{KP22}, in  \cite{ES23} a ${\cal P}$-pairwise $(1+\epsilon)$-spanner of size  $\approx\beta\cdot|{\cal P}|$, with $\beta= O\left(\frac{\log k}{\epsilon}\right)^{\log k}$  was shown (where $k$ is such that $|{\cal P}|\approx n^{1+1/k}$). 

\noindent{\bf Our Results:} In this paper we show a near matching lower bound for ${\cal P}$-pairwise $(1+\epsilon)$-spanners, that $\beta\ge\Omega\left(\frac{1}{\log^2 k\cdot\epsilon}\right)^{\log k-1}$, establishing that the size overhead must be $\beta\approx \left(\frac1\epsilon\right)^{\log k}$. We derive this lower bound by a delicate adaptation of the techniques of \cite{ABP18} to the case of pairwise spanners.

\subsubsection{Larger Stretch Regime} 
As the stretch grows to be bounded away from 1, the connection between spanners and hopsets diminishes (note that the lower bound of \cite{ABP18} is meaningless for large stretch). For size $\approx n^{1+1/k}$ and constant stretch $\alpha$, one can obtain $(\alpha,\beta)$-spanners and hopsets with
$\beta\approx k^{1+1/\log\alpha}$ \cite{BP20,NS22}. However, as the stretch grows to $\alpha=k^{\delta}$ (for some constant $0<\delta<1$), there is a hopset with $\beta\approx k^{1-\delta}$ \cite{BP20,NS22}, while spanners must have $\beta=\Omega(k)$.\footnote{Note that an $(\alpha,\beta)$-spanner is also an $(\alpha+\beta,0)$-spanner. Thus, with size $n^{1+1/k}$ it must have $\alpha+\beta\ge k$ \cite{ADDJS93}.}
In \cite{NS22}, a lower bound of $\beta=\Omega\left(\frac k\alpha\right)$ for hopsets was shown.

\noindent{\bf Our Results:} In this work we exhibit a similar tradeoff for pairwise spanners, as is known for hopsets.\footnote{For stretch $3+\epsilon$ with $0<\epsilon<1$, a pairwise spanner of size $\beta\cdot|{\cal P}|$ with $\beta\approx k^{\log 3 +O(\log(1/\epsilon))}$ was shown in \cite{ES23}. However, we are not aware of any result for larger stretch.} In particular,  
we devise a ${\cal P}$-pairwise $\alpha$-spanner with size $\beta\cdot|{\cal P}|$ where $\beta\approx \frac{k^{1+1/\log\alpha}}{\alpha}$ (and $k$ is such that $|{\cal P}|\approx n^{1+1/k}$).
We also show a lower bound $\beta=\Omega\left(\frac{k}{\alpha}\right)$ on the size overhead of such pairwise spanners. 
Note that this lower bound nearly matches the upper bound for $\alpha\ge k^\delta$ for any constant $\delta>0$, and in this regime we have $\beta\approx \frac k\alpha$ for both hopsets and pairwise spanners. 

\subsection{Source-wise and Prioritized Spanners and Distance Oracles}

Source-wise spanners were first studied by \cite{RTZ05,Par14}. Given an integer parameter $k\ge 1$ and a subset $A\subseteq V$ of sources, they showed a source-wise spanner with stretch $2k-1$ and size $O(kn\cdot |S|^{1/k})$ (the spanners of \cite{RTZ05} could also be distance oracles, while \cite{Par14} had slightly improved stretch $2k-2$ for some pairs). By increasing the stretch to $4k-1$, \cite{EFN17} obtained improved size $O(n+\sqrt{n}\cdot |S|^{1+1/k})$. The current state-of-the-art result is by \cite{KP22}. For any $0<\epsilon<1$, they gave a source-wise spanner with stretch $4k-1+\epsilon$ and size $O(n+|S|^{1+1/k})\cdot\beta$, where $\beta\approx\left(\frac{\log\log n}{\epsilon}\right)^{\log\log n}$.

Given an undirected possibly weighted graph $G=(V,E)$, a {\em prioritized} metric structure (such as spanner, hopset, distance oracle) is also given an ordering of the vertices $V=(v_1,\dots,v_n)$, and is required to provide improved guarantees (e.g., stretch, hopbound, query time, etc.) for higher ranked vertices.

In \cite{EFN15}, among other results, prioritized distance oracles for general graphs were shown.\footnote{A distance oracle is a data structure that can efficiently report approximate distances. The parameters of interest are usually the query time, the size, and the multiplicative stretch.} For an $n$-vertex graph with priority ranking $(v_1,\dots,v_n)$, a distance oracle with size $O(n\log\log n)$, query time $O(1)$, and stretch $O\left(\frac{\log n}{\log n-\log j}\right)$ for any pair containing $v_j$ was shown.\footnote{Additional results with even smaller size and larger stretch were shown in \cite{EFN15} as well.} Note that the stretch is constant for any pair containing a vertex ranked among the first $n^{1-\delta}$ vertices (for constant $\delta>0$). However, that distance oracle could only return distances, and could not report paths. In \cite{EFN15} an additional {\em path-reporting} distance oracle was shown. Given an integer parameter $k\ge 1$, it had size $O(kn^{1+1/k})$, stretch  $2\left\lceil\frac{k\log j}{\log n}\right\rceil-1$ for pairs containing $v_j$, and query time $O\left(\left\lceil\frac{k\log j}{\log n}\right\rceil\right)$.
Note that this oracle size is $\Omega(n\log n)$ for any $k$, and with size $O(n\log n)$ it has prioritized stretch $2\log j-1$, which is much worse than the stretch that the previous oracle had for higher ranked vertices.

\noindent{\bf Our Results:}
By using a sequence of reductions, from pairwise spanners to source-wise spanners to prioritized spanners, we improve on the state-of-the-art results. In particular, for any integer parameter $k\ge 1$, any $0<\epsilon<1$, and a subset $A\subseteq V$, we obtain a source-wise spanner with stretch $4k-1+\epsilon$ and size $O(n\log\log n+ |S|^{1+1/k}\cdot\beta)$  where $\beta\approx\left(\frac{\log\log n}{\epsilon}\right)^{\log\log n}$. Note that  $\beta$ only multiplies the term $|S|^{1+1/k}$, which could be much smaller than $n$, while in \cite{KP22} the size is always at least $\Omega(n\cdot\beta)$.

Using the fact that the constructions of pairwise spanners can also be {\em path-reporting}, and that our reductions preserve this property, we devise {\em path-reporting} distance oracles with size $O(n(\log\log n)^2)$, query time $O(1)$, and prioritized stretch $O\left(\frac{\log n}{\log n-\log j-o(\log\log n)^2}\right)$. That is, we nearly achieve the improved parameters of the non-path-reporting oracles of \cite{EFN15} (in particular, we get constant stretch for any pair containing a vertex ranked among the first $n^{1-\delta}$ vertices). %In fact, our results are more general, and can also provide \textit{spanners} with such size and prioritized stretch.

\subsection{Technical Overview}

\subsubsection{Lower Bound for Near-Exact Pairwise Spanners}

%\alert{I deleted the name after the citation, I think in annonymous submission it is best to not mention names. Also, other people whose name we didn't mention might get offended..}
%Ok, no problem

In \cite{ABP18}, a series of lower bounds was shown for graph compression structures, such as near-additive spanners, emulators, distance oracles and hopsets. Specifically, for the first three, \cite{ABP18} proved that for any positive integer $\kappa$, any such structure with size $O(n^{1+\frac{1}{2^\kappa-1}})$, that preserves distances $d$ up to $(1+\epsilon)d+\beta$, must have $\beta=\Omega\left(\frac{1}{\epsilon\kappa}\right)^{\kappa-1}$. Almost the same lower bound\footnote{For hopsets of size $O(n^{1+\frac{1}{2^\kappa-1}})$ and stretch $1+\epsilon$, the actual lower bound on $\beta$ that was proved in \cite{ABP18} was $\beta=\Omega\left(\frac{1}{\epsilon2^\kappa}\right)^{\kappa-1}$, instead of $\Omega\left(\frac{1}{\epsilon\kappa}\right)^{\kappa-1}$. It was suggested though, that a more careful analysis might achieve the same lower bound as of $(1+\epsilon,\beta)$-spanners.} was proved by \cite{ABP18} for $(1+\epsilon,\beta)$-hopsets of size $O(n^{1+\frac{1}{2^\kappa-1}})$.

All the lower bounds mentioned above were demonstrated in \cite{ABP18} on essentially the same graph. The construction of this graph relied on a \textit{base graph} $\ddot{B}$, that was presented in \cite{AB17} and had a \textit{layered structure}. That is, the vertices of $\ddot{B}$ are partitioned into $2l+1$ subsets, such that any edge can only be between vertices of adjacent layers. The first and the last layers of $\ddot{B}$ are called \textit{input ports} and \textit{output ports} respectively. This graph also had a set $\ddot{\mathcal{P}}$ of pairs of input and output ports $(u,v)$ such that there is a unique shortest path $P_{u,v}$ from $u$ to $v$ in $\ddot{B}$, that visits every layer exactly once. Moreover, each edge of $\ddot{B}$ is labeled by a label from a set $\ddot{\mathcal{L}}$, such that the edges of every such path $P_{u,v}$ are alternately labeled by a {\em unique} pair of labels $a,b\in\ddot{\mathcal{L}}$ (meaning that no other pair $(u',v')\in\ddot{\mathcal{P}}$ has its shortest path $P_{u',v'}$ alternately labeled by the same labels $a,b$).

On top of the base graph $\ddot{B}$, a graph $\mathcal{H}_\kappa$ was constructed, for every positive integer $\kappa$. In fact, for \textit{hopsets}, a slightly different graph was constructed in \cite{ABP18} than for \textit{near-additive spanners} (denoted as $\mathcal{H}_\kappa$ in both cases). The two constructions are recursive, with different base cases. Moreover, since the construction for near-additive spanners must be unweighted, edges with large weight in the construction for hopsets, must be replaced by long unweighted paths. In this paper, we construct a graph, still denoted as $\mathcal{H}_\kappa$, that is actually a \textit{combination} of their two constructions. In particular, we use the same base case as for near-additive spanners, which is simply a complete bipartite graph $K_{p,p}$, while using edges of large weight, similarly to the construction of $\mathcal{H}_\kappa$ for hopsets. In addition, we shift the indices of the sequence $\{\mathcal{H}_\kappa\}$ by $1$. We describe here the construction of $\mathcal{H}_1$, as the general construction is more involved and appears fully in Section \ref{sec:NearExactLowerBound}.

The graph $\mathcal{H}_1$ is the same graph as $\ddot{B}$, where every vertex in its $2l-1$ middle layers is replaced by a copy of the complete bipartite graph $K_{p,p}$, where $p$ is the number of labels in $\ddot{\mathcal{L}}$ (for $\kappa>1$, these vertices are replaced by copies of $\mathcal{H}_{\kappa-1}$). Each vertex in each side of the copy of $K_{p,p}$ is mapped to a label in $\ddot{\mathcal{L}}$. An original edge of $\ddot{B}$, that had label $a$, now connects these copies of $K_{p,p}$, by their vertices that correspond to the label $a$ (see Figure \ref{fig:NearExactLowerBound} in Section \ref{sec:NearExactLowerBound} for a detailed illustration). In addition, these edges are assigned with a large weight of $2l-1$ (for $\kappa>1$, this weight is $(2l-1)^\kappa$). This means that any shortest path in $\ddot{B}$ between a pair $(u,v)\in\ddot{\mathcal{P}}$ now must pass through $2l-1$ copies of $K_{p,p}$. Other paths that connect $u,v$, on the other hand, must visit a layer more than once, and therefore their weight is larger than $d_{\mathcal{H}_1}(u,v)$ by at least $2(2l-1)$. Choosing $l\approx\frac{1}{\epsilon}$ (for $\kappa>1$, $l\approx\left(\frac{1}{\epsilon\kappa}\right)^\kappa$), this means that paths other than the unique shortest path have stretch more then $1+\frac{2(2l-1)}{d(u,v)}\approx1+\frac{2(2l-1)}{2l(2l-1)}\approx1+\epsilon$. Hence, any $\ddot{\mathcal{P}}$-pairwise $(1+\epsilon)$-spanner must contain the unique shortest paths between any pair $(u,v)\in\ddot{\mathcal{P}}$.

In our proof of a lower bound for pairwise spanners, we utilize an additional property of the graph $\mathcal{H}_1$. Recall that the edges of $P_{u,v}$ (a shortest path in $\ddot{B}$ that connects some $(u,v)\in\ddot{\mathcal{P}}$) are alternately labeled by a unique pair of labels $a,b\in\ddot{\mathcal{L}}$. This means that in $\mathcal{H}_1$, shortest paths that connect different pairs $(u,v),(u',v')\in\ddot{\mathcal{P}}$ cannot share an edge of a copy of $K_{p,p}$. This is due to the fact that if such path goes through the edge $(x,y)$ of a copy of $K_{p,p}$, and $x$ corresponds to a label $a$, and $y$ to a label $b$, then it uniquely determines the pair of labels $a,b$ of the path $P_{u,v}$. The result is that any $\ddot{\mathcal{P}}$-pairwise $(1+\epsilon)$-spanner for $\mathcal{H}_1$ must contain a disjoint set of $2l-1$ edges, for every $(u,v)\in\ddot{\mathcal{P}}$. Hence, its size must be at least 
\[(2l-1)\cdot|\ddot{\mathcal{P}}|=\Omega\left(\frac{1}{\epsilon}\cdot|\ddot{\mathcal{P}}|\right)~.\]

For larger $\kappa$'s, the number of the edges that these paths do not share grows to $\beta=\Omega_\kappa\left(\frac{1}{\epsilon}\right)^\kappa$, therefore the lower bound for the size is $\beta|\mathcal{P}_\kappa|$, where $\mathcal{P}_\kappa$ is the corresponding set of pairs of the graph $\mathcal{H}_\kappa$. It can be proved that the number of pairs in $\mathcal{P}_\kappa$ is approximately $n^{1+\frac{1}{2^{\kappa+1}-1}}$, where $n$ is the number of vertices in $\mathcal{H}_\kappa$.

\subsubsection{Lower Bound for Pairwise Spanners with Large Stretch}

Our proof of a lower bound for pairwise spanners with stretch larger than $1+\epsilon$ (typically a constant stretch, or stretch between $2$ and $k$, where $k$ is such that $|\mathcal{P}|\approx n^{1+\frac{1}{k}}$) is demonstrated on an unweighted graph with high number of edges and high girth (the length of the minimal cycle). This kind of graphs was used in \cite{TZ01,NS22} to show lower bounds for spanners, distance oracles and hopsets. Specifically, we use a graph $G=(V,E)$ from \cite{LPS88}, that is regular, has $\Omega(n^{1+\frac{1}{k}})$ edges, and has girth larger than $k$.

Given a desired stretch $\alpha>1$, we consider pairs of vertices in $G$ of distance $\delta=\left\lfloor\frac{k}{\alpha+1}\right\rfloor$, henceforth $\delta$-pairs. The useful property of $\delta$-pairs is that due to the high girth of $G$, they have a unique shortest path that connects them, while any other path must have stretch more that $\alpha$. Note the similarity of this property to the property of the unique shortest paths from the lower bound for near-exact pairwise spanners. This property implies that any $\mathcal{P}$-pairwise $\alpha$-spanner, for any set $\mathcal{P}$ of $\delta$-pairs, must contain all the unique shortest paths that connect the $\delta$-pairs of $\mathcal{P}$. We call such paths the $\delta$-paths that correspond to the $\delta$-pairs in $\mathcal{P}$.

Then, we use the regularity of the graph $G=(V,E)$, as well as the high girth of $G$, to prove some combinatorial properties of $G$. We prove that there are many $\delta$-pairs - approximately $n^{1+\frac{\delta}{k}}$ pairs - and that each edge $e\in E$ participates in a large number of corresponding $\delta$-paths - approximately $\delta n^{\frac{\delta-1}{k}}$. %This is done by considering the BFS trees of every vertex, and noticing that there are no "collisions" in it up to level $\delta$. Otherwise, this is a contradiction to the girth property of $G$. By regularity, we know that each vertex has around $n^{\frac{1}{k}}$ neighbors, and therefore these BFS trees have around $n^{\frac{\delta}{k}}$ vertices in their $\delta$-th layer.

To finally prove our lower bound for pairwise $\alpha$-spanners, we sample a set $\mathcal{P}$ of $O\left(\frac{ n^{1+\frac{1}{k}}}{\delta}\right)$ $\delta$-pairs out of all the possible $n^{1+\frac{\delta}{k}}$ such pairs. This means that the sample probability is $\frac{1}{\delta}n^{-\frac{\delta-1}{k}}$. But since the number of $\delta$-pairs whose corresponding $\delta$-path pass through a specific edge $e\in E$ is approximately $\delta n^{\frac{\delta-1}{k}}$, we expect that %many of the edges $e\in E$ to still be covered by at least one $\delta$-path of a pair in $\mathcal{P}$. Namely, 
a constant fraction of the edges of $G$ still participate in a unique shortest path of the $\delta$-pairs in $\mathcal{P}$. Every $\mathcal{P}$-pairwise $\alpha$-spanner must contain these edges, and therefore must have size 
\[\Omega(|E|)=\Omega(n^{1+\frac{1}{k}})=\Omega(\delta|\mathcal{P}|)=\Omega\left(\frac{k}{\alpha}|\mathcal{P}|\right)~.\]
This proves our lower bound of $\beta=\Omega\left(\frac{k}{\alpha}\right)$.

\subsubsection{Upper Bound for Pairwise Spanners with Large Stretch}

The state-of-the-art pairwise $(1+\epsilon)$-spanners and pairwise $(3+\epsilon)$-spanners of \cite{ES23} were achieved using the following \textit{semi-reduction} from hopsets. Suppose that the hopset $H$ has stretch $\alpha$ and hopbound $\beta$, for a graph $G=(V,E)$. Given a set of pairs of vertices $\mathcal{P}$, we construct a pairwise spanner $S$ that consists of all the $\beta$-edges paths in $G\cup H$ that connect pairs in $\mathcal{P}$ and have stretch $\alpha$. Of course, these paths include edges of $H$ that are not allowed to be on the final pairwise spanner $S$ (as they don't exist in $G$). Therefore, an additional step is required, that adds more edges to $S$, such that every edge $(x,y)\in H$ will have a path in $S$ with weight $w(x,y)$. This way, the distance between every pair $(u,v)\in\mathcal{P}$ is the same distance as in the graph $G\cup H$, which is at most $\alpha\cdot d_G(u,v)$. The size of the pairwise spanner $S$ is $\beta\cdot|\mathcal{P}|$, plus the number of edges that are required to preserve the distances between every $(x,y)\in H$.

In \cite{ES23}, preserving the distances between pairs in $H$ is done by relying on the specific properties of the hopset $H$. Namely, it was observed that the relevant hopset $H$ has small \textit{supporting size} - the minimal number of edges required to preserve the distances in $H$. We, however, take an alternative approach. Instead of relying on specific properties of the hopset $H$ to preserve it \textit{accurately}, we use a pairwise spanner with low stretch, to preserve the distance in $H$ \textit{approximately}. In fact, we use the very same $H$-pairwise $(1+\epsilon)$-spanner of \cite{ES23} to do that.

Considering the process described above, we are only left to choose the hopset $H$ that we use, to achieve a pairwise spanner with relatively large stretch. The advantage of having a large stretch $\alpha$, is that the hopbound $\beta$ can be much smaller, and as a result, so is the size overhead of the resulting pairwise spanner. In particular, for pairwise $(1+\epsilon)$-spanners, we know by Section \ref{sec:NearExactLowerBound} that the size overhead must be $\Omega_k\left(\frac{1}{\epsilon}\right)^{\log k}$. But for a larger stretch $O(\alpha)$, we can get a much smaller hopbound, and therefore a much smaller size overhead, of $\beta=\frac{k^{1+\frac{2}{\ln\alpha}}}{\alpha}$. This is done by using the state-of-the-art hopsets with this type of stretch, from \cite{BP20,NS22}.

However, the process of directly applying a pairwise spanner with low stretch on a hopset described above, results in a somewhat large size of the pairwise spanner. This is because the size of the hopset, which is roughly $n^{1+\frac{1}{k}}$, is multiplied by the size overhead coefficient $\beta$, which is at least $\textrm{poly}(k)$.  To reduce this size, and avoid the additional term of $\textrm{poly}(k)\cdot n^{1+\frac{1}{k}}$, we eventually do use certain properties of the hopsets of \cite{NS22}. Specifically, we prove that their hopsets could be partitioned into three sets $H_1,H_2,H_3$, such that $H_1,H_2$ can be efficiently preserved, while $H_3$ is relatively small. Thus, we can use a pairwise spanner with low stretch only on $H_3$, instead of using it on the whole hopset $H$. Then, the coefficient $\beta\approx poly(k)$ multiplies only the size of $H_3$, which is significantly smaller than $n^{1+\frac{1}{k}}$.

\subsubsection{Subset, Source-wise and Prioritized Spanners}

Our new results for subset, source-wise and prioritized spanners are achieved via a series of reductions. These reductions implicitly appeared in \cite{ES23,EFN15}. The first reduction receives a pairwise spanner and uses it in order to construct a {\em subset} spanner. A {\em subset spanner} is a special type of pairwise spanner, for which ${\cal P}=A\times A$ for some $A\subseteq V$. The second is a quite simple reduction that turns a subset spanner into a source-wise spanner with almost the same properties. The last reduction uses source-wise spanners to construct prioritized spanners. Using the new state-of-the-art pairwise spanners, we achieve new results for each of these types of spanners.

Besides these new results, we believe that these reductions themselves could be of interest. It is immediate to find backwards reductions, from prioritized spanners to source-wise spanners, and from source-wise spanners to subset spanners. This means that, in a sense, these three\footnote{Note that the reduction that constructs a subset spanner, given a pairwise spanner, remains a one-way reduction, in the sense that no efficient reduction in the other direction is known.} types of spanners are equivalently hard to construct. Every new construction of one of these spanners immediately implies new constructions for the others.

Next we shortly describe the three reductions mentioned above. 

\vspace{5mm}
\noindent{\textbf{From pairwise to subset spanners.}}
Let $G=(V,E)$ be an undirected weighted $n$-vertex graph, and let $A\subseteq V$ be a subset, for which we want to construct a subset spanner. We consider the graph $K=(A,\binom{A}{2})$, where every pair of vertices $u,v\in A$ is connected by an edge of weight $d_G(u,v)$. On the graph $K$, we apply a known construction of \textit{emulator} (see Definition \ref{def:PathReportingEmulator}). That is, we find a small graph $R=(A,E')$, such that
\[d_G(u,v)=d_K(u,v)\le d_R(u,v)\leq\alpha_E\cdot d_G(u,v)~,\]
for any $u,v\in A$. The parameter $\alpha_E$ is the stretch of the emulator $R$. For this purpose, we use either the distance oracle of Thorup and Zwick from \cite{TZ01}, which can be thought of as an emulator with stretch $\alpha_E=2k-1$ and size $O(k|A|^{1+\frac{1}{k}})$, or the emulator from \cite{ES23}, that has stretch $\alpha_E=O(k)$ and size $O(|A|^{1+\frac{1}{k}})$. Here, $k$ is a positive integer parameter for our choice. These two emulators are \textit{path-reporting}, meaning that given $u,v\in A$, they can also efficiently report a path of stretch $\alpha_E$ inside the emulator itself.

Next, we consider the set $R$ as a set of pairs of vertices from $G$, and apply a pairwise spanner on this set. We use the existing pairwise spanners from \cite{ES23} that have stretch either $\alpha_P=1+\epsilon$ or $\alpha_P=3+\epsilon$, and size\footnote{In \cite{ES23}, the size of these pairwise spanners was presented as $O(|\mathcal{P}|\cdot\beta+n\log k+n^{1+\frac{1}{k}})$, where $\beta\approx \left(\frac{\log k}{\epsilon}\right)^{\log k}$, or $\beta=k^{O\left(\log\frac{1}{\epsilon}\right)}$, and $\mathcal{P}$ is the set of pairs. We choose here $k=\log n$ since for small enough $\mathcal{P}$, this is the most sparse that these spanners can get. In our case, the set $\mathcal{P}$ is the emulator $R$, which is indeed small enough.} 
\[O(|R|\cdot\beta+n\log\log n)~.\]
Here, $\beta\approx \left(\frac{\log\log n}{\epsilon}\right)^{\log\log n}$ when $\alpha_P=1+\epsilon$, and $\beta=(\log n)^{O\left(\log\frac{1}{\epsilon}\right)}$ when $\alpha_P=3+\epsilon$ (there is also additional dependency on $\epsilon$ in the size of the pairwise spanner, in case that $\alpha_P=1+\epsilon$). These two pairwise spanners are also path-reporting.

Now fix a pair of vertices $(u,v)\in A^2$. In the graph $K$, the shortest path between $u,v$ is the single edge $(u,v)$, that has weight $d_G(u,v)$, and thus in $R$ there is a path between $u,v$ that has weight $\alpha_E\cdot d_G(u,v)$. Note that this path consists of edges of $R$, which are not actual edges of the graph $G$. For that reason, we use our pairwise spanner, to find, for every edge $e$ on this path, a path in $G$ that replaces $e$ and has weight of at most $\alpha_P\cdot w(e)$. The result is a path in $G$, with stretch $\alpha_E\cdot\alpha_P$. That is, the resulting spanner has a $\alpha_E\cdot\alpha_P$-stretch path for every pair of vertices in $A$. Hence, this is a subset spanner. Note that since both the emulator $R$ and the pairwise spanner we used were path-reporting, then so is the resulting subset spanner.

\vspace{5mm}
\noindent{\textbf{From subset to source-wise spanner.}}
Let $G=(V,E)$ be an undirected weighted graph, and let $A\subseteq V$ be a subset. Suppose that $S$ is an $A$-subset spanner with stretch $\alpha$. To construct a source-wise spanner for $A$, we just add to $S$ a shortest path, from every vertex $v\in V$ to its closest vertex $p(v)$ in $A$. If the shortest paths and the vertices $p(v)$ are chosen in a consistent manner, it is not hard to prove that the added edges form a \textit{forest}. Thus, the increase in the size of the spanner $S$ is by at most $n-1$ edges, and also we can easily navigate from a vertex $v\in V$ to the corresponding $p(v)\in A$ (we simply go towards the root of $v$'s tree).

%\alert{I think we can safely remove the stretch analysis, since it isn't really new or very interesting. See if you're ok with this}
%Idan: yes, it's ok.

%For analysing the stretch of the new spanner, fix some $a\in A$ and $v\in V$ of distance $d_G(a,v)=\Delta$. Notice that by definition of $p(v)$, the distance between $v,p(v)$ is at most $\Delta$. By the triangle inequality,
%\[d_G(p(v),a)\leq d_G(p(v),v)+d_G(v,a)\leq2\Delta~.\]
%Therefore, in the spanner $S$, there is a path of weight at most $\alpha\cdot2\Delta$ between $p(v),a$. Together with the shortest path from $v$ to $p(v)$, we get a path of weight $(2\alpha+1)\Delta$ between $v,a$. That is, the new spanner is an $A$-source-wise spanner with stretch $2\alpha+1$. 

Given any $v\in V$ and $a\in A$, the spanner contains the $v-p(v)$ shortest path, concatenated with the path in $S$ from $p(v)$ to $a$. It can be shown that it has at most $2\alpha+1$ stretch. Also, if the subset spanner $S$ is path-reporting, then so is the new source-wise spanner.% is path-reporting as well, since for every pair $v,a$ as above, we can efficiently recover the paths from $v$ to $p(v)$ and from $p(v)$ to $a$.

%\alert{Since this is not new, maybe we can shorten this and the next part.}

\vspace{5mm}
\noindent{\textbf{From source-wise to prioritized spanner.}}
In the setting of prioritized spanners, we get a priority ranking of the vertices of $V$: $(v_1,v_2,...,v_n)$. To construct a prioritized spanner based on source-wise spanners, we consider a sequence of non-decreasing \textit{prefixes} of this list:
\[A_1=\{v_1,...,v_{f(1)}\},\;\;A_2=\{v_1,...,v_{f(1)},...,v_{f(2)}\},\;\cdots\;A_T=\{v_1,...,v_{f(1)},...,v_{f(2)},...,v_{f(T)}\}~,\]
where $f$ is some non-decreasing function, and we apply a source-wise spanner on each of them. The idea is that in the construction of the source-wise spanners for the first prefixes in the list, we may use a smaller stretch. Then, since these prefixes are small enough, the size of the resulting source-wise spanners will not be too large.

Specifically, our path-reporting prioritized spanner with size $O(n(\log\log n)^2)$ is achieved by using a sequence of prefixes of sizes $n^{\frac{1}{2}},n^{\frac{3}{4}},n^{\frac{7}{8}},...,n^{1-\frac{1}{2^i}}$. Recall that the size of an $A$-source-wise spanner with stretch $O(k)$ is $O(|A|^{1+\frac{1}{k}}\cdot\beta+n\log\log n)$, where $\beta=n^{o(1)}$. Note that for $A=A_i$, by choosing $k=2^i-1$, the first term is $|A_i|^{1+\frac{1}{k}}\leq n^{\frac{2^i-1}{2^i}\cdot\frac{2^i}{2^i-1}}=n$. Thus, we must choose $k$ a bit larger than $2^i-1$, in order to obtain an $A_i$-source-wise spanner with size $O(n\log\log n)$ (this is the sparsest $A_i$-source-wise spanner we can obtain, since our pairwise/subset/source-wise spanners always have an additive term of $n\log\log n$ in their size). Namely, for cancelling the factor $\beta$, we must choose $k=\frac{2^i-1}{1-2^i\cdot o(1)}$. 

In the choice of $k$ above, the $o(1)$ factor in the denominator is actually $\frac{\log\beta}{\log n}$. This means that we cannot make this choice for prefixes $A_i$ with $i>\log\log n-\log\log\beta=(1-o(1))\log\log n$. For this reason, we cannot choose $T$ such that the prefixes $A_1,...,A_T$ will cover the entire priority ranking $(v_1,...,v_n)$, or most of it (to cover $\frac{n}{2}$ vertices, we need $T=\log\log n$). Instead, we choose $T$ to be roughly $\log\log n-\log\log\beta$. The size of the resulting prioritized spanner is $T\cdot O(n\log\log n)=O(n(\log\log n)^2)$, and the covered vertices are the vertices $v_j$ such that $j\leq\frac{n}{\beta}$. For queries that include a covered vertex $v_j$, the stretch is roughly $O(2^i)$, where $i$ is the minimal integer such that $j\leq n^{1-\frac{1}{2^i}}$. That is, the stretch is approximately $O\left(\frac{\log n}{\log n-\log j}\right)$.

For uncovered vertices (i.e., vertices $v_j$ such that $j>\frac{n}{\beta}$), we simply add a (non-prioritized) spanner, that has stretch $O\left(\frac{\log n}{\log\log\log n}\right)$ for \textit{all} the vertices of $G$. The size of such spanner may be as low as $O(n(\log\log n)^2)$, using results from \cite{ES23}. This spanner, as well as all the source-wise spanners we used, are path-reporting, therefore the resulting prioritized spanner is also path-reporting.

We provide another variation of a path-reporting prioritized spanner with reduced stretch, at the cost of increasing the size to $O(n\log n)$. To achieve that, we change the sequence of prefixes we use, to a sequence of prefixes with sizes $n^{\frac{1}{2}},n^{\frac{2}{3}},n^{\frac{3}{4}},...,n^{1-\frac{1}{i}}$. This sequence grows slower than the previous one. While increasing the size of the resulting prioritized spanner, this enable us to control the stretch of each source-wise spanner more carefully. %We note, however, that the improvement in the stretch is only by a multiplicative factor of $2$, and only for high-ranked vertices.

\subsection{Organization}

After some preliminaries in Section \ref{sec:Preliminaries}, we prove our new lower bounds for pairwise spanners in Section \ref{sec:LowerBounds} (lower bounds for pairwise spanners with stretch $1+\epsilon$ are in Section \ref{sec:NearExactLowerBound}, while Section \ref{sec:HighStretchLowerBound} is for higher stretch). For the high stretch regime, we prove our new upper bounds for pairwise spanners in Section \ref{sec:UpperBounds}. Lastly, in Section \ref{sec:PRDOTypes} we show the reductions between pairwise, subset, source-wise and prioritized spanners, which result in new upper bounds for these types of spanners. Appendix \ref{sec:ProofExtendedVersion} provides missing proofs for Section \ref{sec:UpperBounds}.

\section{Preliminaries} \label{sec:Preliminaries}

Given an undirected weighted graph $G=(V,E)$, we denote by $d_G(u,v)$ the distance between the two vertices $u,v\in V$. When the graph $G$ is clear from the context, we sometimes omit the subscript $G$ and write $d(u,v)$. %Given a subset $S\subseteq E$, we denote by $d_S(u,v)$ the distance between $u,v$ in the subgraph of $G$ that is induced by the edges of $S$.

When the given graph is weighted, $w(e)>0$ denotes the weight of the edge $e$. For every set of edges $P$ (e.g., a tree or a path), we denote $w(P)=\sum_{e\in P}w(e)$. If $P$ is a path, we denote by $|P|$ the length of $P$, i.e., the number of its edges.

\subsection{Path-Reporting Pairwise Spanners} \label{sec:PairwiseSpannersPre}

In this paper we provide pairwise spanners with the additional property of being \textit{path-reporting}.

\begin{definition}[Path-Reporting Pairwise Spanner] \label{def:PathReportingPairwiseSpanner}
Let $G=(V,E)$ be an undirected weighted graph, and let $\mathcal{P}\subseteq V^2$ be a set of pairs of vertices in $G$. A {\em path-reporting $\mathcal{P}$-pairwise $\alpha$-spanner} for $G$, with query time $q$, is a pair $(S,D)$, where $S$ is a subgraph of $G$, and $D$ is an oracle (a data structure), that given a pair $(u,v)\in\mathcal{P}$, returns a $u-v$ path $P$ in $S$, within time $O(q+|P|)$, such that
\[w(P)\leq\alpha d_G(u,v)~.\]
The {\em size} of a path-reporting pairwise spanner $(S,D)$ is defined as $\max\{|E_S|,|D|\}$. Here, $E_S$ is the set of edges in $S$, and $|D|$ is the storage size of the oracle $D$, measured in \textit{words}.

The subgraph $S$ is called a {\em pairwise spanner} of $G$. A {\em path-reporting preserver} is a path-reporting pairwise $1$-spanner.
\end{definition}

\begin{remark} \label{remark:QueryTimeOfPathReportingStructures}
Note that in path-reporting pairwise spanners (and also in path-reporting emulators, which will be defined later), a query that outputs a path $P$ always has query time $\Omega(|P|)$. This is inevitable, because a part of the task of a query is to report the path $P$ itself. A standard approach in the area of path-reporting structures (see \cite{EP15,NS23,ES23}) is to measure, instead of the whole query time, only the \textit{additional} query time besides reporting the path. More specifically, if the running time of each query of a path-reporting spanner/emulator is $O(q+|P|)$, where $P$ is the output path of the query, we say that this path-reporting spanner/emulator has query time $O(q)$.
\end{remark}

%Note that a path-reporting pairwise spanner is generally a stronger notion than a classic pairwise spanner. That is, if there is a path-reporting pairwise spanner for $G,\mathcal{P}$ with certain properties, then there is also a pairwise spanner for $G,\mathcal{P}$ with the same properties. Path-reporting pairwise spanners are also stronger than the following notion of a \textit{pairwise PRDO}.

In \cite{ES23}, two path-reporting pairwise spanners were presented, one of them with stretch $1+\epsilon$ (a \textit{near-exact} pairwise spanner), and one of them with stretch $3+\epsilon$. We state these results again here, so we can use them later in Section \ref{sec:PRDOTypes}.

\begin{theorem}[\cite{ES23}] \label{thm:NearExactPairwiseSpanner}
Given an undirected weighted $n$-vertex graph $G=(V,E)$, an integer $k\in[3,\log_2n]$, a positive parameter $\epsilon\leq1$, and a set of pairs $\mathcal{P}\subseteq V^2$, there is a $\mathcal{P}$-pairwise $(1+\epsilon)$-spanner for $G$ with size 
\[O(|\mathcal{P}|\cdot\beta_1+n\cdot(\log k+\log\log(\frac{1}{\epsilon}))+n^{1+\frac{1}{k}})~,\]
where $\beta_1=\beta_1(\epsilon,k)=O(\frac{\log k}{\epsilon})^{(1+o(1))\log_2k+\log_{4/3}\log(\frac{1}{\epsilon})}$.
\end{theorem}

\begin{theorem}[\cite{ES23}] \label{thm:Near3PairwiseSpanner}
Given an undirected weighted $n$-vertex graph $G=(V,E)$, an integer $k\geq3$, a positive parameter $\epsilon\in(0,40]$, and a set of pairs $\mathcal{P}\subseteq V^2$, there is a $\mathcal{P}$-pairwise $(3+\epsilon)$-spanner for $G$ with size 
\[O(|\mathcal{P}|\cdot\beta_2+n\log k+n^{1+\frac{1}{k}})~,\]
where $\beta_2=\beta_2(\epsilon,k)=k^{\log_{4/3}(12+\frac{40}{\epsilon})}$.
\end{theorem}

\subsection{Other Types of Path-Reporting Spanners} \label{sec:PRDOTypesIntro}

Besides pairwise spanners, there are several special types of spanners, which we define below.

\begin{definition}
Let $G=(V,E)$ be an undirected weighted graph, and let $A\subseteq V$ be a subset of its vertices.
\begin{itemize}
    \item A {\em path-reporting $\alpha$-spanner} for $G$ is a path-reporting $\mathcal{P}$-pairwise $\alpha$-spanner, where $\mathcal{P}=V^2$.
    \item A {\em path-reporting $A$-subset $\alpha$-spanner} for $G$ is a path-reporting $\mathcal{P}$-pairwise $\alpha$-spanner, where $\mathcal{P}=A\times A$.
    \item A {\em path-reporting $A$-source-wise $\alpha$-spanner} for $G$ is a path-reporting $\mathcal{P}$-pairwise $\alpha$-spanner, where $\mathcal{P}=A\times V$.
\end{itemize}
The subgraph $S$ is called a {\em spanner}, a {\em subset spanner}, or a {\em source-wise spanner}, respectively.
\end{definition}

\begin{comment}
\begin{definition}[Path-Reporting Spanner] \label{def:PathReportingSpanner}
Let $G=(V,E)$ be an undirected weighted graph. A {\em path-reporting $\alpha$-spanner} for $G$ is a path-reporting $\mathcal{P}$-pairwise $\alpha$-spanner $(S,D)$, where $\mathcal{P}=V^2$. The subgraph $S$ is called a {\em spanner} of $G$.
\end{definition}

\begin{definition}[Path-Reporting Subset Spanner] \label{def:PathReportingSubsetSpanner}
Let $G=(V,E)$ be an undirected weighted graph, and let $A\subseteq V$ be a subset of its vertices. A {\em path-reporting $A$-subset $\alpha$-spanner} for $G$ is a path-reporting $\mathcal{P}$-pairwise $\alpha$-spanner $(S,D)$, where $\mathcal{P}=A\times A$. The subgraph $S$ is called a {\em subset spanner} of $G$.
\end{definition}

\begin{definition}[Path-Reporting Source-wise Spanner] \label{def:PathReportingSourcewiseSpanner}
Let $G=(V,E)$ be an undirected weighted graph, and let $A\subseteq V$ be a subset of its vertices. A {\em path-reporting $A$-source-wise $\alpha$-spanner} for $G$ is a path-reporting $\mathcal{P}$-pairwise $\alpha$-spanner $(S,D)$, where $\mathcal{P}=A\times V$. The subgraph $S$ is called a {\em source-wise spanner} of $G$.
\end{definition}
\end{comment}

\begin{definition}[Path-Reporting Prioritized Spanner] \label{def:PrioritizedSpanner}
Let $G=(V,E)$ be an undirected weighted graph, and fix a permutation $(v_1,v_2,...,v_n)$ of $V$. This permutation is called a {\em priority ranking} of $V$. A {\em path-reporting prioritized spanner} for $G$, with prioritized stretch $\alpha:\{1,2,...,n\}\rightarrow\mathbb{R}_{\geq0}$ and prioritized query time $q:\{1,2,...,n\}\rightarrow\mathbb{R}_{\geq0}$, is a pair $(S,D)$, where $S$ is a subgraph of $G$, and $D$ is an oracle that given a query $(v_j,v_{j'})$ such that $j<j'$, returns a $v_j-v_{j'}$ path $P$ in $S$ with weight at most $\alpha(j)\cdot d_G(v_j,v_{j'})$, within time at most $O(q(j)+|P|)$. The size of $(S,D)$ is defined as $\max\{|E_S|,|D|\}$, where $E_S$ is the set of edges in $S$, and $|D|$ is the space required to store $D$. The subgraph $S$ is called a {\em prioritized spanner}.
\end{definition}

A result of \cite{ES23} was the construction of several path-reporting spanners (referred as \textit{interactive spanners} in \cite{ES23}) with smaller size, stretch and query time than what was known before. We state here two of these results.

\begin{theorem}[Theorem 6 in \cite{ES23}] \label{thm:ES6}
For every undirected weighted $n$-vertex graph and parameters $k,\epsilon$ such that $3\leq k\leq\log n$ and $\sqrt{\frac{\log k\cdot\log\log k}{k}}\leq\epsilon\leq\frac{1}{2}$, there is a path-reporting spanner with stretch $(4+\epsilon)k$, query time $O(\log k)$ and size
\[O\left(\left\lceil\frac{k\cdot\log\log n\cdot\log\log\log n}{\epsilon\log n}\right\rceil\cdot n^{1+\frac{1}{k}}\right)~.\]
\end{theorem}

\begin{theorem}[\cite{ES23}] \label{thm:ES7}
For every undirected weighted $n$-vertex graph and parameter $k$ such that $3\leq k\leq\log n$, there is a path-reporting spanner with stretch $O(k)$, query time $O(\log\left\lceil\frac{k\log\log n}{\log n}\right\rceil)$ and size
\[O\left(n\log k+\left\lceil\frac{k\cdot\log\log n}{\log n}\right\rceil\cdot n^{1+\frac{1}{k}}\right)~.\]
\end{theorem}

\subsection{Path-Reporting Emulators} \label{sec:EmulatosPre}

A key component of our constructions in Section \ref{sec:PRDOTypes} is the use of \textit{path-reporting emulators}.

\begin{definition} \label{def:PathReportingEmulator}
Let $G=(V,E)$ be an undirected weighted graph. A {\em path-reporting $\alpha$-emulator} for $G$ is a pair $(Q,D)$, where $Q=(V,E')$ is a weighted graph with weights $w(x,y)=d_G(x,y)$, and $D$ is an oracle (a data structure), that given $u,v\in V$, returns a $u-v$ path $P$ in $Q$, such that
\[w(P)\leq\alpha d_G(u,v)~.\]
The graph $Q$ is called an $\alpha$-{\em emulator}. The parameter $\alpha$ is called the {\em stretch} of the path-reporting emulator $(Q,D)$. We say that the {\em query time} of $(Q,D)$ is $q$ if the running time of the oracle $D$ on a single query $(u,v)\in V^2$ is $O(q+|P|)$, where $P$ is the output path. The {\em size} of $(Q,D)$ is defined as $\max\{|E'|,|D|\}$, where $|D|$ is the storage size of the oracle $D$, measured in \textit{words}.
\end{definition}

The distance oracle of from \cite{TZ01} can be viewed as a path-reporting emulator. Considering the improvement to its query time by \cite{WN13}, it can be summarized in the following theorem.

\begin{theorem}[\cite{TZ01,WN13}] \label{thm:TZEmulator}
For every undirected weighted $n$-vertex graph $G=(V,E)$ and integer $k\geq1$, there is a path-reporting $(2k-1)$-emulator for $G$ with query time $O(\log k)$ and size $O(kn^{1+\frac{1}{k}})$.
\end{theorem}

In \cite{ES23}, the authors introduced another path-reporting emulator, that had reduced size and query time, with respect to that of Theorem \ref{thm:TZEmulator}, while its stretch was larger by a constant factor. This emulator is based on the distance oracle of \cite{MN06}.

\begin{theorem}[\cite{ES23}] \label{thm:MNEmulator}
For every undirected weighted $n$-vertex graph $G=(V,E)$ and integer $k\geq1$, there is a path-reporting $O(k)$-emulator for $G$ with query time $O(1)$ and size $O(n^{1+\frac{1}{k}})$.
\end{theorem}

\begin{comment}
\subsection{Hopsets}

A key notion in Section \ref{sec:UpperBounds} is that of a \textit{hopset}.

\begin{definition}[Hopset] \label{def:Hopset}
Let $G=(V,E)$ be an undirected weighted graph. A {\em hopset with stretch $\alpha$ and hopbound $\beta$} for $G$, shortly an {\em $(\alpha,\beta)$-hopset}, is a set $H\subseteq\binom{V}{2}$ such that for every pair of vertices $u,v\in V$,
\[d^{(\beta)}_{G\cup H}(u,v)\leq\alpha\cdot d_G(u,v)~.\]
Here, $G\cup H$ denotes the weighted graph that is obtained by adding every edge $(x,y)\in H$ to $G$, with weight $d_G(x,y)$. The notation $d^{(\beta)}_{G\cup H}(u,v)$ stands for $\beta$\textit{-bounded} $u-v$ distance in $G\cup H$, i.e., the weight of the shortest $u-v$ path in $G\cup H$ with at most $\beta$ edges.
\end{definition}

\end{comment}

\section{Lower Bounds for Pairwise Spanners} \label{sec:LowerBounds}

\subsection{Near-exact Pairwise Spanners} \label{sec:NearExactLowerBound}

To prove our lower bound for pairwise $(1+\epsilon)$-spanners, we use almost the same construction of a graph as the one appears in \cite{ABP18}. Our argument that achieves this lower bound, however, is somewhat different. Before we go into the specific details of the construction, we overview the properties of the graph of \cite{ABP18}. 

The authors of \cite{ABP18} constructed a sequence of graphs $\{\mathcal{H}_\kappa\}_{\kappa=0}^\infty$, each with a layered structure. The first and last layers of $\mathcal{H}_\kappa$ serves as \textit{input} and \textit{output} ports (respectively), while the interior layers are made out of many copies of $\mathcal{H}_{\kappa-1}$. Then a relatively large set $\mathcal{P}_\kappa$ of input-output pairs is defined, such that for every $(u,v)\in\mathcal{P}_\kappa$ there is a unique shortest path in $\mathcal{H}_\kappa$ between $u,v$, that passes through each layer exactly once. The edges between the layers of $\mathcal{H}_\kappa$ are heavy enough, such that any other path from $u$ to $v$ suffers a large stretch, since it must visit at least two layers more than once (in the unweighted version, the edges between the layers are replaced with long paths).

We now construct a sequence $\{\mathcal{H}_\kappa\}$ with the same properties. The construction is essentially the same as in \cite{ABP18}. However, we fully describe it in details here, because (1) there are slight differences from the original construction, and (2) our lower bound proof refers to the specific details of this graph and the way it was constructed. We do use the \textit{second base graph} from Section 2.2 in \cite{ABP18} as it is (this graph is actually originated in \cite{AB17}). 

The second base graph is an undirected unweighted graph, denoted by $\ddot{B}[p,l]$, where $p,l>0$ are two positive integer parameters. The vertices of this graph are organized in $2l+1$ layers, each of them with size $p$, such that all the edges of the graph are between adjacent layers. The vertices on the first layer of $\ddot{B}[p,l]$ are called \textit{input ports}, and the vertices on the last layer of $\ddot{B}[p,l]$ are called \textit{output ports}. Aside from the graph itself, there is a set of pairs of input-output ports $\ddot{\mathcal{P}}[p,l]$, such that every pair in this set is connected in $\ddot{B}[p,l]$ by a unique shortest path. The size of $\ddot{\mathcal{P}}[p,l]$ is $p^{2-o(1)}$, thus it contains a large portion of all the $p^2$ possible input-output pairs. Lastly, the graph $\ddot{B}[p,l]$ has labels on its edges, such that the edges on the unique shortest path of every input-output pair in $\ddot{\mathcal{P}}[p,l]$ are labeled alternately by two labels. One can think of these labels as routing directions to get from an input to an output. %This is thanks to the interesting property of these labels, that shortest paths that shares a vertex and have the same pair of labels, are the same path.

Formally, the graph $\ddot{B}[p,l]$ is described in the following lemma. The proof of this lemma is implicit in \cite{ABP18}, in which the construction of the graph $\ddot{B}[p,l]$ is described.

\begin{lemma} \label{lemma:SecondBaseGraph}
Let $p,l>0$ be two integer parameters. There is a function $\xi(p,l)=2^{O(\sqrt{\log p\log l})}$ which is non-decreasing in the parameter $p$, and there is an undirected unweighted graph $\ddot{B}[p,l]$, with the following properties. 
\begin{enumerate}
    \item The vertices of the graph $\ddot{B}[p,l]$ are partitioned into $2l+1$ disjoint layers $L_0,L_1,...,L_{2l}$, each of them of size $p$, such that every edge of $\ddot{B}[p,l]$ is between vertices that belongs to adjacent layers.
    \item There is a set of edge-labels $\ddot{\mathcal{L}}[p,l]$ of size $|\ddot{\mathcal{L}}[p,l]|\in\left[\frac{\sqrt{p}}{\xi(\sqrt{p},l)},\frac{\sqrt{p}}{2}\right]$, such that for every $i<2l$, every vertex $x\in L_i$, and every label $a\in\ddot{\mathcal{L}}[p,l]$, there is exactly one edge from $x$ to a vertex $y\in L_{i+1}$, that is labeled by $a$ (the vertex $y$ is different for every label $a\in\ddot{\mathcal{L}}[p,l]$).
    \item Given a vertex $u\in L_0$ and a pair of labels $(a,b)\in(\ddot{\mathcal{L}}[p,l])^2$, let $P_{u,v}$ be the path that starts at $u$, and its edges are labeled alternately by the labels $a,b$ (starting by $a$). Here, $v\in L_{2l}$ is the vertex in which the path $P_{u,v}$ ends, and we denote $v=out(u,a,b)$. Then, $P_{u,v}$ is the unique shortest path in $\ddot{B}[p,l]$. Moreover, for any other $u'\in L_0\setminus\{u\}$, the path $P_{u',v'}$, that is alternately labeled by the same labels $a,b$ and ends in $v'=out(u',a,b)\in L_{2l}$, is vertex-disjoint from $P_{u,v}$.
%    \item There is a set $\ddot{\mathcal{P}}[p,l]\subseteq L_0\times L_{2l}$ of size $|\ddot{\mathcal{P}}[p,l]|\geq\frac{p^2}{\xi(\sqrt{p},l)^2}$, such that for every $(u,v)\in\ddot{\mathcal{P}}[p,l]$ there is a unique shortest path $P_{u,v}$ in $\ddot{B}[p,l]$ with length $2l$.
%    \item There is a set of edge-labels $\ddot{\mathcal{L}}[p,l]$ of size $|\ddot{\mathcal{L}}[p,l]|\in\left[\frac{\sqrt{p}}{\xi(\sqrt{p},l)},\frac{\sqrt{p}}{2}\right]$, such that every edge in $\ddot{B}[p,l]$ is labeled by one of the labels in $\ddot{\mathcal{L}}[p,l]$, and for every $(u,v)\in\ddot{\mathcal{P}}[p,l]$, the edges of $P_{u,v}$ are alternately labeled by two labels $a,b\in\ddot{\mathcal{L}}[p,l]$. Moreover, if $(u,v)\neq(u',v')\in\ddot{\mathcal{P}}[p,l]$, and $a,b$ are the labels that corresponds to $P_{u,v}$, and $a',b'$ are the labels that corresponds to $P_{u',v'}$, then either $P_{u,v},P_{u',v'}$ are vertex-disjoint, or $(a,b)\neq(a',b')$.
\end{enumerate}
\end{lemma}

We define the set $\ddot{\mathcal{P}}[p,l]\subseteq L_0\times L_{2l}$ as
\begin{equation} \label{eq:SecondBasePairs}
\ddot{\mathcal{P}}[p,l]=\{(u,out(u,a,b))\;|\;u\in L_0,\;a,b\in\ddot{\mathcal{L}}[p,l]\}~.
\end{equation}
The size of this set is $p\cdot|\ddot{\mathcal{L}}[p,l]|^2$, which is at least $p\cdot\left(\frac{\sqrt{p}}{\xi(\sqrt{p},l)}\right)^2=\frac{p^2}{\xi(\sqrt{p},l)^2}$ and at most $p\cdot\left(\frac{\sqrt{p}}{2}\right)^2=\frac{p^2}{4}$, by Lemma \ref{lemma:SecondBaseGraph}.

The specific details of the construction of the graph $\ddot{B}[p,l]$ appear in \cite{ABP18}. We, however, only use the properties that are described in Lemma \ref{lemma:SecondBaseGraph}, and do not need these details for our proof. We now define the sequence of graphs $\{\mathcal{H}_\kappa[p,l]\}_{\kappa=0}^\infty$ recursively, where $p,l>0$ are any two integer parameters. For every $\kappa$, we also define a corresponding set $\mathcal{P}_\kappa[p,l]$ of pairs of vertices from $\mathcal{H}_\kappa[p,l]$.

The graph $\mathcal{H}_0[p,l]$ is defined to be the complete bipartite graph $K_{p,p}$. The corresponding set of pairs $\mathcal{P}_\kappa[p,l]$ is defined to be all the pairs $(u,v)$, of a vertex $u$ from the left side of $\mathcal{H}_0[p,l]=K_{p,p}$ and a vertex $v$ from its right side.

To construct $\mathcal{H}_\kappa[p,l]$ for $\kappa>0$, we start with the graph $\ddot{B}[p,l]$ from Lemma \ref{lemma:SecondBaseGraph}. The vertices of this graph are partitioned into layers $L_0,L_1,...,L_{2l}$, where edges only exist in between adjacent layers. We call the vertices in the first layer $L_0$ \textit{input ports}, and the vertices in the last layer $L_{2l}$ \textit{output ports}. The rest of the vertices of $\ddot{B}[p,l]$ are called \textit{internal vertices}. The input and output ports also serves as the input/output ports of the graph $\mathcal{H}_\kappa[p,l]$ (in particular, there are $p$ input ports and $p$ output ports in $\mathcal{H}_\kappa[p,l]$). Let $\ddot{\mathcal{L}}[p,l]$ be the set of edge-labels, as described in Lemma \ref{lemma:SecondBaseGraph}, and let $\ddot{\mathcal{P}}[p,l]$ be the corresponding set of pairs from Definition \ref{eq:SecondBasePairs}. Let $p'=|\ddot{\mathcal{L}}[p,l]|$. We fix an arbitrary bijection $\pi:\ddot{\mathcal{L}}[p,l]\rightarrow\{1,2,...,p'\}$.

Consider the graph $\mathcal{H}_{\kappa-1}[p',l]$. Using $\pi$, we match each input/output port of this graph to a label in $\ddot{\mathcal{L}}[p,l]$. We replace each internal vertex of $\ddot{B}[p,l]$ by a copy of the graph $\mathcal{H}_{\kappa-1}[p',l]$. For a vertex $u$ in $\ddot{B}[p,l]$, denote this copy by $\mathcal{H}_{\kappa-1}^u[p',l]$. Let $(u,v)$ be an edge in $\ddot{B}[p,l]$ with label $a\in\ddot{\mathcal{L}}[p,l]$, such that $u\in L_i$, $v\in L_{i+1}$. In $\mathcal{H}_\kappa[p,l]$, we replace this edge by an edge of weight $(2l-1)^\kappa$ as follows.
\begin{itemize}
    \item If $i$ is even, the new edge is added from the $\pi(a)$-th {\em input} port of $\mathcal{H}_{\kappa-1}^u[p',l]$ (or, in case that $i=0$, from $u$ itself) to the $\pi(a)$-th {\em input} port of $\mathcal{H}_{\kappa-1}^v[p',l]$.
    \item If $i$ is odd, the new edge is added from the $\pi(a)$-th {\em output} port of $\mathcal{H}_{\kappa-1}^u[p',l]$ to the $\pi(a)$-th {\em output} port of $\mathcal{H}_{\kappa-1}^v[p',l]$ (or, in case that $i=2l-1$, to $v$ itself).
\end{itemize}

In other words, we can imagine that the copies of $\mathcal{H}_{\kappa-1}[p',l]$ are inserted as they are in odd layers, and \textit{reversed} in even layers. This way, input ports are connected to input ports, and output ports are connected to output ports. See Figure \ref{fig:NearExactLowerBound} for an illustration.

\begin{center}
\begin{figure}[ht!]
    \centering
    \includegraphics[width=14cm, height=8cm]{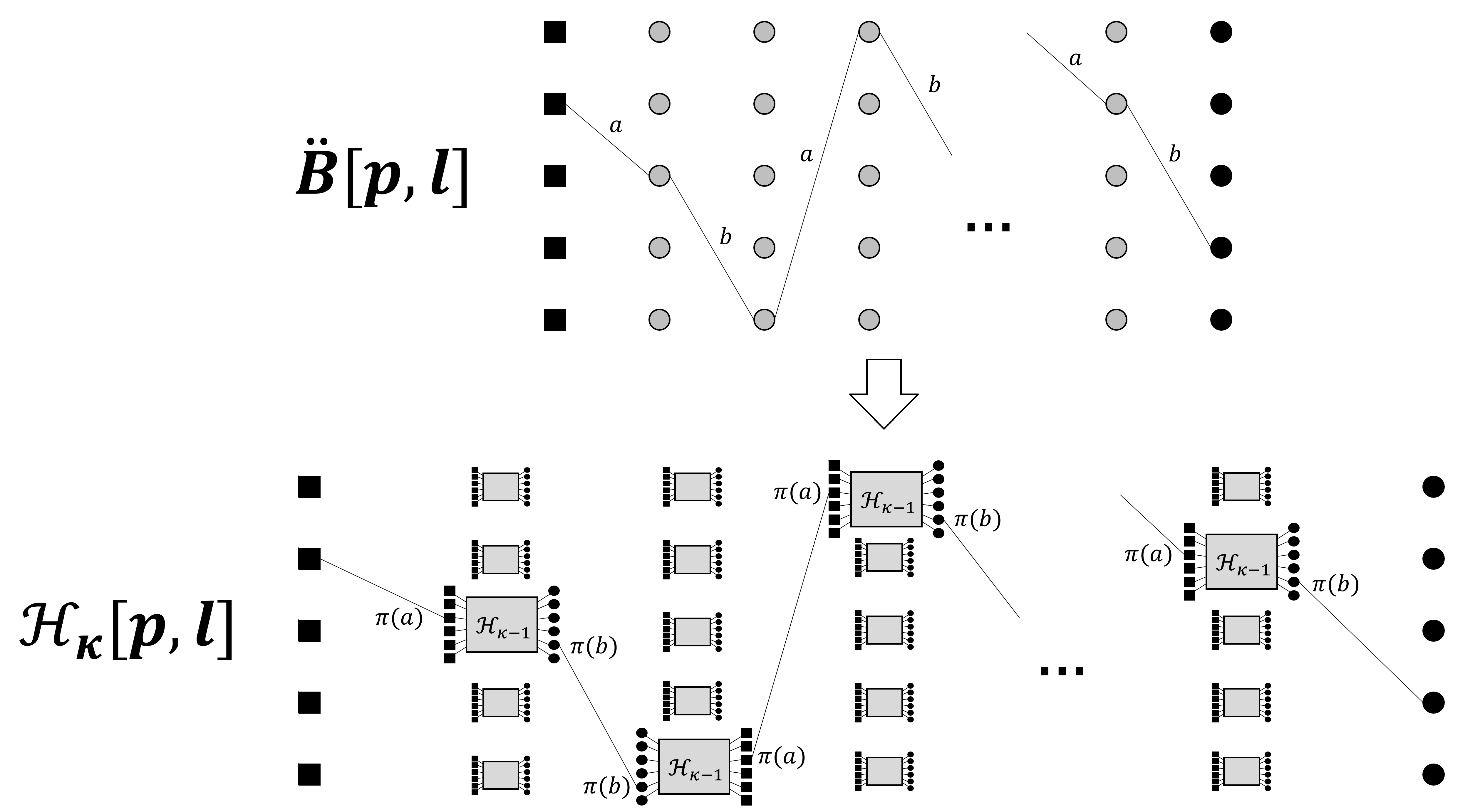}
    \caption{An illustration of the graph $\mathcal{H}_\kappa[p,l]$. The internal vertices of $\ddot{B}[p,l]$ are replaced by copies of $\mathcal{H}_{\kappa-1}[p',l]$, where $p'$ is the number of edges-labels in $\ddot{B}[p,l]$. An edge of $\ddot{B}[p,l]$ that had label $a$, and is from an even layer to an odd layer, is replaced by an edge that connects the $\pi(a)$-th \textit{input} ports of the corresponding copies of $\mathcal{H}_{\kappa-1}[p',l]$ (input port are represented in the figure by a square shape). If the edge is from an odd layer to an even layer, the same happens for the \textit{output} ports of these copies (represented by circular shape).}
    \label{fig:NearExactLowerBound}
\end{figure}
\end{center}

This completes the description of the graph $\mathcal{H}_\kappa[p,l]$. We define the corresponding set $\mathcal{P}_\kappa[p,l]$ as follows. Given some $(u,v)\in\ddot{\mathcal{P}}[p,l]$, let $(a,b)\in(\ddot{\mathcal{L}}[p,l])^2$ be the unique pair of labels such that the $u-v$ shortest path in $\ddot{B}[p,l]$ is labeled alternately with $a,b$. Henceforth, we call $(a,b)$ the {\em corresponding labels} to $(u,v)$. Denote by $u'$ the $\pi(a)$-th input port of $\mathcal{H}_{\kappa-1}[p',l]$, and by $v'$ the $\pi(b)$-th output port of $\mathcal{H}_{\kappa-1}[p',l]$. We say that $(u,v)\in\mathcal{P}_\kappa[p,l]$ if and only if $(u',v')\in\mathcal{P}_{\kappa-1}[p',l]$.

The following lemma is our version of Lemma 2.2 from \cite{ABP18}. We prove it here since our construction of $\mathcal{H}_\kappa[p,l],\mathcal{P}_\kappa[p,l]$ is slightly different.

\begin{lemma} \label{lemma:NumberOfPairs}

\[|\mathcal{P}_\kappa[p,l]|\geq\frac{p^2}{\xi(\sqrt{p},l)^{2\kappa}}~.\]
\end{lemma}

\begin{proof}

We prove the lemma by induction on $\kappa$. For $\kappa=0$, by definition 
\[|\mathcal{P}_0[p,l]|=p^2=\frac{p^2}{\xi(\sqrt{p},l)^{2\cdot0}}~.\] 

Fix $\kappa>0$, and fix an input port $u$ of $\mathcal{H}_\kappa[p,l]$. Denote by $A_u$ the set of pairs $(u,v)\in\mathcal{P}_\kappa[p,l]$, i.e., the set of pairs in $\mathcal{P}_\kappa[p,l]$ with $u$ as their input port. We show a bijection between the sets $A_u$ and $\mathcal{P}_{\kappa-1}[p',l]$. Let $(u,v)$ be a pair in $A_u$, and let $(a,b)\in(\ddot{\mathcal{L}}[p,l])^2$ be the corresponding labels to $(u,v)$. By the definition of $\mathcal{P}_\kappa[p,l]$, the $\pi(a)$-th input port $u'$ and the $\pi(b)$-th output port $v'$ of $\mathcal{H}_\kappa[p,l]$ satisfy $(u',v')\in\mathcal{P}_{\kappa-1}[p',l]$. Thus, we map the pair $(u,v)\in A_u$ to the pair $(u',v')\in\mathcal{P}_{\kappa-1}[p',l]$.

To prove that this mapping is a bijection, we now show the inverse mapping. Note that for any $(u',v')\in\mathcal{P}_{\kappa-1}[p',l]$, there are unique labels $a,b\in\ddot{\mathcal{L}}[p,l]$ such that $u'$ is the $\pi(a)$-th input port and $v'$ is the $\pi(b)$-th output port of $\mathcal{H}_{\kappa-1}[p',l]$. This is true since $\pi$ is a bijection. Let $v$ be the output port of $\mathcal{H}_\kappa[p,l]$, such that $v=out(u,a,b)$ (using the notation $out()$ from Lemma \ref{lemma:SecondBaseGraph}). Then, $(a,b)$ are the corresponding labels to $(u,v)$, and since $(u',v')\in\mathcal{P}_{\kappa-1}[p',l]$, we conclude that $(u,v)\in\mathcal{P}_\kappa[p,l]$. That is, $(u,v)\in A_u$.

The two mappings that were described in the two paragraphs above are the inverse of each other, hence our mapping is a bijection, and we conclude that $|A_u|=|\mathcal{P}_{\kappa-1}[p',l]|$. Summing over all the input ports $u$ of $\mathcal{H}_\kappa[p,l]$, and using the induction hypothesis, we get
\[|\mathcal{P}_\kappa[p,l]|=\sum_u|A_u|=\sum_u|\mathcal{P}_{\kappa-1}[p',l]|\geq p\cdot\frac{(p')^2}{\xi(\sqrt{p'},l)^{2\kappa-2}}~.\]
By Lemma \ref{lemma:SecondBaseGraph}, we know that $p'=|\ddot{\mathcal{L}}[p,l]|\geq\frac{\sqrt{p}}{\xi(\sqrt{p},l)}$, and that $\xi(\sqrt{p'},l)\leq\xi(\sqrt{p},l)$, since $\xi$ is a non-decreasing function in the first variable (and $p'\leq\frac{\sqrt{p}}{2}<p$). Hence,
\[|\mathcal{P}_\kappa[p,l]|\geq p\cdot\frac{(p')^2}{\xi(\sqrt{p'},l)^{2\kappa-2}}\geq p\cdot\frac{p}{\xi(\sqrt{p},l)^2\xi(\sqrt{p'},l)^{2\kappa-2}}\geq\frac{p^2}{\xi(\sqrt{p},l)^{2\kappa}}~.\]

%\alert{I wasn't convinced by the above paragraph, it may depend on how $\pi$ is chosen exactly (I think you need it to be a bijection, so there are dependencies). Also, what if $a=b$?}

%\alert{Idan: I changed Lemma \ref{lemma:SecondBaseGraph}. Apparently, the properties I wrote there were not enough to prove this current lemma. I rewrote a few things in this subsection accordingly.}

\end{proof}

Next, we estimate the number of vertices in $\mathcal{H}_\kappa[p,l]$. Denote this number by $n_\kappa[p,l]$.

\begin{lemma} \label{lemma:NumberOfVertices}

For every $\kappa\geq0$, 
\[\frac{2(2l-1)^\kappa}{\xi(\sqrt{p},l)^{2\kappa}}p^{2-\frac{1}{2^\kappa}}\leq n_\kappa[p,l]\leq2(2l)^\kappa p^{2-\frac{1}{2^\kappa}}~.\]

\end{lemma}

\begin{proof}

We again prove the lemma by induction on $\kappa$. For $\kappa=0$, the graph $\mathcal{H}_0[p,l]$ is the complete bipartite graph $K_{p,p}$. Therefore, 
\[\frac{2(2l-1)^0}{\xi(\sqrt{p},l)^{2\cdot0}}p^{2-\frac{1}{2^0}}=2p=n_0[p,l]=2p=2(2l)^0p^{2-\frac{1}{2^0}}~.\]

For $\kappa>0$, recall that $\mathcal{H}_\kappa[p,l]$ was obtained by replacing each of the $(2l-1)p$ vertices in the interior layers of $\ddot{B}[p,l]$, by a copy of $\mathcal{H}_{\kappa-1}[p',l]$. The number of vertices in any of these copies is
\[n_{\kappa-1}[p',l]\leq2(2l)^{\kappa-1}(p')^{2-\frac{1}{2^{\kappa-1}}}\leq2(2l)^{\kappa-1}\left(\frac{\sqrt{p}}{2}\right)^{2-\frac{2}{2^\kappa}}<2(2l)^{\kappa-1}p^{1-\frac{1}{2^\kappa}}~.\]
On the other hand, this number is also
\[n_{\kappa-1}[p',l]\geq\frac{2(2l-1)^{\kappa-1}}{\xi(\sqrt{p'},l)^{2\kappa-2}}(p')^{2-\frac{1}{2^{\kappa-1}}}\geq\frac{2(2l-1)^{\kappa-1}}{\xi(\sqrt{p'},l)^{2\kappa-2}}\left(\frac{\sqrt{p}}{\xi(\sqrt{p},l)}\right)^{2-\frac{2}{2^\kappa}}\geq\frac{2(2l-1)^{\kappa-1}}{\xi(\sqrt{p},l)^{2\kappa}}p^{1-\frac{1}{2^\kappa}}~.\]
In these bounds we used the induction hypothesis, the fact that $\xi$ is a non-decreasing function, and the bounds on $p'$ from Lemma \ref{lemma:SecondBaseGraph}.

We conclude that the number of vertices in $\mathcal{H}_\kappa[p,l]$ is
\begin{eqnarray*}
n_\kappa[p,l]&\leq&2p+(2l-1)p\cdot2(2l)^{\kappa-1}p^{1-\frac{1}{2^\kappa}}\\
&=&2p-2(2l)^{\kappa-1}p^{2-\frac{1}{2^\kappa}}+2(2l)^\kappa p^{2-\frac{1}{2^\kappa}}\\
&\stackrel{\kappa\geq1}{\leq}&2p-2\cdot1\cdot p^{3/2}+2(2l)^\kappa p^{2-\frac{1}{2^\kappa}}\leq2(2l)^\kappa p^{2-\frac{1}{2^\kappa}}~,
\end{eqnarray*}
and also
\[n_\kappa[p,l]\geq2p+(2l-1)p\cdot\frac{2(2l-1)^{\kappa-1}}{\xi(\sqrt{p},l)^{2\kappa}}p^{1-\frac{1}{2^\kappa}}\geq\frac{2(2l-1)^\kappa}{\xi(\sqrt{p},l)^{2\kappa}}p^{2-\frac{1}{2^\kappa}}~.\]
This completes the inductive proof.

\end{proof}

In \cite{ABP18}, the authors showed that any $(1+\epsilon,\beta)$-spanner for this graph\footnote{For the lower bound for near-additive spanners, one need to use an \textit{unweighted graph}. Thus, \cite{ABP18} actually used a similar graph where the edges between the copies of $\mathcal{H}_{\kappa-1}[p',l]$ are replaced by paths of length $(2l-1)^{\kappa-1}$. This graph has essentially the same properties as the graph $\mathcal{H}_\kappa[p,l]$ described here.}, that has less than $|\mathcal{P}_\kappa[p,l]|$ edges, must have $\beta=\Omega\left(\frac{1}{\epsilon\kappa}\right)^{\kappa-1}$. But note that by Lemma \ref{lemma:NumberOfVertices}, the number of vertices in $\mathcal{H}_\kappa[p,l]$ is $n=n_\kappa[p,l]\approx p^{2-\frac{1}{2^\kappa}}$, while Lemma \ref{lemma:NumberOfPairs} proves that the number of pairs in $\mathcal{P}_\kappa[p,l]$ is roughly
\[p^2\approx\left(n^{\frac{1}{2-\frac{1}{2^\kappa}}}\right)^2=n^{\frac{2^{\kappa+1}}{2^{\kappa+1}-1}}=n^{1+\frac{1}{2^{\kappa+1}-1}}~.\]
Thus, the result of \cite{ABP18} means that less than $n^{1+\frac{1}{2^{\kappa+1}-1}}$ edges in a near-additive spanner implies $\beta=\Omega\left(\frac{1}{\epsilon\kappa}\right)^{\kappa-1}$. %\alert{That last sentence isn't clear - does the bound come from \cite{ABP18} argument?} \alert{Idan: Yes, I added "the result of \cite{ABP18} means that ..."}

In our case, we will show that any $\mathcal{P}_\kappa[p,l]$-pairwise $(1+\epsilon)$-spanner for $\mathcal{H}_\kappa[p,l]$ must have at least $\beta|\mathcal{P}_\kappa[p,l]|$ edges, for $\beta=\Omega\left(\frac{1}{\epsilon\kappa^2}\right)^{\kappa}$. Otherwise, the stretch guarantee will not hold for at least one of the pairs in $\mathcal{P}_\kappa[p,l]$. %This means that any general construction of a pairwise $(1+\epsilon)$-spanner for an $n$-vertex graph $G$ 
%\alert{not for any graph $G$, for the specific $G={\cal H}_\kappa[p,l]$ we build here} \alert{Idan: Yes, my idea was to claim that no theorem of the form "for every graph and set of pairs there is a pairwise $(1+\epsilon)$-spanner with size ..." is possible. But this sentence is not so important anyway, so I commented it out.} %with at most $n^{1+\frac{1}{2^{\kappa+1}-1}}$ pairs $\mathcal{P}$, must have at least $\beta|\mathcal{P}|$ edges, for $\beta=\Omega\left(\frac{1}{\epsilon\kappa^2}\right)^\kappa$.

To achieve this goal, we prove some properties of the shortest paths that connect the pairs in $\mathcal{P}_\kappa[p,l]$. A key notion will be that of a \textit{critical} edge, which also appears in \cite{ABP18}.

\begin{definition} \label{def:CriticalEdge}
An edge $e$ of $\mathcal{H}_\kappa[p,l]$ is said to be {\em critical} if it lies in a copy of $\mathcal{H}_0[p,l]$. 

%\alert{What is originated? I'm not sure this is clear..} \alert{Idan: I replaced "is originated" with "lies"}
\end{definition}

The following lemma is parallel to Lemma 2.3 in \cite{ABP18}.

\begin{lemma} \label{lemma:ShortestPathsProperties}
The distance between any input and output port of $\mathcal{H}_\kappa[p,l]$ is at least $(2l\kappa+1)(2l-1)^\kappa$.

Moreover, for every pair $(u,v)\in\mathcal{P}_\kappa[p,l]$, there is a unique shortest path $P_{u,v}$ in $\mathcal{H}_\kappa[p,l]$ that connects $u,v$, and has weight $w(P_{u,v})=(2l\kappa+1)(2l-1)^\kappa$. This path does not share a critical edge with any other shortest path $P_{u',v'}$, for $(u',v')\in\mathcal{P}_\kappa[p,l]\setminus\{(u,v)\}$. That is, there are no critical edges in $P_{u,v}\cap P_{u',v'}$, for any pair $(u',v')\neq(u,v)$ in $\mathcal{P}_\kappa[p,l]$.
\end{lemma}

\begin{proof}

We prove the Lemma by induction over $\kappa\geq0$. For $\kappa=0$, the graph $\mathcal{H}_0[p,l]$ is the complete bipartite graph $K_{p,p}$. One of its sides consists of the input ports, and the other consists of the output ports. Thus, the distance between any input port and output port is at least $1=(2l\cdot0+1)(2l-1)^0$. For every $(u,v)\in\mathcal{P}_0[p,l]$, the edge $(u,v)$ is clearly the unique shortest path between $u,v$ that has weight $1$. %Any path other than the edge $(u,v)$ itself is of weight at least $3=1+2(2l-1)^0$, since the graph is bipartite. 
In addition, this path does not share its only critical edge $(u,v)$ with any other $u'-v'$ shortest path in $\mathcal{H}_0[p,l]=K_{p,p}$, for $(u',v')\in\mathcal{P}_0[p,l]$.

Fix $\kappa>0$. Every path that starts from an input port of $\mathcal{H}_\kappa[p,l]$ and ends in an output port must visit at least $2l-1$ copies of $\mathcal{H}_{\kappa-1}[p',l]$, each one of them replaces a vertex from a different layer of $\ddot{B}[p,l]$. By the induction hypothesis, passing through a copy of $\mathcal{H}_{\kappa}[p',l]$, from an input port to an output port (or vice versa), requires a path of weight at least $(2l(\kappa-1)+1)(2l-1)^{\kappa-1}$. The edges that connect the different copies are of weight $(2l-1)^\kappa$. Note that any path from the input layer of $\mathcal{H}_\kappa[p,l]$ to its output layer must pass through at least $2l$ of these edges. Overall, such path must have weight of at least
\[2l(2l-1)^\kappa+(2l-1)(2l(\kappa-1)+1)(2l-1)^{\kappa-1}=(2l+2l(\kappa-1)+1)(2l-1)^\kappa=(2l\kappa+1)(2l-1)^\kappa~.\]

Now fix a pair $(u,v)\in\mathcal{P}_\kappa[p,l]$. Recall that by definition, we know in particular that $u,v$ are vertices of $\ddot{B}[p,l]$, and $(u,v)\in\ddot{\mathcal{P}}[p,l]$. By Lemma \ref{lemma:SecondBaseGraph}, in the graph $\ddot{B}[p,l]$ there is a unique $u-v$ shortest path $P=(u=u_0,u_1,u_2,...,u_{2l}=v)$, labeled by some two labels $a,b\in\ddot{\mathcal{L}}[p,l]$. Denote by $H_1,H_2,...,H_{2l-1}$ the copies of $\mathcal{H}_{\kappa-1}[p',l]$ that replaced the vertices $u_1,u_2,...,u_{2l-1}$ in the construction of $\mathcal{H}_\kappa[p,l]$.

Recall that by its construction, the graph $\mathcal{H}_\kappa[p,l]$ contains the following edges. For every \textit{even} $i\in[0,2l-1]$, it contains an edge of weight $(2l-1)^\kappa$ from the $\pi(a)$-th input port of $H_i$ (or from $u$ itself if $i=0$) to the $\pi(a)$-th input port of $H_{i+1}$. For every \textit{odd} $i\in[0,2l-1]$, it contains an edge of weight $(2l-1)^\kappa$ from the $\pi(b)$-th output port of $H_i$ to the $\pi(b)$-th output port of $H_{i+1}$ (or to $v$ itself if $i=2l-1$). Also, recall that since $(u,v)\in\mathcal{P}_\kappa[p,l]$, it means that $(\pi(a),\pi(b))\in\mathcal{P}_{\kappa-1}[p',l]$ (here, and in the rest of this proof, we identify $\pi(a),\pi(b)$ with the $\pi(a)$-th input port and $\pi(b)$-output port of $\mathcal{H}_{\kappa-1}[p',l]$). Thus, when using the $2l$ edges described above, we can also find paths inside the copies $H_1,...,H_{2l-1}$, each of them with weight $(2l(\kappa-1)+1)(2l-1)^{\kappa-1}$. We conclude that there is a path in $\mathcal{H}_\kappa[p,l]$ with weight
\[2l(2l-1)^\kappa+(2l-1)(2l(\kappa-1)+1)(2l-1)^{\kappa-1}=(2l\kappa+1)(2l-1)^\kappa~.\]
Denote this path by $P_{u,v}$. Since we already proved that the distance between any input port and any output port is at least $(2l\kappa+1)(2l-1)^\kappa=w(P_{u,v})$, we deduce that $P_{u,v}$ is a shortest path between $u,v$.

Let $P'_{u,v}$ be a different path than $P_{u,v}$ between $u$ and $v$ in $\mathcal{H}_\kappa[p,l]$. Consider the list of copies of $\mathcal{H}_{\kappa-1}[p',l]$ that $P'_{u,v}$ passes through, by the same order they appear on $P'_{u,v}$. Since $P'_{u,v}\neq P_{u,v}$, there are two cases: either this list is identical to $H_1,H_2,...,H_{2l-1}$, but for at least one $j$ the path $P'_{u,v}$ passes through $H_j$ using a different path than $P_{u,v}$, or this list is not identical to $H_1,H_2,...,H_{2l-1}$.

In the first case, by the induction hypothesis, the path that $P'_{u,v}$ uses inside $H_j$ has weight strictly larger than $(2l(\kappa-1)+1)(2l-1)^{\kappa-1}$. The path inside the other copies has weight of at least $(2l(\kappa-1)+1)(2l-1)^{\kappa-1}$, again by the induction hypothesis. Together with the $2l$ edges with weight $(2l-1)^\kappa$ that connect these copies, we get that the weight of $P'_{u,v}$ is strictly more than
\[2l(2l-1)^\kappa+(2l-1)(2l(\kappa-1)+1)(2l-1)^{\kappa-1}=(2l\kappa+1)(2l-1)^\kappa~.\]

In the second case, we "translate" the path $P'_{u,v}$ into a path $Q$ in $\ddot{B}[p,l]$, by replacing each copy of $\mathcal{H}_{\kappa-1}[p',l]$ it passes through by the corresponding vertex of $\ddot{B}[p,l]$. The path $Q$ is different from the path $P=(u,u_1,u_2,...,u_{2l-1},v)$. Since the latter is the unique $u-v$ shortest path in $\ddot{B}[p,l]$, the path $Q$ must pass through a layer of $\ddot{B}[p,l]$ more than once (otherwise its weight would be equal to the weight of $P$). That is, $Q$ passes through at least $2l+1$ internal vertices of $\ddot{B}[p,l]$ and contains at least $2l+2$ edges. This means that the path $P'_{u,v}$ %passes through at least $2l+1$ copies of $\mathcal{H}_{\kappa-1}[p',l]$, and 
contains at least $2l+2$ edges of weight $(2l-1)^\kappa$. Also, note that $P'_{u,v}$ (like any other input-output path in $\mathcal{H}_\kappa[p,l]$) must pass through at least $2l-1$ copies of $\mathcal{H}_{\kappa-1}[p',l]$. By the induction hypothesis, the weight of $P'_{u,v}$ is at least
\[(2l-1)(2l(\kappa-1)+1)(2l-1)^{\kappa-1}+(2l+2)(2l-1)^\kappa>(2l\kappa+1)(2l-1)^\kappa~.\]
%\alert{I may be confused here, but if the path goes back and forth, it may not completely pass through a copy of $\mathcal{H}_{\kappa-1}[p',l]$, so I don't see how to get the bound you used (I also don't think we need it, I guess $2l-$ copies plus the $2l+1$ heavy edges will be enough to say it is too long.} \alert{Idan: Okay, I changed it.}
%\[>(2l-1)(2l(\kappa-1)+1)(2l-1)^{\kappa-1}+2l(2l-1)^\kappa=(2l\kappa+1)(2l-1)^\kappa~.\]

In conclusion, the path $P_{u,v}$ has length $(2l\kappa+1)(2l-1)^\kappa$, while any other $u-v$ path in $\mathcal{H}_\kappa[p,l]$ has a larger weight. Thus, $P_{u,v}$ is a unique shortest path between $u,v$ in $\mathcal{H}_\kappa[p,l]$.

To complete the proof, we have to show that if $P_{u,v}$ and $P_{u',v'}$ share a critical edge, for some $(u,v),(u',v')\in\mathcal{P}_\kappa[p,l]$, then it must be that $(u,v)=(u',v')$. Let $(u,v),(u',v')$ be such two pairs. Since their paths share a critical edge, they must pass through the same copy of $\mathcal{H}_{\kappa-1}[p',l]$. Denote this copy by $H$. We saw that the path $P_{u,v}$ is originated in a path $P$ in the graph $\ddot{B}[p,l]$, which is the unique shortest path between $u,v$ that satisfies $(u,v)\in\ddot{\mathcal{P}}[p,l]$, and that the edges of $P$ are alternately labeled by some $a,b\in\ddot{\mathcal{L}}[p,l]$. Moreover, $(\pi(a),\pi(b))\in\mathcal{P}_{\kappa-1}[p',l]$, because $(u,v)\in\ddot{\mathcal{P}}[p,l]$ (recall that we still identify the numbers $\pi(a),\pi(b)$ with input and output ports of $\mathcal{H}_{\kappa-1}[p',l]$). Symmetrically, the path $P'$ and the labels $a',b'$ that correspond to the path $P_{u',v'}$ in $\ddot{B}[p,l]$ satisfy $(\pi(a'),\pi(b'))\in\mathcal{P}_{\kappa-1}[p',l]$. Note that the unique shortest paths between $\pi(a),\pi(b)$ and between $\pi(a'),\pi(b')$ still share a critical edge. %\alert{The last sentence is confusing, since these are labels, not vertices. Perhaps elaborate more} \alert{Idan: I added another clarification that we identify these numbers with ports of $\mathcal{H}_{\kappa-1}[p',l]$ (in the parenthesis that starts with "Recall that").} 
Thus, by the induction hypothesis, $(\pi(a),\pi(b))=(\pi(a'),\pi(b'))$, or equivalently $(a,b)=(a',b')$.

Notice that the two paths $P,P'$ pass through the same vertex in $\ddot{B}[p,l]$, because $P_{u,v},P_{u',v'}$ pass through the same copy $H$. We also know that they are alternately labeled by the same to labels $(a,b)=(a',b')$. By Lemma \ref{lemma:SecondBaseGraph}, it must be that $P=P'$, and in particular $(u,v)=(u',v')$, otherwise they cannot have the same pair of labels $a,b$. This completes the inductive proof.

\end{proof}

We are now ready to prove our main theorem.

\begin{theorem} \label{thm:LowerBoundSmallStretch}
For infinitely many integers $n>0$, and for any integer $1<\kappa\leq\log\log n$ and real $0<\epsilon<\frac{1}{12\kappa}$, there is an $n$-vertex graph $G=(V,E)$ and a set of pairs $\mathcal{P}\subseteq V^2$ with size at least $n^{1+\frac{1}{2^{\kappa+1}-1}-o(1)}$, such that any $\mathcal{P}$-pairwise $(1+\epsilon)$-spanner for $G$ must have at least $\beta\cdot|\mathcal{P}|$ edges, where $\beta=\Omega\left(\frac{1}{\epsilon\kappa^2}\right)^\kappa$.
\end{theorem}

\begin{proof}

Fix $\kappa$ and $0<\epsilon\leq\frac{1}{12\kappa}$. Let $G,\mathcal{P}$ be $\mathcal{H}_\kappa[p,l],\mathcal{P}_\kappa[p,l]$, for $l=\left\lfloor\frac{1}{6\epsilon\kappa}\right\rfloor\geq1$ and some arbitrary $p$. Denote $b=\left\lfloor\frac{2l-1}{\kappa}\right\rfloor+1$. We will show that any $\mathcal{P}$-pairwise $(1+\epsilon)$-spanner for $G$ must have size of at least $b^\kappa|\mathcal{P}|$. This proves the theorem for 
\[\beta=b^\kappa\geq\left(\frac{2l-1}{\kappa}\right)^\kappa\geq\left(\frac{\frac{1}{3\epsilon\kappa}-3}{\kappa}\right)^\kappa\geq\left(\frac{\frac{1}{3\epsilon\kappa}-\frac{3}{12\epsilon\kappa}}{\kappa}\right)^\kappa=\left(\frac{1}{12\epsilon\kappa^2}\right)^\kappa~.\]

For the size of $\mathcal{P}$, recall that the number of vertices in $G$ is $n\leq2(2l)^\kappa p^{2-\frac{1}{2^\kappa}}$, by Lemma \ref{lemma:NumberOfVertices}. 
Thus, 
\[p^2\geq\left(\frac{n}{2(2l)^\kappa}\right)^{\frac{2}{2-\frac{1}{2^\kappa}}}\geq n^{1+\frac{1}{2^{\kappa+1}-1}}\cdot\frac{(2l)^{-2\kappa}}{4}~.\]

Therefore, by Lemma \ref{lemma:NumberOfPairs}, the size of $\mathcal{P}$ is at least
\[\frac{p^2}{\xi(\sqrt{p},l)^{2\kappa}}\geq n^{1+\frac{1}{2^{\kappa+1}-1}}\cdot\frac{(2l)^{-2\kappa}}{4\cdot2^{O\left(\kappa\sqrt{\log p\log l}\right)}}=n^{1+\frac{1}{2^{\kappa+1}-1}}\cdot2^{-2\kappa\log(2l)-2-O\left(\kappa\sqrt{\log n\log l}\right)}\]
\[=n^{1+\frac{1}{2^{\kappa+1}-1}-o(1)}~.\]
Here we used Lemma \ref{lemma:SecondBaseGraph} to bound $\xi(\sqrt{p},l)$, we used the fact that $\kappa\leq\log\log n$, and we used the fact that $p\leq n$ (this is trivial, by the construction of $\mathcal{H}_\kappa[p,l]$).

%\alert{Maybe mention that $\kappa\le\log\log n$? (for the last step). Is it actually true?} \alert{Idan: I added the assumption that $\kappa\leq\log\log n$ to the theorem. I think it's okay to assume that, because we use this theorem only for $\kappa$ such that $|\mathcal{P}|\approx n^{1+\frac{1}{2^\kappa}}$.}

Let $S$ be a subset of the edges of $G=\mathcal{H}_\kappa[p,l]$, with $|S|<b^\kappa|\mathcal{P}|=b^\kappa|\mathcal{P}_\kappa[p,l]|$. By Lemma \ref{lemma:ShortestPathsProperties}, for every pair $(u,v)\in\mathcal{P}$, the unique $u-v$ shortest path in $G$ has a disjoint set of critical edges that it goes through. %By the construction of $\mathcal{H}_\kappa[p,l]$, this set contains at least one critical edge. 
Therefore, there must be some $(u,v)\in\mathcal{P}$ such that $S$ contains less than $b^\kappa$ of its critical edges.

%\alert{Aren't all the critical edges the path uses are disjoint from the other paths? why just 1?}
%\alert{Idan: I don't know why I wrote this commented-out sentence...}

We prove by induction on $\kappa\geq0$ that in the graph $\mathcal{H}_\kappa[p,l]$, if a pair $(u,v)\in\mathcal{P}_\kappa[p,l]$ has less than $b^\kappa$ of its critical edges in $S$, then
\[d_S(u,v)\geq(2l\kappa+1)(2l-1)^\kappa+2(2l-b)^\kappa~.\]
For $\kappa=0$, a pair $(u,v)\in\mathcal{P}_0[p,l]$ that has less than $b^0=1$ of its critical edges in $S$, means a pair such that $(u,v)\notin S$. Since the graph $\mathcal{H}_0[p,l]=K_{p,p}$ is bipartite, $d_S(u,v)\geq3=(2l\cdot0+1)(2l-1)^0+2(2l-b)^0$.

Fix $\kappa>0$, and let $P'_{u,v}$ be a $u-v$ shortest path in $S$, for a pair $(u,v)\in\mathcal{P}_\kappa[p,l]$ that has less than $b^\kappa$ of its critical edges in $S$. Consider the path $Q$ in $\ddot{B}[p,l]$ that is obtained by replacing each copy of $\mathcal{H}_{\kappa-1}[p',l]$ that $P'_{u,v}$ passes through by the corresponding vertex of $\ddot{B}[p,l]$. If $Q$ is not the unique shortest path $P$ between $u,v$ in $\ddot{B}[p,l]$, then the path $Q$ must pass through a layer of $\ddot{B}[p,l]$ more than once (otherwise its weight would be equal to the weight of $P$). That is, $Q$ passes through at least $2l+1$ internal vertices of $\ddot{B}[p,l]$ and contains at least $2l+2$ edges. This means that the path $P'_{u,v}$ contains at least $2l+2$ of weight $(2l-1)^\kappa$. Note that $P'_{u,v}$ must pass through at least $2l-1$ copies of $\mathcal{H}_{\kappa-1}[p',l]$, just to get from $u$ to $v$. By Lemma \ref{lemma:ShortestPathsProperties}, the weight of $P'_{u,v}$ is at least
\begin{eqnarray*}
&&(2l-1)(2l(\kappa-1)+1)(2l-1)^{\kappa-1}+(2l+2)(2l-1)^\kappa\\
&=&(2l-1)(2l(\kappa-1)+1)(2l-1)^{\kappa-1}+2l(2l-1)^\kappa+2(2l-1)^\kappa\\
&=&(2l\kappa+1)(2l-1)^\kappa+2(2l-1)^\kappa\geq(2l\kappa+1)(2l-1)^\kappa+2(2l-b)^\kappa~.
\end{eqnarray*}

%\alert{Same concern as before, there are $2l+2$ heavy edges, but I think you can also argue that it passes through $2l-1$ copies. Surprisingly, the math below works because the first thing we do is replace that $2l+1$ by $2l-1$...} \alert{Idan: I made the same change.}

If the path $Q$ equals to $P$, the unique $u-v$ shortest path in $\ddot{B}[p,l]$, then it passes through exactly $2l-1$ internal vertices of $\ddot{B}[p,l]$. Hence, $P'_{u,v}$ passes through exactly $2l-1$ of copies of $\mathcal{H}_{\kappa-1}[p',l]$. We would like to know how many of them contain less than $b^{\kappa-1}$ critical edges of $(u,v)$. Note that there cannot be less than $2l-b$ such copies: otherwise the number of copies with \textit{at least} $b^{\kappa-1}$ critical edges of $(u,v)$ is more than $2l-1-(2l-b)=b-1$, i.e., at least $b$, so there are at least $b\cdot b^{\kappa-1}=b^\kappa$ critical edges of $(u,v)$ in $S$, in contradiction. Therefore, the number $t$ of copies in which there are less than $b^{\kappa-1}$ critical edges of $(u,v)$ is at least $2l-b$. In these copies, $P'_{u,v}$ suffers a weight of at least $(2l(\kappa-1)+1)(2l-1)^{\kappa-1}+2(2l-b)^{\kappa-1}$, by the induction hypothesis. In the other $2l-1-t$ copies, $P'_{u,v}$ must have a weight of at least $(2l(\kappa-1)+1)(2l-1)^{\kappa-1}$, by Lemma \ref{lemma:ShortestPathsProperties}. Together with the $2l$ edges that connect these copies and have weight of $(2l-1)^\kappa$, we get
\begin{eqnarray*}
d_S(u,v)&=&w(P'_{u,v})=2l(2l-1)^\kappa+(2l-1-t)(2l(\kappa-1)+1)(2l-1)^{\kappa-1}\\
&&+t\left[(2l(\kappa-1)+1)(2l-1)^{\kappa-1}+2(2l-b)^{\kappa-1}\right]\\
&=&2l(2l-1)^\kappa+(2l-1)(2l(\kappa-1)+1)(2l-1)^{\kappa-1}+t\cdot2(2l-b)^{\kappa-1}\\
&=&(2l\kappa+1)(2l-1)^\kappa+t\cdot2(2l-b)^{\kappa-1}\\
&\geq&(2l\kappa+1)(2l-1)^\kappa+(2l-b)\cdot2(2l-b)^{\kappa-1}\\
&=&(2l\kappa+1)(2l-1)^\kappa+2(2l-b)^\kappa~.
\end{eqnarray*}

This completes the inductive proof. It shows that there is a pair $(u,v)\in\mathcal{P}_\kappa[p,l]=\mathcal{P}$ with
\begin{eqnarray*}
d_S(u,v)&\geq&(2l\kappa+1)(2l-1)^\kappa+2(2l-b)^\kappa\\
&=&(2l\kappa+1)(2l-1)^\kappa\left(1+\frac{2(2l-b)^\kappa}{(2l\kappa+1)(2l-1)^\kappa}\right)\\
&=&d_G(u,v)\cdot\left(1+\frac{2}{(2l\kappa+1)}\left(\frac{2l-b}{2l-1}\right)^\kappa\right)\\
&=&d_G(u,v)\cdot\left(1+\frac{2}{(2l\kappa+1)}\left(1-\frac{b-1}{2l-1}\right)^\kappa\right)\\
&\geq&d_G(u,v)\cdot\left(1+\frac{2}{(2l\kappa+1)}\left(1-\frac{1}{\kappa}\right)^\kappa\right)
\stackrel{\kappa>1}{\geq} d_G(u,v)\cdot\left(1+\frac{2}{(2l\kappa+1)}\cdot\frac{1}{4}\right)\\
&>&d_G(u,v)\cdot\left(1+\frac{2}{3l\kappa}\cdot\frac{1}{4}\right)=d_G(u,v)\cdot\left(1+\frac{1}{6l\kappa}\right)\geq d_G(u,v)\cdot\left(1+\epsilon\right)~,
\end{eqnarray*}
Where in the last step we used the definition of $l=\left\lfloor\frac{1}{6\epsilon\kappa}\right\rfloor$. Thus, $S$ cannot be a $\mathcal{P}_\kappa[p,l]$-pairwise $(1+\epsilon)$-spanner. In other words, any $\mathcal{P}_\kappa[p,l]$-pairwise $(1+\epsilon)$-spanner must have size of at least $\beta|\mathcal{P}|$, where $\beta\geq\left(\frac{1}{12\epsilon\kappa^2}\right)^\kappa$.

\end{proof}

\begin{remark}
Notice that the proof of Theorem \ref{thm:LowerBoundSmallStretch} also works for every subset of the pairs $\mathcal{P}_\kappa[p,l]$. Therefore, the phrasing of Theorem \ref{thm:LowerBoundSmallStretch} may be strengthen as follows.

For infinitely many integers $n>0$, and for any integer $1<\kappa\leq\log\log n$ and real $0<\epsilon<\frac{1}{12\kappa}$, there is an $n$-vertex graph $G=(V,E)$ and a number $Q=n^{1+\frac{1}{2^{\kappa+1}-1}-o(1)}$, such that for every integer $q\leq Q$, there is a set of pairs $\mathcal{P}\subseteq V^2$ with size $q$, such that any $\mathcal{P}$-pairwise $(1+\epsilon)$-spanner for $G$ must have at least $\beta\cdot|\mathcal{P}|$ edges, where $\beta=\Omega\left(\frac{1}{\epsilon\kappa^2}\right)^\kappa$.
\end{remark}

\subsection{Large Stretch} \label{sec:HighStretchLowerBound}

The lower bound for pairwise spanners with large stretch is achieved using a graph with high girth and large number of edges. This properties of a graph were used for proving lower bounds for distance oracles and spanners (see \cite{TZ01}) and also for hopsets (see \cite{NS22}). In particular, we use the same graph that was used in \cite{NS22}, which was introduced by Lubotzky, Phillips and Sarnak in \cite{LPS88}. This graph has the additional convenient property of being \textit{regular}, besides its high girth and large number of edges. Its exact properties are described in the following theorem.

\begin{theorem}[\cite{LPS88}] \label{thm:LPS}
For infinitely many integers $n\in\mathbb{N}$, and for every integer $k\geq1$, there exists a $(p+1)$-regular graph $G=(V,E)$ with $|V|=n$ and girth at least $\frac{4}{3}k(1-o(1))$, where $p=D\cdot n^{\frac{1}{k}}$, for some universal constant $D$.
\end{theorem}

Fix $\alpha,k\geq1$ such that $k\geq\alpha+1$, and a large enough $n\in\mathbb{N}$ as in Theorem \ref{thm:LPS}. Let $G=(V,E)$ be the corresponding $(p+1)$-regular graph from Theorem \ref{thm:LPS}. The girth of $G$ is at least $\frac{4}{3}k(1-o(1))$, thus larger than $k$. Denote $\delta=\left\lfloor\frac{k}{\alpha+1}\right\rfloor$. Define the following set of pairs $\mathcal{P}_0\subseteq V^2$.
\[\mathcal{P}_0=\{(u,v)\in V^2\;|\;d_G(u,v)=\delta\}~.\]

\begin{lemma} \label{lemma:UniqueSP}
For every $(u,v)\in\mathcal{P}_0$, there is a unique shortest path $P_{u,v}$ between $u,v$ in $G$. Furthermore, for every tour $P'_{u,v}$ between $u,v$, that has length $|P'_{u,v}|\leq\alpha\delta$, we have $P_{u,v}\subseteq P'_{u,v}$.
\end{lemma}

\begin{proof}

Let $P_{u,v}$ be some $u-v$ shortest path in $G$, and let $P'_{u,v}$ be a $u-v$ tour with $|P'_{u,v}|\leq\alpha\delta$. If $P_{u,v}\not\subseteq P'_{u,v}$, then the union $P_{u,v}\cup P'_{u,v}$ contains a cycle. This cycle is of length at most
\[|P_{u,v}|+|P'_{u,v}|\leq\delta+\alpha\delta=(\alpha+1)\delta\leq k~,\]
by the definition of $\delta=\left\lfloor\frac{k}{\alpha+1}\right\rfloor$. This is a contradiction to the fact that the girth of $G$ is larger than $k$. Hence, $P_{u,v}\subseteq P'_{u,v}$. In case $P'_{u,v}$ is also a $u-v$ shortest path, i.e., $|P'_{u,v}|=\delta$, then $P_{u,v}\subseteq P'_{u,v}$ implies $P_{u,v}=P'_{u,v}$. That is, $P_{u,v}$ is the unique shortest path between $u,v$ in $G$.

\end{proof}

Henceforth, we use the notations from Lemma \ref{lemma:UniqueSP}, that is, we denote by $P_{u,v}$ the $u-v$ shortest path in $G$.

\begin{lemma} \label{lemma:Incompressibility}
Let $\mathcal{P}\subseteq\mathcal{P}_0$ be some subset, and suppose that $S$ is a $\mathcal{P}$-pairwise $\alpha$-spanner for $G$. Then,
\[\bigcup_{(u,v)\in\mathcal{P}}P_{u,v}\subseteq S~.\]
\end{lemma}

\begin{proof}

Fix some $(u,v)\in\mathcal{P}\subseteq\mathcal{P}_0$. Since $S$ has stretch $\alpha$ for every pair in $\mathcal{P}$, we know that there is a $u-v$ path $P'_{u,v}\subseteq S$ with $|P'_{u,v}|\leq\alpha|P_{u,v}|=\alpha\delta$. By Lemma \ref{lemma:UniqueSP}, the $u-v$ shortest path satisfies $P_{u,v}\subseteq P'_{u,v}\subseteq S$. In conclusion, $P_{u,v}\subseteq S$ for every $(u,v)\in\mathcal{P}$, thus $\bigcup_{(u,v)\in\mathcal{P}}P_{u,v}\subseteq S$.

\end{proof}

The following lemma describes several combinatorial properties of the graph $G$.

\begin{lemma} \label{lemma:CombinatorialProperties}
The number of edges in $G=(V,E)$ is $|E|=\frac{n(p+1)}{2}$. The number of paths in $\mathcal{P}_0$ is $\frac{n(p+1)p^{\delta-1}}{2}$. For every edge $e\in E$, there are $\delta p^{\delta-1}$ pairs $(u,v)\in\mathcal{P}_0$ such that $e\in P_{u,v}$.
\end{lemma}

\begin{proof}

The number of edges $|E|$ is half the sum of the degrees in $G$. Since $G$ is $(p+1)$-regular, we get $|E|=\frac{n(p+1)}{2}$. 

Now fix some $u\in V$, and consider its BFS tree up to distance $\delta$. The root $u$ has $p+1$ children in this tree, and every other internal vertex has a set of $p$ children, disjoint from the children set of any other vertex in this tree. This is true since otherwise there would be a cycle of length at most $2\delta\leq(\alpha+1)\delta\leq k$, in contradiction to the girth of $G$ being larger than $k$. Thus, the number of vertices $v$ in the $\delta$-th layer of this tree, is $(p+1)p^{\delta-1}$. That is, the number of $v\in V$ such that $d_G(u,v)=\delta$ is $(p+1)p^{\delta-1}$. Hence,
\[\sum_{u\in V}|\{v\in\;|\;d_G(u,v)=\delta\}|=\sum_{u\in V}(p+1)p^{\delta-1}=n(p+1)p^{\delta-1}~.\]
In this sum, each pair $(u,v)\in\mathcal{P}_0$ is counted exactly twice, therefore $|\mathcal{P}_0|=\frac{n(p+1)p^{\delta-1}}{2}$.

The proof of the third property is very similar. Fix some edge $e=(v_1,v_2)\in E$. For every integer $i\in[0,\delta-1]$, consider the BFS tree $T^i_1$ of $v_1$ up to distance $i$. Symmetrically, $T^i_2$ denotes the BFS tree of $v_2$ up to distance $i$. As before, the children sets of the vertices in $T^i_1\cup T^{\delta-1-i}_2$ are disjoint - otherwise there would be a cycle of length at most
\[\max\{i+1+\delta-1-i,2i,2(\delta-1-i)\}<2\delta\leq k~,\]
in contradiction. Hence, the number of vertices in the $i$-th layer of $T^i_1$ is $p^i$ (note that $v_1$ has $p$ children in this tree), and the number of vertices in the $(\delta-1-i)$-th layer of $T^{\delta-1-i}_2$ is $p^{\delta-1-i}$. For every pair $(u,v)\in\mathcal{P}_0$ such that $e\in P_{u,v}$, the path $P_{u,v}$ has one end in the $i$-th layer of $T^i_1$ and the other end in the $(\delta-1-i)$-th layer of $T^{\delta-1-i}_2$, for some $i\in[0,\delta-1]$. See Figure \ref{fig:HighStretchLowerBound} for an illustration. Thus, the number of such pairs is
\[\sum_{i=0}^{\delta-1}p^i\cdot p^{\delta-1-i}=\delta p^{\delta-1}~.\]

\end{proof}

\begin{center}
\begin{figure}[ht!]
    \centering
    \includegraphics[width=6cm, height=5cm]{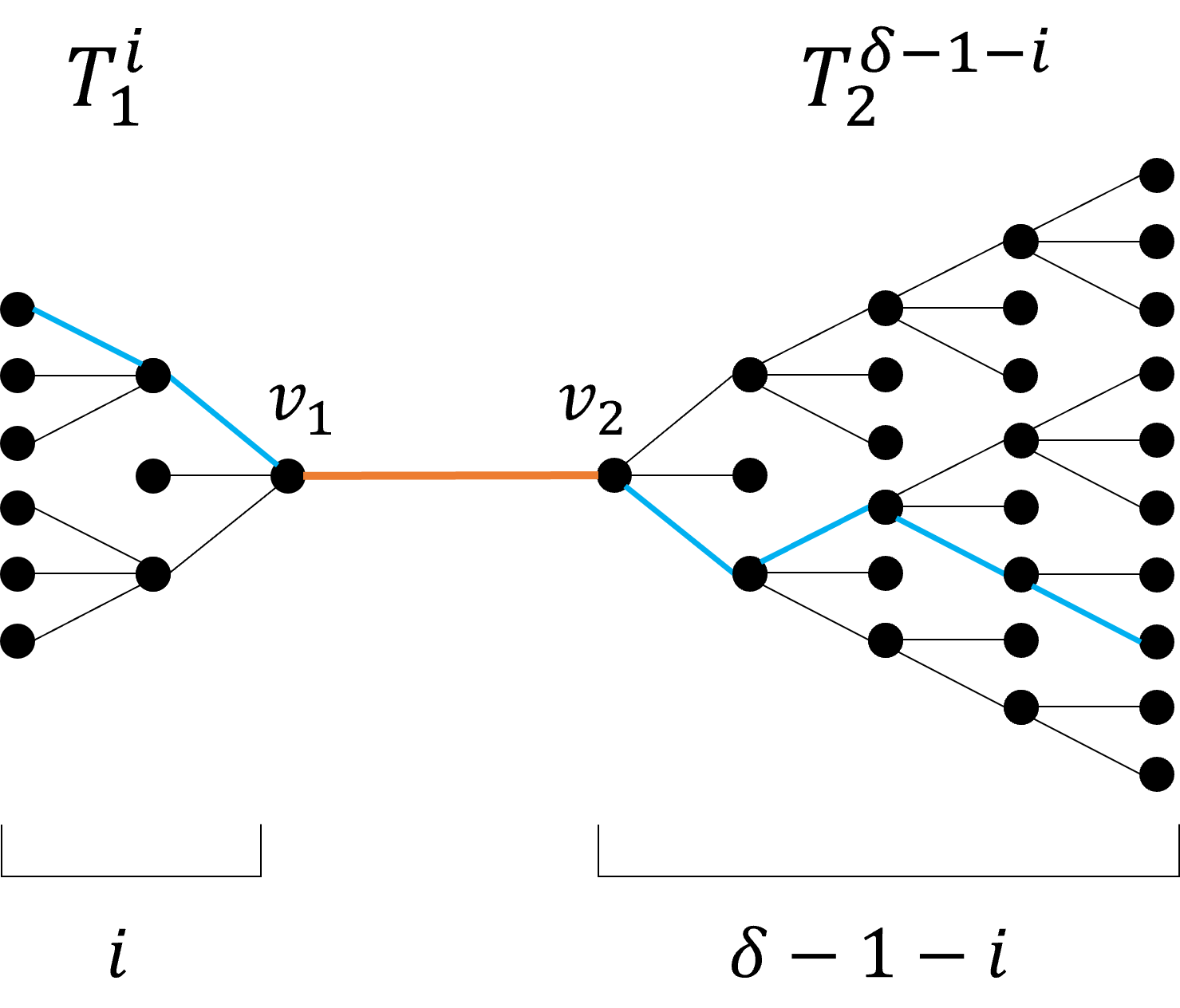}
    \caption{Given an edge $e=(v_1,v_2)$ (colored orange in the figure), we consider the BFS trees $T^i_1$ and $T^{\delta-1-i}_2$ rooted at $v_1$ and $v_2$ respectively, up to distance $i$ and $\delta-1-i$ respectively. There are no cycles within these two trees, because of the girth guarantee. By regularity, we know that each vertex in these trees, except the leaves, has exactly $p$ children. Every path of length $\delta$ that passes through $e$, such as the blue path in the figure, is determined by a leaf of $T^i_1$ and a leaf of $T^{\delta-1-i}_2$.}
    \label{fig:HighStretchLowerBound}
\end{figure}
\end{center}

We are now ready for the main theorem of this section.

\begin{theorem} \label{thm:LowerBoundLargeStretch}
For infinitely many integers $n>0$, and for any real $\alpha\geq1$ and integer $k$ such that $\alpha+1\leq k\leq\log n$, there is an $n$-vertex graph $G=(V,E)$ and a set of pairs $\mathcal{P}\subseteq V^2$ with size $\Theta\left(\frac{\alpha}{k}\cdot n^{1+\frac{1}{k}}\right)$, such that any $\mathcal{P}$-pairwise $\alpha$-spanner for $G$ must have at least $\Omega(n^{1+\frac1k})$ edges, that is, the size overhead is $\beta=\Omega\left(\frac{k}{\alpha}\right)$.
\end{theorem}

\begin{proof}

%Fix an integer $k>1$ \alert{I put strict inequality otherwise the following inequality is not true.} and real $\alpha\geq1$. If $k<\alpha+1$, then $\frac{k}{\alpha}<\frac{k}{k-1}\leq2$, and the statement of the theorem is trivial since $|\mathcal{P}|$ is a trivial lower bound for $\mathcal{P}$-pairwise $\alpha$-spanners, for any $\alpha\geq1$.

By Theorem \ref{thm:LPS}, for infinitely many integers $n>0$, there is a $(p+1)$-regular graph $G=(V,E)$ with $n$ vertices and girth larger than $k$, where $p=\Theta(n^{\frac{1}{k}})$. We use the same notations for $\delta$ and $\mathcal{P}_0$ as in the beginning of this section.

Let $\mathcal{P}\subseteq\mathcal{P}_0$ be a subset that is formed by sampling each pair in $\mathcal{P}_0$ independently with probability $\frac{1}{\delta p^{\delta-1}}$. The expected number of pairs in $\mathcal{P}$ is 
\[\frac{|\mathcal{P}_0|}{\delta p^{\delta-1}}=\frac{n(p+1)p^{\delta-1}}{2\delta p^{\delta-1}}=\frac{n(p+1)}{2\delta}=\frac{|E|}{\delta}~,\]
by Lemma \ref{lemma:CombinatorialProperties}. Moreover, by Chernoff bound,
\begin{equation} \label{eq:PSize}
    \Pr\left[\left||\mathcal{P}|-\frac{|E|}{\delta}\right|>\frac{|E|}{2\delta}\right]\leq2e^{-\frac{|E|}{12\delta}}=2e^{-\frac{n(p+1)}{24\delta}}\leq2e^{-\frac{n}{24\log n}}\leq2e^{-2}~,
\end{equation}
for large enough $n$, where we used the fact that $\delta\leq k\leq\log n$.

We say that a pair $(u,v)\in\mathcal{P}_0$ {\em covers} an edge $e\in E$ if $e\in P_{u,v}$. For an edge $e\in E$, the number of pairs in $\mathcal{P}_0$ that cover $e$ is $\delta p^{\delta-1}$, by Lemma \ref{lemma:CombinatorialProperties}. Therefore, the probability that none of the pairs that cover $e$ are in $\mathcal{P}$ is $(1-\frac{1}{\delta p^{\delta-1}})^{\delta p^{\delta-1}}\leq e^{-1}$. Hence, If we denote by $E'\subseteq E$ the set of edges that are not covered by any $(u,v)\in\mathcal{P}$, then $\mathbb{E}[|E'|]\leq|E|\cdot e^{-1}$. By Markov's inequality,
\begin{equation} \label{eq:E'Size}
    \Pr\left[|E'|>\frac{2}{e}|E|\right]\leq\frac{1}{2}
\end{equation}

Now, by the union bound, the probability that either $\left||\mathcal{P}|-\frac{|E|}{\delta}\right|>\frac{|E|}{2\delta}$, or $|E'|>\frac{2}{e}|E|$, is at most $2e^{-2}+\frac{1}{2}<1$, using Inequalities (\ref{eq:PSize}) and (\ref{eq:E'Size}). Therefore, there is a way to choose the subset $\mathcal{P}\subseteq\mathcal{P}_0$, such that the number of edges in $E$ that are not covered by any $(u,v)\in\mathcal{P}$ is at most $\frac{2}{e}|E|$, and such that $\frac{|E|}{2\delta}\leq|\mathcal{P}|\leq\frac{3|E|}{2\delta}$. In particular,
\[|\mathcal{P}|=\Theta\left(\frac{|E|}{\delta}\right)=\Theta\left(\frac{n(p+1)}{2\delta}\right)=\Theta\left(\frac{n^{1+\frac{1}{k}}}{\delta}\right)=\Theta\left(\frac{\alpha}{k}\cdot n^{1+\frac{1}{k}}\right)~.\]
For this choice of $\mathcal{P}$, the number of edges $e\in E$ that satisfy $e\in P_{u,v}$ for some $(u,v)\in\mathcal{P}$ is at least $\left(1-\frac{2}{e}\right)|E|$. Notice that these are exactly the edges in $\bigcup_{(u,v)\in\mathcal{P}}P_{u,v}$. By Lemma \ref{lemma:Incompressibility}, any $\mathcal{P}$-pairwise $\alpha$-spanner for $G$ must contain this set, and therefore must have size at least
\[\left(1-\frac{2}{e}\right)|E|\geq\frac{1}{4}|E|\geq\frac{1}{4}\cdot\frac{2\delta}{3}|\mathcal{P}|=\frac{\delta}{6}|\mathcal{P}|~.\]
This proves the theorem for $\beta=\frac{\delta}{6}=\Omega\left(\frac{k}{\alpha}\right)$.

\end{proof}

\section{An Upper Bound for Path-Reporting Pairwise Spanners} \label{sec:UpperBounds}

To construct path-reporting pairwise spanners with relatively large stretch (e.g., pairwise $\alpha$-spanners with $\alpha>3$) and small size, we combine hopsets that have large stretch with pairwise spanners that have small stretch. More precisely, suppose that the hopset $H$ for an $n$-vertex graph $G=(V,E)$ has stretch $\alpha$ and hopbound $\beta$. Recall that for every pair of vertices $(u,v)\in V^2$, there is a path $P_{u,v}\subseteq G\cup H$ with at most $\beta$ edges, and weight at most $\alpha\cdot d_G(u,v)$. To construct a pairwise spanner for the graph $G$ and some pairs set $\mathcal{P}$, we first add the edges of $P_{u,v}\cap G$, for every $(u,v)\in\mathcal{P}$, to the desired spanner. Note that for each pair in $\mathcal{P}$ we added at most $\beta$ edges. Next, we consider $H$ as a set of pairs of vertices in $G$, and we use a pairwise spanner $Q\subseteq E$ with small stretch $t$ on $H$. We insert the edges of $Q$ to our pairwise spanner. The result is a $\mathcal{P}$-pairwise spanner for $G$, with stretch at most $t\cdot\alpha$, and size at most $\beta\cdot|\mathcal{P}|+|Q|$. To make this pairwise spanner path-reporting, we make sure that the pairwise $t$-spanner we used is also path-reporting, and we also store in the oracle the paths $P_{u,v}$ for every $(u,v)\in\mathcal{P}$.

A similar technique, of combining a hopset with a pairwise spanner, appeared in \cite{ES23}. In \cite{ES23}, however, preserving the edges of the hopset $H$ was mostly done by relying on the properties of $H$ itself, and claiming that there is a small size \textit{exact} preserver (i.e., pairwise $1$-spanner) for $H$.

The combination described above, between a hopset and a pairwise spanner with small stretch, is relatively simple. We later show a more complicated construction that achieves better results. This construction relies on specific properties of the hopset we use. Both of these constructions of pairwise spanners with large stretch are also path-reporting.

\subsection{Warm Up: Path-Reporting Pairwise Spanners via Hopsets}

We start with the relatively simple procedure of constructing a path-reporting pairwise spanner, via a combination of a hopset and a path-reporting pairwise $t$-spanner, for small $t$. 

\begin{lemma} \label{lemma:CombineHopsetSpanner}
Let $G=(V,E)$ be an undirected weighted graph on $n$ vertices, and let $H$ be a hopset for $G$ with stretch $\alpha$, hopbound $\beta$ and size $M(n)$. Suppose that for every set of pairs $\mathcal{P}\subseteq V^2$, there is a path-reporting $\mathcal{P}$-pairwise $t$-spanner for $G$ with query time $O(1)$ and size $Q(n,|\mathcal{P}|)$.

Then, for every set of pairs $\mathcal{P}$ in $G$, there is a path-reporting $\mathcal{P}$-pairwise $t\alpha$-spanner for $G$, with query time $O(1)$ and size
\[\beta\cdot|\mathcal{P}|+Q(n,M(n))~.\]
\end{lemma}

\begin{proof}

For every pair of vertices $(u,v)\in V^2$, there is a path $P_{u,v}\subseteq G\cup H$ with at most $\beta$ edges, and weight at most $\alpha\cdot d_G(u,v)$. We consider the hopset $H$ as a set of pairs in $G$, and conclude that there is a path-reporting $H$-pairwise $t$-spanner $(S,D)$ for $G$ with size at most $Q(n,|H|)=Q(n,M(n))$.

Define a path-reporting pairwise spanner $(R,D')$ as $R=S\cup\bigcup_{(u,v)\in\mathcal{P}}(P_{u,v}\cap E)$. In the oracle $D'$, we store the oracle $D$ and all the paths $P_{u,v}$ for every $(u,v)\in\mathcal{P}$. Given a pair $(u,v)\in\mathcal{P}$, let $P_{u,v}=(u=u_0,u_1,u_2,...,u_b=v)\subseteq E\cup H$ be the stored path for $u,v$. By the hopbound guarantee of the hopset $H$, we know that $b\leq\beta$, and by the stretch guarantee we know that $\sum_{i=0}^{b-1}d_G(u_i,u_{i+1})\leq\alpha d_G(u,v)$. For every $i\in[0,b-1]$, if $(u_i,u_{i+1})\in E$, then also $(u_i,u_{i+1})\in R$, and thus $d_R(u_i,u_{i+1})=d_G(u_i,u_{i+1})$. If $(u_i,u_{i+1})\in H$, then by the stretch guarantee of the spanner $(S,D)$, we have $d_R(u_i,u_{i+1})\leq t\cdot d_G(u_i,u_{i+1})$. See Figure \ref{fig:PreservingHopsets}.

\begin{center}
\begin{figure}[ht!]
    \centering
    \includegraphics[width=14cm, height=4cm]{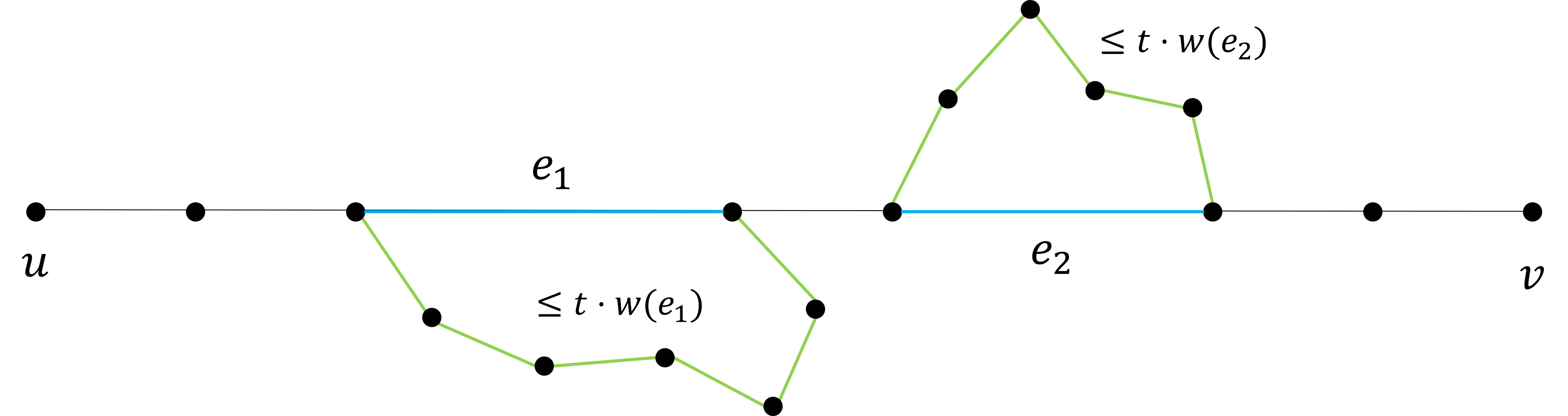}
    \caption{The path $P_{u,v}$, which consists of at most $\beta$ edges from $G\cup H$. The hop-edges $e_1$ and $e_2$, that are colored in blue, have a path in the pairwise spanner $S$ (depicted by green edges), with stretch at most $t$. Thus, when using the edge of $P_{u,v}\cap E$, and the edges of $S$, we get a path in $R=S\cup\bigcup_{(u,v)\in\mathcal{P}}(P_{u,v}\cap E)$ with weight at most $t\cdot w(P_{u,v})$. This weight, by the stretch guarantee of the hopset $H$, is at most $\alpha\cdot t\cdot d_G(u,v)$.}
    \label{fig:PreservingHopsets}
\end{figure}
\end{center}

In conclusion, we saw that $d_R(u_i,u_{i+1})\leq t\cdot d_G(u_i,u_{i+1})$ for every $i\in[0,b-1]$. Hence,
\[d_R(u,v)\leq\sum_{i=0}^{b-1}d_R(u_i,u_{i+1})\leq\sum_{i=0}^{b-1}t\cdot d_G(u_i,u_{i+1})\leq t\alpha\cdot d_G(u,v)~.\]
Therefore, the stretch of our pairwise spanner $R$ is at most $t\alpha$. 

The oracle $D'$, given a pair $(u,v)\in\mathcal{P}$, reports the edges of $P_{u,v}$ one by one. In case it encounters an edge of $H$, it uses the oracle $D$ to replace it with a path in $S$ with stretch at most $t$. Notice that this output path is contained in $R$, and by a similar proof as above, it has stretch at most $t\alpha$. The running time of this procedure is proportional to the length of the resulting path. Hence, the query time of our pairwise spanner is $O(1)$.

For the size of $R$, we have
\[|R|\leq|S|+\sum_{(u,v)\in\mathcal{P}}|P_{u,v}|\leq Q(n,M(n))+\sum_{(u,v)\in\mathcal{P}}\beta=\beta\cdot|\mathcal{P}|+Q(n,M(n))~.\]
The same bound holds for the oracle $D'$.

\end{proof}

To achieve path-reporting pairwise spanners with large stretch and small size, we use the best known hopsets that have large stretch. These hopsets were presented in \cite{BP20}, and were later reformulated and improved in \cite{NS22}. Some of the results of \cite{NS22} can be summarized by the following theorem\footnote{For stretch $3+\epsilon$, \cite{NS22} provides a hopset with hopbound $O\left(k^{\log_2(3+\frac{16}{\epsilon})}\right)$ and size $O\left(n\log k+n^{1+\frac{1}{k}}\right)$. We focus here, however, on hopsets that have larger stretch and smaller hopbound.}.

\begin{theorem}[\cite{NS22}] \label{thm:HopsetsWithLargeStretch}
Let $G=(V,E)$ be an undirected weighted graph on $n$ vertices, and let $1<c\leq k$ be two integer parameters. There is a hopset $H$ for $G$ with stretch $8c+3$, hopbound $O\left(\frac{1}{c}\cdot k^{1+\frac{2}{\ln c}}\right)$, and size $O\left(c\log_ck\cdot n+cn^{1+\frac{1}{k}}\right)$.
\end{theorem}

In addition, we use the state-of-the-art path-reporting pairwise spanner with stretch $3+\epsilon$, from Theorem \ref{thm:Near3PairwiseSpanner}. Lemma \ref{lemma:CombineHopsetSpanner}, when applied with Theorems \ref{thm:HopsetsWithLargeStretch},\ref{thm:Near3PairwiseSpanner}, provides the following result.

\begin{theorem} \label{thm:PairwiseSpanner1}
Let $G=(V,E)$ be an undirected weighted graph on $n$ vertices, and let $\mathcal{P}\subseteq V^2$ be a set of pairs of vertices in $G$. For every two integer parameter $1<c\leq k$, there is a path-reporting $\mathcal{P}$-pairwise $O(c)$-spanner for $G$, with query time $O(1)$ and size 
\[O\left(|\mathcal{P}|\cdot\frac{k^{1+\frac{2}{\ln c}}}{c}+ck^9\log_ck\cdot n+ck^9\cdot n^{1+\frac{1}{k}}\right)~.\]
\end{theorem}

\begin{proof}

In Theorem \ref{thm:Near3PairwiseSpanner}, choose $\epsilon=40$. By Lemma \ref{lemma:CombineHopsetSpanner}, applied with Theorems \ref{thm:HopsetsWithLargeStretch},\ref{thm:Near3PairwiseSpanner}, we get a path-reporting pairwise spanner with stretch $(3+\epsilon)\cdot(8c+3)=344c+129=O(c)$, and size
\begin{eqnarray*}
&&O\left(\frac{1}{c}\cdot k^{1+\frac{2}{\ln c}}\right)\cdot|\mathcal{P}|+O\left((c\log_ck\cdot n+cn^{1+\frac{1}{k}})\cdot k^{\log_{4/3}(12+\frac{40}{\epsilon})}+n\log k+n^{1+\frac{1}{k}}\right)\\
&=&O\left(|\mathcal{P}|\cdot\frac{k^{1+\frac{2}{\ln c}}}{c}+ck^9\log_ck\cdot n+ck^9\cdot n^{1+\frac{1}{k}}\right)~.
\end{eqnarray*}

\end{proof}

\subsection{Path-Reporting Pairwise Spanner with Improved Size}

In this section we further improve the size of the path-reporting pairwise spanner from Theorem \ref{thm:PairwiseSpanner1}, by decreasing the coefficient $k^9$ of $c\log_ck\cdot n$ and of $cn^{1+\frac{1}{k}}$. To do that, we first prove an extended version of Theorem \ref{thm:HopsetsWithLargeStretch}. In this version, we observe that when considering the hopset from \cite{NS22} as a set of pairs of vertices, a large part of this set is easy to preserve using shortest paths. That is, for this subset of pairs, there is a path-reporting \textit{preserver} with small size. For the other pairs, we will use a path-reporting pairwise spanner with small stretch, as we did in Lemma \ref{lemma:CombineHopsetSpanner}.

\begin{theorem}[Extended version of Theorem \ref{thm:HopsetsWithLargeStretch}] \label{thm:HopsetsExtendedVersion}
Let $G=(V,E)$ be an undirected weighted graph on $n$ vertices, and let $1<c\leq k$ be two integer parameters. %For any real parameter $0<\delta\leq1$, 
There is a hopset $H$ for $G$ with stretch $8c+3$, hopbound $O\left(\frac{1}{c}\cdot k^{1+\frac{2}{\ln c}}\right)$, and size $O\left(cn\log_ck+ck^{\frac{9}{c-1}}n^{1+\frac{1}{k}}\right)$.
%and size $O\left(c\cdot(n\log_ck+\frac{1}{\delta}n^{1+\frac{1}{k}})\right)$.

Moreover, the hopset $H$ can be divided into three disjoint subsets $H=H_1\cup H_2\cup H_3$, such that $H_1$ has a path-reporting preserver with size $O(cn\log_ck)$, $H_2$ has a path-reporting preserver with size $O\left(ck^{\frac{9}{c-1}}n^{1+\frac{1}{k}}\right)$, and 
\[|H_3|=O\left(ck^{-9}n^{1+\frac{1}{k}}\right)~.\]

%Moreover, the hopset $H$ can be divided into three disjoint subsets $H=H_1\cup H_2\cup H_2$, such that the support size of $H_1$ is $O(cn\log_ck)$, the support size of $H_2$ is $O\left(\frac{c}{\delta}n^{1+\frac{1}{k}}\right)$, and 
%\[|H_3|=O\left(c\delta^{c-1}n^{1+\frac{1}{k}}\right)~.\]
\end{theorem}

Note that Theorem \ref{thm:HopsetsExtendedVersion} provides a path-reporting preserver only for $H_1,H_2$, while for $H_3$, it only bounds the size. The proof of this theorem involves diving into the details of the proof of Theorem \ref{thm:HopsetsWithLargeStretch} from \cite{NS22}. We bring the full proof in Appendix \ref{sec:ProofExtendedVersion}.

Given Theorem \ref{thm:HopsetsExtendedVersion}, we use the same scheme as in the proof of Lemma \ref{lemma:CombineHopsetSpanner} to produce a path-reporting pairwise spanner. This time, we use the path-reporting pairwise spanner from Theorem \ref{thm:Near3PairwiseSpanner} only on $H_3$. For $H_1,H_2$, we simply use their path-reporting preserver from Theorem \ref{thm:HopsetsExtendedVersion}.

\begin{theorem} \label{thm:PairwiseSpanner2}
Let $G=(V,E)$ be an undirected weighted graph on $n$ vertices, and let $\mathcal{P}\subseteq V^2$ be a set of pairs of vertices in $G$. For every two integer parameter $1<c\leq k$, there is a path-reporting $\mathcal{P}$-pairwise $O(c)$-spanner for $G$ with size 
\[O\left(|\mathcal{P}|\cdot\frac{k^{1+\frac{2}{\ln c}}}{c}+cn\log_ck+ck^{\frac{9}{c-1}}n^{1+\frac{1}{k}})\right)~.\]
\end{theorem}

\begin{proof}

Let $H$ be the hopset from Theorem \ref{thm:HopsetsExtendedVersion}, and let $H_1,H_2,H_3$ be the corresponding partition of $H$. Let $(S,D)$ be the path-reporting pairwise spanner from Theorem \ref{thm:Near3PairwiseSpanner} for the graph $G$ and the set of vertices pairs $H_3$. For $(S,D)$, we choose the same parameter $k$ as for the hopset $H$. The parameter $\epsilon$ is chosen as $\epsilon=40$, so that $(S,D)$ has stretch $43$ and size
\begin{eqnarray*}
O\left(|H_3|\cdot k^{\log_{4/3}13}+n\log k+n^{1+\frac{1}{k}}\right)&=&O\left(ck^{-9}n^{1+\frac{1}{k}}\cdot k^9+n\log k+n^{1+\frac{1}{k}}\right)\\
%&=&O\left(ck^{-9}n^{1+\frac{1}{k}}\cdot k^9+n\log k+n^{1+\frac{1}{k}}\right)\\
&=&O\left(n\log k+cn^{1+\frac{1}{k}}\right)~.
\end{eqnarray*}

Let $(E_1,D_1),(E_2,D_2)$ be the path-reporting preservers for $H_1,H_2$. For every pair $(u,v)\in\mathcal{P}$, let $P_{u,v}$ be the $u-v$ path in $G\cup H$ that has at most $\beta=O\left(\frac{1}{c}\cdot k^{1+\frac{2}{\ln c}}\right)$ edges and stretch at most $\alpha=8c+3$. We define a path-reporting $\mathcal{P}$-pairwise spanner $(R,D')$ for $G$ by 
\[R=E_1\cup E_2\cup S\cup\bigcup_{(u,v)\in\mathcal{P}}(P_{u,v}\cap E)~.\]
The oracle $D'$ stores the oracles $D_1,D_2,D$, and all the paths $P_{u,v}$, for every $(u,v)\in\mathcal{P}$.

Given a query $(u,v)\in\mathcal{P}$, the oracle $D'$ restores the path $P_{u,v}$, and reports its edges one by one. Whenever it encounters an edge of $H_1$ or $H_2$, it uses the oracle $D_1$ or $D_2$ respectively, to replace this edge with a shortest path, and it reports its edges. For an edges of $H_3$, it uses the oracle $D$ for the same purpose. In this case the resulting path might have stretch at most $43$.

By a similar analysis to the one in Lemma \ref{lemma:CombineHopsetSpanner}, we get a $u-v$ path in $R$ that has stretch at most $43\cdot(8c+3)=O(c)$. The time required to report this path is proportional to the length of this output path.

Lastly, we use Theorem \ref{thm:HopsetsExtendedVersion} to conclude that the size of $R$ is at most
\begin{eqnarray*}
&&|E_1|+|E_2|+|S|+\sum_{(u,v)\in\mathcal{P}}|P_{u,v}|\\
&=&O\left(cn\log_ck+ck^{\frac{9}{c-1}}n^{1+\frac{1}{k}}+n\log k+cn^{1+\frac{1}{k}}+\sum_{(u,v)\in\mathcal{P}}\beta\right)\\
&=&O\left(cn\log_ck+ck^{\frac{9}{c-1}}n^{1+\frac{1}{k}}+n\log k+cn^{1+\frac{1}{k}}+|\mathcal{P}|\beta\right)\\
&=&O\left(|\mathcal{P}|\beta+cn\log_ck+ck^{\frac{9}{c-1}}n^{1+\frac{1}{k}}\right)~.
\end{eqnarray*}
In the last step, we used the fact that $\log_2k\leq c\log_ck$ for every $1<c,k$. The size of the oracle $D'$ has the same bound.

\end{proof}

For a constant $c$, we get the following result.

\begin{corollary} \label{cor:ConstantStretch}
For every \textit{constant} $c$, and $k\geq c$, there is a path-reporting $\mathcal{P}$-pairwise spanner for $G$, with stretch $O(c)$ and with size
\[O\left(|\mathcal{P}|\cdot k^{1+\frac{2}{\ln c}}+k^{\frac{9}{c-1}}n^{1+\frac{1}{k}}\right)~.\]
\end{corollary}

Now let $c=k^\epsilon$, for some $0<\epsilon<1$. If we set $\epsilon\geq\frac{\log\log k}{\log k}$, then we have $k^\epsilon\geq\log k$, and then $k^{O(\frac{1}{c})}=k^{O(\frac{1}{\log k})}=O(1)$. By Theorem \ref{thm:PairwiseSpanner2}, we get the following result.

\begin{corollary} \label{cor:KEpsStretch}
For every integer $k>1$ and $\epsilon\geq\frac{\log\log k}{\log k}$, there is a path-reporting $\mathcal{P}$-pairwise spanner for $G$, with stretch $O(k^\epsilon)$ and with size
\[O\left(|\mathcal{P}|\cdot k^{1-\epsilon}\cdot e^{\frac{2}{\epsilon}}+\epsilon^{-1}k^\epsilon\cdot n+k^\epsilon\cdot n^{1+\frac{1}{k}}\right)~.\]
\end{corollary}

Lastly, for $c=\log k$, we get from Theorem \ref{thm:PairwiseSpanner2}:

\begin{corollary} \label{cor:LogKStretch}
For every integer $k>1$ there is a path-reporting $\mathcal{P}$-pairwise spanner for $G$, with stretch $O(\log k)$ and with size
\[O\left(|\mathcal{P}|\cdot\frac{k^{1+\frac{2}{\log\log k}}}{\log k}+\frac{\log^2k}{\log\log k}\cdot n+\log k\cdot n^{1+\frac{1}{k}}\right)~.\]
\end{corollary}

\section{Subset, Source-wise and Prioritized Spanners} \label{sec:PRDOTypes}

In this section we describe several reductions between the types of path-reporting spanners that were described in Section \ref{sec:PRDOTypesIntro}. Using these reductions, we obtain a variety of results on these spanners. %, as an application of our new path-reporting pairwise spanners from Section \ref{sec:UpperBounds}, and of the state-of-the-art path-reporting pairwise spanners from \cite{ES23}. 
We start by describing the reductions in a robust way, such that the results themselves will be obtained by substituting specific properties. %The first reduction, from pairwise spanners to subset spanners, also uses a path-reporting \textit{emulator} (see definition \ref{def:PathReportingEmulator}). The same reduction appears implicitly in \cite{ES23}.

\subsection{Subset and Source-wise Spanners} \label{sec:SubsetSourcewiseSpanners}

The first reduction, from pairwise spanners to subset spanners, also uses a path-reporting \textit{emulator} (see definition \ref{def:PathReportingEmulator}). The same reduction appears implicitly in \cite{ES23}.

\begin{lemma} \label{lemma:PairwiseToSubset}
Suppose that every undirected weighted $n$-vertex graph $G=(V,E)$ admits a path-reporting $\alpha_E$-emulator with query time $q(n)$ and size $M_E(n)$. In addition, suppose that any such graph and any set of pairs $\mathcal{P}\subseteq V^2$ admit a path-reporting $\mathcal{P}$-pairwise $\alpha_P$-spanner with size $M_P(n,|\mathcal{P}|)$.

Then, for every undirected weighted $n$-vertex graph $G=(V,E)$ and a subset $A\subseteq V$, there is a path-reporting $A$-subset $\alpha_P\cdot\alpha_E$-spanner with query time $O(q(|A|))$ and size $M_E(|A|)+M_P(n,M_E(|A|))$.
\end{lemma}

\begin{proof}

Let $G=(V,E)$ be an undirected weighted $n$-vertex graph, and let $A\subseteq V$ be a subset. Consider the graph $K=(A,\binom{A}{2})$ with weights $w(x,y)=d_G(x,y)$ on the edges. By our assumption, $K$ has a path-reporting $\alpha_E$-emulator with query time $q(|A|)$ and size $M_E(|A|)$. Denote this emulator by $R_E\subseteq\binom{A}{2}\subseteq\binom{V}{2}$ and its corresponding oracle by $D_E$. We consider $R_E$ as a set of pairs $R_E\subseteq V^2$ by orienting each edge in $R_E$ in an arbitrary direction (alternatively, we may view each edge as two pairs, one in each direction). %\alert{Can't we just define ${\cal P}\subseteq\binom{V}{2}$ in the introduction and everywhere?}

The graph $G$ and the set of pairs $R_E$ admit a path-reporting $R_E$-pairwise $\alpha_P$-spanner with size $M_P(n,|R_E|)\leq M_P(n,M_E(|A|))$. Denote this pairwise spanner by $R_P$ and its corresponding oracle by $D_P$. See Figure \ref{fig:PairwiseToSubset} for an illustration.

\begin{center}
\begin{figure}[ht!]
    \centering
    \includegraphics[width=14cm, height=5cm]{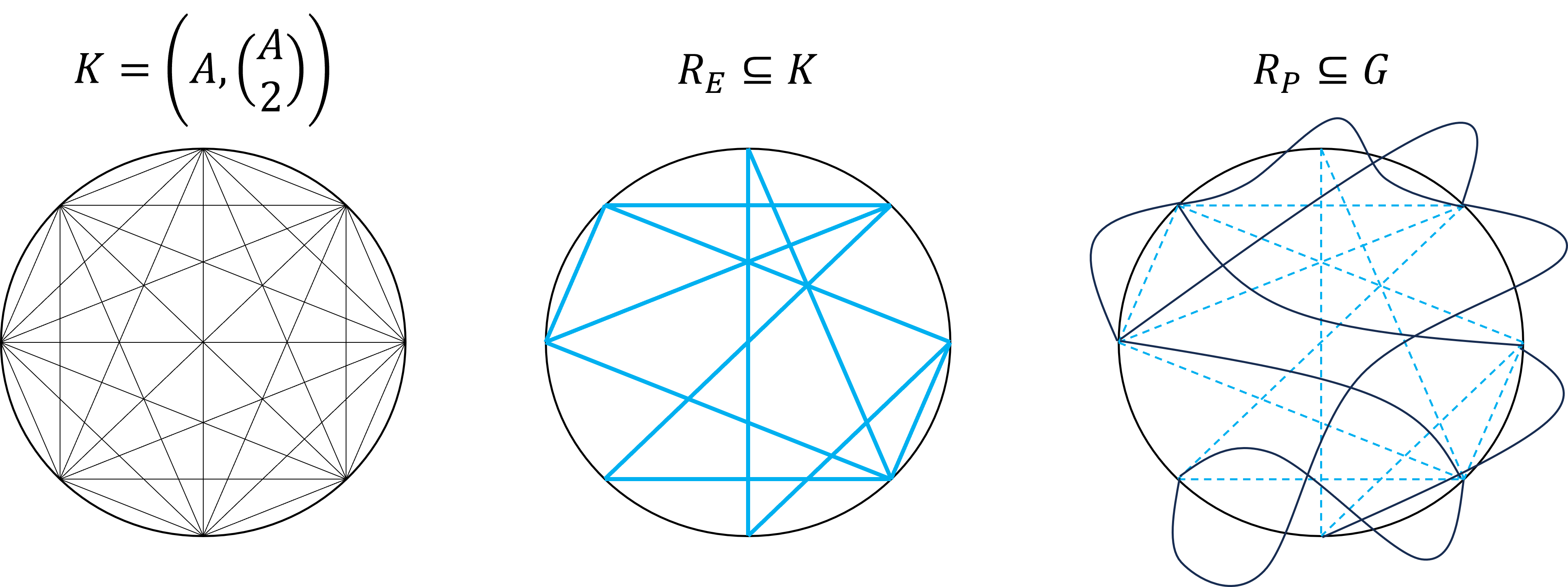}
    \caption{The graph $K$ is a complete graph on the vertices of the set $A$, where the weights of its edges represent the shortest path distance in $G$ between their ends. Then, a path-reporting emulator is applied on $K$, resulting in a small subgraph $R_E$ of $K$ (depicted in blue), that preserve distance up to stretch $\alpha_E$. Lastly, a path-reporting pairwise spanner is applied on the set of edges of $R_E$. In the resulting spanner $R_P$, each emulator edge is replaced by an actual path in the graph $G$ (depicted by a curvy line), with stretch at most $\alpha_P$.}
    \label{fig:PairwiseToSubset}
\end{figure}
\end{center}

We now define an $A$-subset spanner as follows. The spanner itself is the subgraph $R_P$ of $G$. The corresponding oracle $D$ contains the oracles $D_E$ and $D_P$. Given a pair of vertices $(u,v)\in A^2$, the oracle $D$ first finds a $u-v$ path $Q_{u,v}$ in $R_E$, with stretch $\alpha_E$, using the oracle $D_E$. Then, $D$ applies the oracle $D_P$ on every edge of $Q_{u,v}$, and concatenates the resulting paths. The result is a tour between $u,v$ in $R_P$ with weight at most $\alpha_P\cdot\alpha_E\cdot d_K(u,v)=\alpha_P\cdot\alpha_E\cdot d_G(u,v)$, i.e., stretch at most $\alpha_P\cdot\alpha_E$.

The query time of $D$ consists of the time required to run a single query of $D_E$, then query $D_P$ on every edge of the resulting path. The first part is $q(|A|)$, while each query in the second part runs in time that is proportional to the length of the resulting path. Since $D$ concatenates all these paths, this running time is overall proportional to the length of the resulting path, thus we get a query time of $q(|A|)$.

The size of the oracle $D$ is the sum of the sizes of $D_E,D_P$, which is $M_E(|A|)+M_P(n,M_E(|A|))$. The size of the spanner $R_P$ itself is $M_P(n,M_E(|A|))$. This completes our proof.

\end{proof}

Before we apply the reduction from Lemma \ref{lemma:PairwiseToSubset} to get subset spanners, we first show a very simple reduction from subset spanners to source-wise spanners. The resulting source-wise spanner from this reduction has a stretch that is twice as the stretch of the original subset spanner. At the same time, its query time and size are asymptotically the same as in the original subset spanner. This reduction appears implicitly in \cite{EFN15}.

\begin{lemma} \label{lemma:SubsetToSourcewise}
Suppose that the undirected weighted $n$-vertex graph $G=(V,E)$ admits a path-reporting $A$-subset $\alpha$-spanner with query time $q$ and size $M(n,|A|)$. Then, there is a path-reporting $A$-source-wise $(2\alpha+1)$-spanner with query time $O(q)$ and size $M(n,|A|)+O(n)$.
\end{lemma}

\begin{proof}

Let $(S_0,D_0)$ be the given path-reporting $A$-subset $\alpha$-spanner. For every vertex $v\in V$, denote by $p(v)$ the closest vertex to $v$ in the set $A$, breaking ties arbitrarily but consistently. Denote by $P_{v,p(v)}$ a shortest path between $v,p(v)$ in $G$. We define a new path-reporting $A$-source-wise spanner $(S,D)$ as follows.

The set $S$ is defined to be the union $S=S_0\cup\bigcup_{v\in V}P_{v,p(v)}$. The oracle $D$ stores the oracle $D_0$ and all the shortest paths $P_{v,p(v)}$, for every $v\in V$, together with a pointer $q(v)$ from $v$ to the next vertex after $v$ on the path $P_{v,p(v)}$. Given a query $(v,a)\in V\times A$, the oracle $D$ uses the pointers $q()$ to find the path $P_{v,p(v)}$. Then, it concatenates this path with the $p(v)-a$ path $Q_{p(v),a}$ that is returned by $D_0$. The oracle $D$ returns the resulting path as an output.

For the query algorithm to be well-defined, we need to show that indeed $P_{v,p(v)}$ can be found using only the pointers $q()$. This is true because for every vertex $x$ on the path $P_{v,p(v)}$, we have $p(x)=p(v)$ (since if there was a closer $A$-vertex to $x$ than $p(v)$, it would be also closer to $v$). Then, the path $P_{x,p(x)}=P_{x,p(v)}$ is a sub-path of $P_{v,p(v)}$, and the pointer $q(x)$ points at the next vertex on this path in the direction to $p(v)$. Hence, the oracle $D$ can use these pointers to obtain the path $P_{v,p(v)}$ entirely.

The path $P_{v,p(v)}\circ Q_{p(v),a}$ that $D$ outputs is a path in $\bigcup_{v\in V}P_{v,p(v)}\cup S_0=S$. Note that since $p(v)$ is the closest vertex to $v$ in $A$, we have $w(P_{v,p(v)})=d_G(v,p(v))\leq d_G(v,a)$. In addition, recall that the stretch of $(S_0,D_0)$ is $\alpha$, thus we have $w(Q_{p(v),a})\leq\alpha\cdot d_G(p(v),a)$. By the triangle inequality,
\[w(Q_{p(v),a})\leq\alpha\cdot d_G(p(v),a)\leq\alpha(d_G(p(v),v)+d_G(v,a))\leq2\alpha\cdot d_G(v,a)~.\]
We conclude that the weight of the resulting path is
\[w(P_{v,p(v)})+w(Q_{p(v),a})\leq d_G(v,a)+2\alpha\cdot d_G(v,a)=(2\alpha+1)d_G(v,a)~.\]

\begin{center}
\begin{figure}[ht!]
    \centering
    \includegraphics[width=7cm, height=4cm]{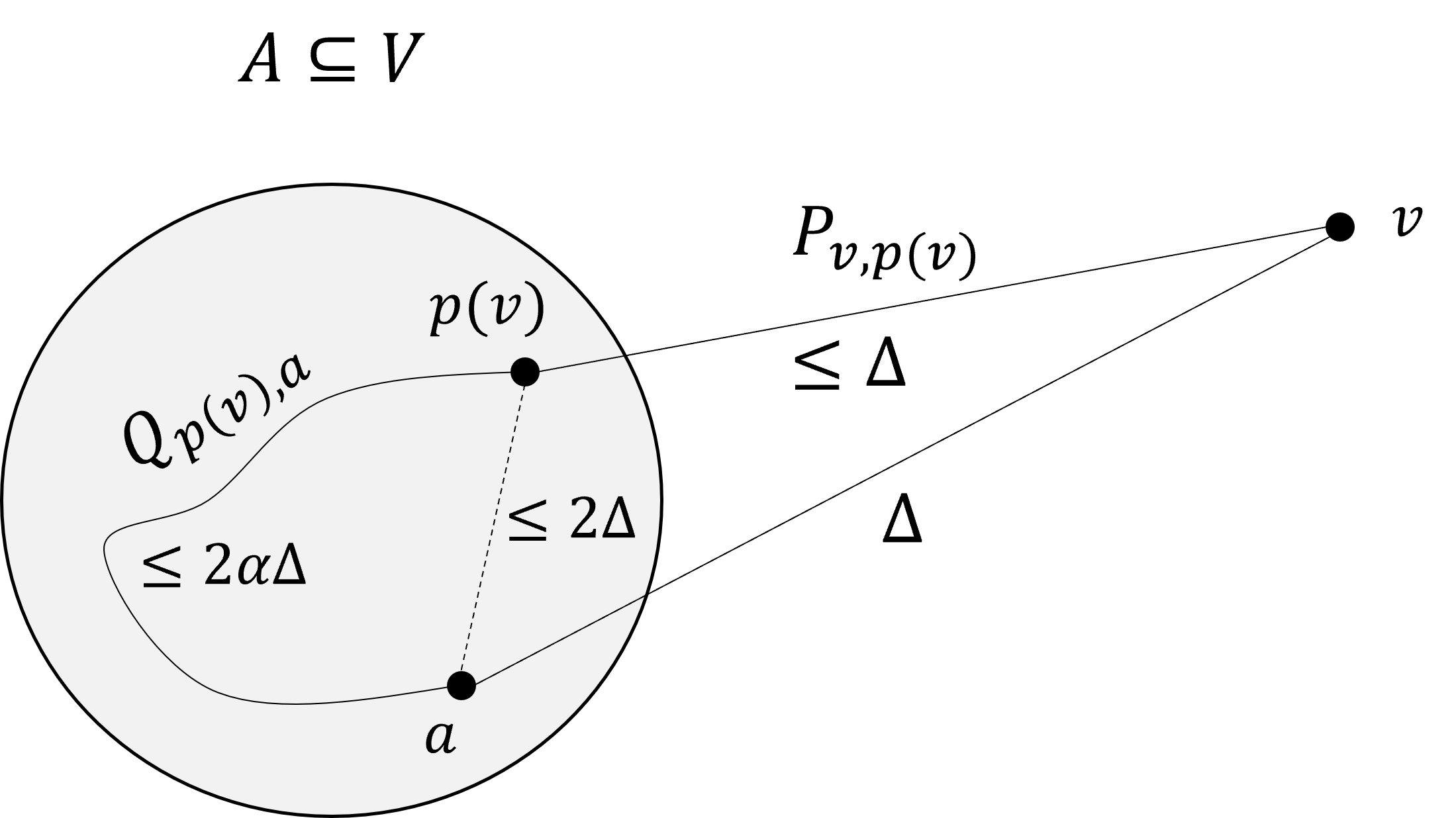}
    \caption{Given a vertex $v\in V$, the vertex $p(v)$ is the closest vertex to $v$ from the subset $A$. For a vertex $a\in A$, the distance $d_G(v,a)$ appears in the figure as $\Delta$. Since $p(v)$ is not farther from $v$ than $a$, we know that $d_G(v,p(v))\leq\Delta$. By the triangle inequality, $d_G(p(v),a)\leq\Delta+\Delta=2\Delta$. Therefore, the path $Q_{p(v),a}$ in the $A$-subset spanner, has weight at most $2\alpha\Delta$. To recover the path $P_{v,p(v)}$, we use the stored pointers in the oracle $D$. To find $Q_{p(v),a}$, we use the oracle $D_0$ of the $A$-subset spanner.}
    \label{fig:SubsetToSourcewise}
\end{figure}
\end{center}

The running time of finding $P_{v,p(v)}$ is proportional to the length of $P_{v,p(v)}$ itself. Obtaining the path $Q_{p(v),a}$ from $D_0$ requires $q+O(|Q_{p(v),a}|)$ time. Therefore, by our convention of measuring query time for path-reporting spanners, the total query time of the oracle $D$ is $O(q)$.

We now bound the size of the path-reporting source-wise spanner $(S,D)$. The oracle $D$ stores the oracle $D_0$, which requires $M(n,|A|)$ size. Aside from $D_0$, the oracle $D$ stores a pointer for every vertex in $V$, thus the size of $D$ is $M(n,|A|)+O(n)$. Notice that every edge in the set $\bigcup_{v\in V}P_{v,p(v)}$ is of the form $(x,q(x))$, for some vertex $x\in V$. We conclude that the size of $S$ is at most $|S_0|+n\leq M(n,|A|)+n=M(n,|A|)+O(n)$.

This completes the proof that $(S,D)$ is a path-reporting $|A|$-source-wise $(2\alpha+1)$-spanner with query time $O(q)$ and size $M(n,|A|)+O(n)$.

\end{proof}

We now use the reductions from Lemma \ref{lemma:PairwiseToSubset} and Lemma \ref{lemma:SubsetToSourcewise} to produce path-reporting subset and source-wise spanners. We apply Lemma \ref{lemma:PairwiseToSubset} on the emulators from Theorems \ref{thm:TZEmulator} and \ref{thm:MNEmulator} and on the pairwise spanners from Theorems \ref{thm:NearExactPairwiseSpanner} and \ref{thm:Near3PairwiseSpanner}. The results are given in Theorems \ref{thm:SubsetSpanner1}, \ref{thm:SubsetSpanner2} and \ref{thm:SubsetSpanner3}.

\pagebreak

\begin{theorem} \label{thm:SubsetSpanner1}
Let $G=(V,E)$ be an undirected weighted $n$-vertex graph and let $A\subseteq V$ be a subset of vertices. Let $k\in[1,\log n]$ be an integer parameter and let $0<\epsilon\leq1$ be a real parameter. Denote $q=\log_n(|A|^{1+\frac{1}{k}})$. Then, there is a path-reporting $A$-subset $(2+\epsilon)k$-spanner, with query time $O(\log k)$ and size 
\[O\left(k|A|^{1+\frac{1}{k}}\cdot\beta_1(\epsilon,p)+n\cdot\left(\log p+\log\log(\frac{1}{\epsilon})\right)\right)~.\]
Here, $\beta_1(\epsilon,p)=O(\frac{\log p}{\epsilon})^{(1+o(1))\log_2p+\log_{4/3}\log(\frac{1}{\epsilon})}$, and
\[p=\begin{cases}
    \frac{1}{q-1} & q>1 \\
    \log n & q\leq1~.
\end{cases}\]

In addition, there is a path-reporting $A$-source-wise $(4+\epsilon)k$-spanner, with query time $O(\log k)$ and the same size.
%\[O\left(k|A|^{1+\frac{1}{k}}\cdot\beta_1(\frac{\epsilon}{4},p)+n\cdot\left(\log p+\log\log(\frac{1}{\epsilon})\right)\right)~.\]

\end{theorem}

\begin{proof}

Apply Lemma \ref{lemma:PairwiseToSubset} on the path-reporting pairwise spanner from Theorem \ref{thm:NearExactPairwiseSpanner}, with parameters $\frac{\epsilon}{2},p$ instead of $\epsilon,k$, and on the path-reporting emulator from Theorem \ref{thm:TZEmulator}. This provides a path-reporting $A$-subset spanner, with stretch $\alpha=(1+\frac{\epsilon}{2})(2k-1)<(2+\epsilon)k$, query time $O(\log k)$ and size\footnote{In the following expression for the size of the subset spanner, $\beta_1(\epsilon,p)$ appears instead of $\beta_1(\frac{\epsilon}{2},p)$. This is due to the fact that
\[\beta_1(\frac{\epsilon}{2},p)<O(\frac{2\log p}{\epsilon})^{(1+o(1))\log_2p+\log_{4/3}\log(\frac{1}{\epsilon})+2.5}=O(\frac{\log p}{\epsilon})^{(1+o(1))\log_2p+\log_{4/3}\log(\frac{1}{\epsilon})}~,\]
since the $2.5$ is absorbed by the $o(1)\cdot\log_2p$. Thus, $\beta_1(\frac{\epsilon}{2},p)$ has the same bound as $\beta_1(\epsilon,p)$.}
\begin{equation} \label{eq:SubsetSpannerSize1}
O\left(k|A|^{1+\frac{1}{k}}\cdot\beta_1(\epsilon,p)+n\cdot\left(\log p+\log\log(\frac{1}{\epsilon})\right)+n^{1+\frac{1}{p}}\right)~.
\end{equation}
If $q>1$, then $n^{1+\frac{1}{p}}=n^q=|A|^{1+\frac{1}{k}}$. Thus, the third term in (\ref{eq:SubsetSpannerSize1}) is smaller than the first term. If, on the other hand, $q\leq1$, then this third term is $n^{1+\frac{1}{p}}=n^{1+\frac{1}{\log n}}=2n$, which is smaller than the second term in (\ref{eq:SubsetSpannerSize1}). Hence, we get the desired size for the resulting subset spanner.

By Lemma \ref{lemma:SubsetToSourcewise}, there is also a path-reporting $A$-source-wise spanner with stretch $2\alpha+1$, query time $O(\log k)$ and size 
\begin{eqnarray*}
&&O\left(k|A|^{1+\frac{1}{k}}\cdot\beta_1(\epsilon,p)+n\cdot\left(\log p+\log\log(\frac{1}{\epsilon})\right)\right)+O(n)\\
&=&O\left(k|A|^{1+\frac{1}{k}}\cdot\beta_1(\epsilon,p)+n\cdot\left(\log p+\log\log(\frac{1}{\epsilon})\right)\right)~.
\end{eqnarray*}

Notice that $\alpha=(1+\frac{\epsilon}{2})(2k-1)=(2+\epsilon)k-(1+\frac{\epsilon}{2})$ and therefore 
\[2\alpha+1=2(2+\epsilon)k-2(1+\frac{\epsilon}{2})+1<(4+2\epsilon)k~.\]
Replacing $\epsilon$ by $\frac{\epsilon}{2}$ in this source-wise spanner gives the desired result.

%\alert{We should probably replace also $\epsilon/2$ by $\epsilon$ everywhere in $\beta_1$, $\beta_2$, since there is always $O(\cdot)$ before it.} \alert{Idan: Right. Also the footnote with the $\beta_1(\frac{\epsilon}{4},p)=O(\beta_1(\frac{\epsilon}{2},p))$ was not accurate.}

\end{proof}

The size of the subset spanner from Theorem \ref{thm:SubsetSpanner1} can be improved at the cost of multiplying the stretch by $3$. The proof of the following theorem is identical to the proof of Theorem \ref{thm:SubsetSpanner1}, when using the pairwise spanner from Theorem \ref{thm:Near3PairwiseSpanner} instead of the one from Theorem \ref{thm:NearExactPairwiseSpanner}.

\begin{theorem} \label{thm:SubsetSpanner2}
Let $G=(V,E)$ be an undirected weighted $n$-vertex graph and let $A\subseteq V$ be a subset of vertices. Let $k\in[1,\log n]$ be an integer parameter and let $0<\epsilon\leq40$ be a real parameter. Denote $q=\log_n(|A|^{1+\frac{1}{k}})$. Then, there is a path-reporting $A$-subset $(6+\epsilon)k$-spanner, with query time $O(\log k)$ and size 
\[O\left(k|A|^{1+\frac{1}{k}}\cdot\beta_2(\epsilon,p)+n\log p\right)~.\]
Here, $\beta_2(\epsilon,p)=p^{O\left(\log\frac{1}{\epsilon}\right)}$, and
\[p=\begin{cases}
    \frac{1}{q-1} & q>1 \\
    \log n & q\leq1~.
\end{cases}\]

In addition, there is a path-reporting $A$-source-wise $(12+\epsilon)k$-spanner, with query time $O(\log k)$ and the same size.
%\[O\left(k|A|^{1+\frac{1}{k}}\cdot\beta_2(\frac{\epsilon}{4},p)+n\log p\right)~.\]

\end{theorem}

\begin{proof}

By Lemma \ref{lemma:PairwiseToSubset}, Theorem \ref{thm:Near3PairwiseSpanner} with parameters $\frac{\epsilon}{2},p$ instead of $\epsilon,k$, and Theorem \ref{thm:TZEmulator}, there is a path-reporting $A$-subset spanner, with stretch $\alpha=(3+\frac{\epsilon}{2})(2k-1)<(6+\epsilon)k$, query time $O(\log k)$ and size
\[O\left(k|A|^{1+\frac{1}{k}}\cdot\beta_2(\epsilon,p)+n\log p+n^{1+\frac{1}{p}}\right)=O\left(k|A|^{1+\frac{1}{k}}\cdot\beta_2(\epsilon,p)+n\log p\right)~,\]
where similarly to the proof of Theorem \ref{thm:SubsetSpanner2}, the third term is smaller from the sum of the first two terms. Here, we replaced $\beta_2(\frac{\epsilon}{2},p)$ by $\beta_2(\epsilon,p)$, because $\beta_2(\frac{\epsilon}{2},p)=p^{O\left(\log\frac{2}{\epsilon}\right)}=p^{O\left(\log\frac{1}{\epsilon}\right)}$, and thus they both have the same bound.

By Lemma \ref{lemma:SubsetToSourcewise}, there is also a path-reporting $A$-source-wise spanner with stretch $2\alpha+1$, query time $O(\log k)$ and size 
\[O\left(k|A|^{1+\frac{1}{k}}\cdot\beta_2(\epsilon,p)+n\log p\right)+O(n)=O\left(k|A|^{1+\frac{1}{k}}\cdot\beta_2(\epsilon,p)+n\log p\right)~.\]
Notice that $\alpha_2=(3+\frac{\epsilon}{2})(2k-1)=(6+\epsilon)k-(3+\frac{\epsilon}{2})$ and therefore 
\[2\alpha_2+1=2(6+\epsilon)k-2(3+\frac{\epsilon}{2})+1<(12+2\epsilon)k~.\]
Replacing $\epsilon$ by $\frac{\epsilon}{2}$ in this source-wise spanner gives the desired result.

\end{proof}

If we allow even larger stretch, the size of the subset spanner further improves. This is achieved by using the path-reporting emulator from Theorem \ref{thm:MNEmulator}, rather then the one from Theorem \ref{thm:TZEmulator}. This also improves the query time of the resulting path-reporting subset spanner.

\begin{theorem} \label{thm:SubsetSpanner3}
Let $G=(V,E)$ be an undirected weighted $n$-vertex graph and let $A\subseteq V$ be a subset of vertices. Let $k\in[1,\log n]$ be an integer parameter. Denote $q=\log_n(|A|^{1+\frac{1}{k}})$. Then, there is a path-reporting $A$-subset $O(k)$-spanner, and a path-reporting $A$-source-wise $O(k)$-spanner, both with query time $O(1)$ and size 
\[O\left(|A|^{1+\frac{1}{k}}\cdot\beta_3(p)+n\log p\right)~,\]
where $\beta_3(p)=p^{\log_{4/3}13}<p^9$ and
\[p=\begin{cases}
    \frac{1}{q-1} & q>1 \\
    \log n & q\leq1~.
\end{cases}\]

\end{theorem}

\begin{proof}

Apply Lemma \ref{lemma:PairwiseToSubset} on the path-reporting pairwise spanner from Theorem \ref{thm:Near3PairwiseSpanner}, with parameters $\epsilon=40,k=p$, and on the path-reporting emulator from Theorem \ref{thm:MNEmulator}. This provides a path-reporting $A$-subset spanner, with stretch $(3+40)\cdot O(k)=O(k)$, query time $O(1)$ and size
\begin{equation} \label{eq:SubsetSpannerSize3}
O\left(|A|^{1+\frac{1}{k}}\cdot p^{\log_{4/3}(12+\frac{40}{\epsilon})}+n\log p+n^{1+\frac{1}{p}}\right)=O\left(|A|^{1+\frac{1}{k}}\cdot \beta_3(p)+n\log p+n^{1+\frac{1}{p}}\right)~.
\end{equation}
Similarly to the proof of Theorem \ref{thm:SubsetSpanner1}, the third term in (\ref{eq:SubsetSpannerSize3}) is smaller than the first term if $q>1$, and smaller than the second term if $q\leq1$. Hence, we get the desired size for the resulting subset spanner.

By Lemma \ref{lemma:SubsetToSourcewise}, there is also a path-reporting $A$-source-wise spanner with stretch\newline $2\cdot O(k)+1=O(k)$, query time $O(\log k)$ and size 
\[O\left(|A|^{1+\frac{1}{k}}\cdot \beta_3(p)+n\log p\right)+O(n)=O\left(|A|^{1+\frac{1}{k}}\cdot \beta_3(p)+n\log p\right)~.\]

\end{proof}

\subsection{Prioritized Spanners} \label{sec:PrioritizedSpanners}

In this section we describe yet another reduction, that turns our path-reporting source-wise spanners from Section \ref{sec:SubsetSourcewiseSpanners} into path-reporting prioritized spanners (see Definition \ref{def:PrioritizedSpanner}). This reduction implicitly appears in \cite{EFN15}. In the following lemma, we use the following notation from \cite{NS22}. Given a function $f:\mathbb{N}\rightarrow\mathbb{N}$, denote
\[f^{-1}(j)=\min\{i\in\mathbb{N}\;|\;f(i)\geq j\}~.\]
This notation is well-defined for every $j\leq\max Im(f)$, where $Im(f)$ is the image of $f$, i.e., the set of outputs of the function $f$.

\begin{lemma} \label{lemma:SourcewiseToPrioritized}
Let $G=(V,E)$ be an undirected weighted $n$-vertex graph. Suppose that every subset $A\subseteq V$ admits a path-reporting $A$-source-wise $\alpha_0(|A|)$-spanner with query time $q_0(|A|)$ and size $M(n,|A|)$. In addition, suppose that there is a path-reporting spanner $(S^*,D^*)$ for $G$, with stretch $\alpha^*$, query time $q^*$ and size $M^*(n)$.

Then, for every monotonically increasing function $f:\{1,2,...,T\}\rightarrow\{1,2,...,n\}$, there is a path-reporting prioritized spanner for $G$ with prioritized stretch 
\[\alpha(j)=\begin{cases}
    \alpha_0(f(f^{-1}(j))) & j\leq f(T) \\
    \alpha^* & j>f(T)
\end{cases}~,\]
prioritized query time 
\[q(j)=\begin{cases}
    q_0(f(f^{-1}(j))) & j\leq f(T) \\
    q^* & j>f(T)
\end{cases}~,\] 
and size $O(n)+M^*(n)+\sum_{i=1}^TM(n,f(i))$.
\end{lemma}

\begin{proof}

Let $(v_1,v_2,...,v_n)$ be the priority ranking of the vertices in $V$. Consider the sequence $\{A_i\}_{i=1}^T$ of monotonically increasing subsets of $V$, where 
\[A_i=\{v_1,v_2,...,v_{f(i)}\}~.\]
Let $(S_i,D_i)$ be a path-reporting $A_i$-source-wise spanner with stretch $\alpha_0(|A_i|)=\alpha_0(f(i))$, query time $q_0(|A_i|)=q_0(f(i))$ and size $M(n,|A_i|)=M(n,f(i))$. We define a path-reporting prioritized spanner $(S,D)$ as follows. Let $S=S^*\cup\bigcup_{i=1}^TS_i$. The oracle $D$ stores all the oracles $D_i$ for $i\in[1,T]$, the oracle $D^*$, and a table of values of the function $f^{-1}(j)$, for every $1\leq j\leq f(T)=\max Im(f)$. Given a query $(v_j,v_{j'})\in V^2$, where $j<j'$, the oracle $D$ consider two cases. If $j\leq f(T)$, then $D$ restores $i=f^{-1}(j)$ from its table, and apply the oracle $D_i$ on the query $(v_j,v_{j'})$. If $j> f(T)$, then $D$ applies the oracle $D^*$ on $(v_j,v_{j'})$. The resulting path is the output of $D$.

Fix $j\leq f(T)$ and let $i=f^{-1}(j)$. Recall that $D_i=D_{f^{-1}(j)}$ is a path-reporting $A_i$-source-wise spanner with stretch $\alpha_0(|A_i|)=\alpha_0(f(i))=\alpha_0(f(f^{-1}(j)))$ and query time $q_0(f(f^{-1}(j)))$. Note also that by definition, $f(f^{-1}(j))\geq j$, and therefore $v_j\in A_{f^{-1}(j)}$. Thus, the prioritized stretch of the oracle $D_i$ on the query $(v_j,v_{j'})$, where $j<j'$ and $j\leq f(T)$, is at most $\alpha_0(f(f^{-1}(j)))$, and the prioritized query time is at most $q_0(f(f^{-1}(j)))$.

For $j>f(T)$, the stretch and the query time for a query $(v_j,v_{j'})$, where $j<j'$, are $\alpha^*$ and $q^*$ respectively.

The storage of the oracle $D$ consists of the oracles $\{D_i\}$, for $1\leq i\leq T$ and the oracle $D^*$. The size of each $D_i$ is $M(n,|A_i|)=M(n,f(i))$ and the size of $D^*$ is $M^*(n)$. Note that the oracle $D$ also stores an $O(n)$-sized table for the values of $f^{-1}(j)$, for every $1\leq j\leq f(T)$. Hence, the total size of $D$ is $O(n)+M^*(n)+\sum_{i=1}^TM(n,f(i))$. The size of the spanner $S$ itself is at most $|S^*|+\sum_{i=1}^T|S_i|\leq M^*(n)+\sum_{i=1}^TM(n,f(i))$. This completes the proof of the lemma.

\end{proof}

To use Lemma \ref{lemma:SourcewiseToPrioritized}, we need to choose a function $f$, path-reporting source-wise spanners $(S_i,D_i)$, and a path-reporting spanner $(S^*,D^*)$. The source-wise spanners we use are from Theorems \ref{thm:SubsetSpanner1},\ref{thm:SubsetSpanner2} and \ref{thm:SubsetSpanner3}. Note that the size of each of these three source-wise spanners is
\[O\left(|A|^{1+\frac{1}{k}}\cdot\gamma+n\log p\right)~,\]
where $\gamma\in\{k\beta_1(\epsilon,p),k\beta_2(\epsilon,p),\beta_3(p)\}$. Here,
\[\beta_1(\epsilon,p)=O\left(\frac{\log p}{\epsilon}\right)^{(1+o(1))\log_2p+\log_{4/3}\log(\frac{1}{\epsilon})},\;\;\;\beta_2(\epsilon,p)=p^{\log_{4/3}(12+\frac{40}{\epsilon})},\;\;\;\beta_3(p)=p^{\log_{4/3}13}~.\]

We apply Lemma \ref{lemma:SourcewiseToPrioritized} on these source-wise spanners. Every choice of the function $f$ yields a path-reporting prioritized spanner with different properties. We now explicitly show the outcomes of several choices of this function. Many other choices that are not presented here can be made, and they might provide prioritized spanners with different properties.

For convenience, in the following theorem we use the notation $\delta_j=\frac{\log j}{\log n}$. That is, $\delta_j$ is the exponent $\delta$ such that $j=n^\delta$. This notation is useful when comparing our next result with the results of \cite{EFN15}. In \cite{EFN15}, {\em non-path-reporting} prioritized distance oracles were presented, that had prioritized stretch %\footnote{For the following values of the stretch, whenever the value exceeds $\log n$, the stretch is actually $\log n$.} 
$\min\{\frac{4}{1-\delta_j}-1,\log n\}, \min\{\frac{8}{1-\delta_j}-5,\log n\}$ %,\frac{4}{(1-\delta_j)^2}-5,4\cdot2^{\frac{1}{1-\delta_j}}-5$, 
with size $O(n\log n),O(n\log\log n)$%$,O(n\log\log\log n),O(n\log^*n)$
, respectively\footnote{\cite{EFN15} also provided additional non-path-reporting prioritized spanners with sizes $O(n\log\log\log n)$ and $O(n\log^*n)$ (and possibly could provide many more prioritized spanners with different sizes and stretch parameters). However, as we will see in the proofs of Theorems \ref{thm:PrioritizedSpanner1} and \ref{thm:PrioritizedSpanner2}, our constructions of {\em path-reporting} prioritized spanners rely on our source-wise spanners from Section \ref{sec:SubsetSourcewiseSpanners}, that have size $\Omega(n\log\log n)$. Thus, in this paper we cannot achieve smaller size while keeping the spanners path-reporting.}. For comparison, a {\em path-reporting} prioritized distance oracle of size $O(n\log n)$ was also given in \cite{EFN15}, but it had a worse prioritized stretch of $2\delta_j\cdot\log n-1$.

\begin{theorem} \label{thm:PrioritizedSpanner1}
For every undirected weighted $n$-vertex graph $G=(V,E)$, and a real parameter $\epsilon$ such that $\sqrt{\frac{\log\log n\cdot\log\log\log n}{\log n}}\leq\epsilon\leq\frac{1}{2}$, there are path-reporting prioritized spanners with the following trade-offs.
\begin{enumerate}
    \item size $O_\epsilon\left(\frac{n\log n}{\log\log\log n}\right)$, prioritized stretch $\alpha(j)=\begin{cases}
        (4+\epsilon)\left\lceil\frac{1}{1-\delta_j-\eta_1(n)}\right\rceil & j\leq\frac{1}{2}n^{1-\eta_1(n)} \\
        (4+\epsilon)\frac{\log n}{\log\log\log n} & j>\frac{1}{2}n^{1-\eta_1(n)}
    \end{cases}$\newline and prioritized query time $O\left(\log\alpha(j)\right)$, where\newline $\eta_1(n)=\frac{4\log\beta_1(\epsilon,\log n)}{\log n}=O_\epsilon\left(\frac{\log\log n\cdot\log\log\log n}{\log n}\right)$.
    
    \item size $O_\epsilon(n\log n)$, prioritized stretch $\alpha(j)=\begin{cases}
        (12+\epsilon)\left\lceil\frac{1}{1-\delta_j-\eta_2(n)}\right\rceil & j<\frac{1}{2}n^{1-\eta_2(n)} \\
        (4+\epsilon)\frac{\log n}{\log\log\log n} & j\geq\frac{1}{2}n^{1-\eta_2(n)}
    \end{cases}$\newline 
    and prioritized query time $O\left(\log\alpha(j)\right)$, where\newline $\eta_2(n)=\frac{4\log\beta_2(\epsilon,\log n)}{\log n}=O_\epsilon\left(\frac{\log\log n}{\log n}\right)$.
    
    \item size $O(n\log n)$, prioritized stretch $\alpha(j)=\begin{cases}
        O\left(\left\lceil\frac{1}{1-\delta_j-\eta_3(n)}\right\rceil\right) & j<\frac{1}{2}n^{1-\eta_3(n)} \\
        O\left(\frac{\log n}{\log\log n}\right) & j\geq\frac{1}{2}n^{1-\eta_3(n)}
    \end{cases}$\newline 
    and query time $O(1)$, where $\eta_3(n)=\frac{4\log\beta_3(\log n)}{\log n}=O\left(\frac{\log\log n}{\log n}\right)$.
\end{enumerate}
\end{theorem}

\begin{proof}

In what follows, we denote $\beta_1=\beta_1(\epsilon,\log n)$, $\beta_2=\beta_2(\epsilon,\log n)$ and $\beta_3=\beta_3(\log n)$. We also denote $a_1=4+\epsilon$, $a_2=12+\epsilon$ and $a_3=C$, where $C$ is the constant such that the source-wise spanner from Theorem \ref{thm:SubsetSpanner3} has stretch $Ck$. We prove the three items in the theorem all at once. To do so, we fix some $r\in\{1,2,3\}$.

Let $f:\{1,2,...,T\}\rightarrow\{1,2,...,n\}$ be the following function
\[f(i)=\lfloor n^{1-\frac{1}{i}}\rfloor~,\]
where $T=\frac{\log n}{4\log\beta_r+1}$. For every $1\leq j\leq n$, note that $j\leq f(i)$ if and only if $j\leq n^{1-\frac{1}{i}}$, i.e., $\frac{1}{1-\delta_j}\leq i$. Therefore, we have $f^{-1}(j)=\left\lceil\frac{1}{1-\delta_j}\right\rceil$. In addition, we have 
\begin{equation} \label{eq:FOfFInv1}
f(f^{-1}(j))=\lfloor n^{1-\frac{1}{f^{-1}(j)}}\rfloor\leq n^{1-\frac{1}{\frac{1}{1-\delta_j}+1}}=n^{\frac{1}{2-\delta_j}}~.
\end{equation}

To use Lemma \ref{lemma:SourcewiseToPrioritized} with the function $f$, we first have to specify for every possible subset $A\subseteq V$, which path-reporting source-wise spanner we use. For $A\subseteq V$ with size $|A|\leq\frac{n}{2\beta_r^2}$, we use the source-wise spanner from Theorem \ref{thm:SubsetSpanner1} (for $r=1$), Theorem \ref{thm:SubsetSpanner2} (for $r=2$), or Theorem \ref{thm:SubsetSpanner3} (for $r=3$), with
\[k=k(|A|)=\left\lceil\frac{\log|A|}{\log n-2\log\beta_r-\log|A|}\right\rceil~.\]
Note that $|A|^{1+\frac{1}{k}}\leq\frac{n}{\beta_r^2}$, and in particular $q<1$, where $q=\log_n(|A|^{1+\frac{1}{k}})$ is the value from these theorems. Then, this source-wise spanner has stretch $a_rk$, query time $O(\log k)$ and size
\[O_\epsilon\left(k|A|^{1+\frac{1}{k}}\cdot\beta_r+n\log\log n\right)=O_\epsilon\left(k\cdot\frac{n}{\beta_r^2}\cdot\beta_r+n\log\log n\right)=O_\epsilon\left(n\log\log n\right)~,\]
where in the last step we used the fact that $k\leq\log n\leq\beta_r$.

For $A\subseteq V$ with size $|A|>\frac{n}{2\beta_r^2}$, the choice of $k$ that we used above is too large, negative or undefined. Thus, in this case we use $k=\log n$ or any other arbitrary\footnote{Note that in Lemma \ref{lemma:SourcewiseToPrioritized}, source-wise spanners of subsets $A\subseteq V$ that are larger than $f(T)$ are not used. In our case, $f(T)=\left\lfloor\frac{n}{2\beta_r^4}\right\rfloor\leq\frac{n}{2\beta_r^2}$, and therefore the choice of $k$ does not matter in case that $|A|>\frac{n}{2\beta_r^2}$.} value of $k$.

Lastly, for the role of the path-reporting spanner in Lemma \ref{lemma:SourcewiseToPrioritized}, we use the path-reporting spanner from Theorem \ref{thm:ES6}, with $k=\frac{\log n}{\log\log\log n}$. This path-reporting spanner has size\newline $O\left(\left\lceil\frac{k\log\log n\cdot\log\log\log n}{\epsilon\log n}\right\rceil\cdot n^{1+\frac{1}{k}}\right)=O_\epsilon\left(n(\log\log n)^2\right)$, stretch $(4+\epsilon)k=(4+\epsilon)\frac{\log n}{\log\log\log n}$ and query time $O(\log k)=O(\log\log n)$.

For $r=3$, we actually use the path-reporting spanner from Theorem \ref{thm:ES7}, with $k=\frac{\log n}{\log\log n}$, instead of the one from Theorem \ref{thm:ES6}. This spanner has stretch $O\left(\frac{\log n}{\log\log n}\right)$, query time $O(1)$, and size $O(n\log n)$.

We now apply Lemma \ref{lemma:SourcewiseToPrioritized}. Let $(v_j,v_{j'})$ be a query such that $j<j'$. If $j\leq f(T)=\left\lfloor\frac{n}{2\beta_r^4}\right\rfloor$ - i.e., $j\leq\frac{1}{2}n^{1-\frac{4\log\beta_r}{\log n}}=\frac{1}{2}n^{1-\eta_r(n)}$ - then this query has the following stretch in the resulting path-reporting prioritized spanner:
\begin{eqnarray*}
a_rk\left(f(f^{-1}(j))\right)&\stackrel{(\ref{eq:FOfFInv1})}{\leq}&a_rk(n^{\frac{1}{2-\delta_j}})\\
&=&a_r\left\lceil\frac{\log(n^{\frac{1}{2-\delta_j}})}{\log n-2\log\beta_r-\log(n^{\frac{1}{2-\delta_j}})}\right\rceil\\
&=&a_r\left\lceil\frac{\frac{1}{2-\delta_j}}{1-\frac{2\log\beta_r}{\log n}-\frac{1}{2-\delta_j}}\right\rceil\\
&=&a_r\left\lceil\frac{1}{1-\delta_j-\frac{(2-\delta_j)2\log\beta_r}{\log n}}\right\rceil\\
&\leq&a_r\left\lceil\frac{1}{1-\delta_j-\frac{2\cdot2\log\beta_r}{\log n}}\right\rceil=a_r\left\lceil\frac{1}{1-\delta_j-\eta_r(n)}\right\rceil~.
\end{eqnarray*}
The query time of such query $(v_j,v_{j'})$ is $O\left(\log k\left(f(f^{-1}(j))\right)\right)\leq O\left(\log\left\lceil\frac{1}{1-\delta_j-\eta_r(n)}\right\rceil\right)$, when $r\in\{1,2\}$. For $r=3$, this query time is $O(1)$, since this is the query time of the source-wise spanners from Theorem \ref{thm:SubsetSpanner3}.

For $j>f(T)=\left\lfloor\frac{n}{2\beta_r^4}\right\rfloor$, i.e., for $j>\frac{1}{2}n^{1-\eta_r(n)}$, the query $(v_j,v_{j'})$ (where $j<j'$) has stretch $(4+\epsilon)\frac{\log n}{\log\log\log n}$ and query time $O(\log\log n)$, when $r\in\{1,2\}$ (these are the same stretch and query time of the path-reporting spanner from Theorem \ref{thm:ES6}). For $r=3$, such query has stretch $O\left(\frac{\log n}{\log\log n}\right)$ and query time $O(1)$ (as in Theorem \ref{thm:ES7}, with $k=\frac{\log n}{\log\log n}$).

Finally, we bound the size of the resulting path-reporting prioritized spanner. Here we divide the analysis for the different $r$'s.

For $r=1$, note that $T=\frac{\log n}{4\log\beta_1+1}=O_\epsilon\left(\frac{\log n}{\log\log n\cdot\log\log\log n}\right)$. Thus, the size of the prioritized spanner is
\begin{eqnarray*}
&&O_\epsilon\left(n+n(\log\log n)^2+T\cdot n\log\log n\right)\\
&=&O_\epsilon\left(n(\log\log n)^2+\frac{\log n}{\log\log n\cdot\log\log\log n}\cdot n\log\log n\right)
=O_\epsilon\left(\frac{n\log n}{\log\log\log n}\right)~.
\end{eqnarray*}

For $r=2$, note that $T=\frac{\log n}{4\log\beta_2+1}=O_\epsilon\left(\frac{\log n}{\log\log n}\right)$. Thus, the size of the prioritized spanner is
\begin{eqnarray*}
O_\epsilon\left(n+n(\log\log n)^2+T\cdot n\log\log n\right)
&=&O_\epsilon\left(n(\log\log n)^2+\frac{\log n}{\log\log n}\cdot n\log\log n\right)\\
&=&O_\epsilon\left(n\log n\right)~.
\end{eqnarray*}

For $r=3$, note that $T=\frac{\log n}{4\log\beta_3+1}=O\left(\frac{\log n}{\log\log n}\right)$. Thus, the size of the prioritized spanner is
\begin{eqnarray*}
O\left(n+n\log n+T\cdot n\log\log n\right)
&=&O\left(n\log n+\frac{\log n}{\log\log n}\cdot n\log\log n\right)\\
&=&O\left(n\log n\right)~.
\end{eqnarray*}

\end{proof}

The following theorem improves the size of the prioritized spanners above to $O(n(\log\log n)^2)$, while only slightly increase the stretch (it actually increases it by a factor of at most $2$, with respect to the prioritized spanner from Item 1 in Theorem \ref{thm:PrioritizedSpanner1}). This time, we use only one of the source-wise spanners from Section \ref{sec:SubsetSourcewiseSpanners}. Using the other two provides similar results.

\begin{theorem} \label{thm:PrioritizedSpanner2}
For every undirected weighted $n$-vertex graph $G=(V,E)$, and a real parameter $\epsilon$ such that $\sqrt{\frac{\log\log n\cdot\log\log\log n}{\log n}}\leq\epsilon\leq\frac{1}{2}$, there is a path-reporting prioritized spanner with size $O_\epsilon\left(n(\log\log n)^2\right)$, prioritized stretch 
\[\alpha(j)=\begin{cases}
    (4+\epsilon)\left\lceil\frac{1+\delta_j}{1-\delta_j-\eta_1(n)}\right\rceil & j\leq\frac{1}{2}n^{1-\eta_1(n)} \\
    (4+\epsilon)\frac{\log n}{\log\log\log n} & j>\frac{1}{2}n^{1-\eta_1(n)}
\end{cases}\]
and prioritized query time $O\left(\log\alpha(j)\right)$, where\newline $\eta_1(n)=\frac{4\log\beta_1(\epsilon,\log n)}{\log n}=O_\epsilon\left(\frac{\log\log n\cdot\log\log\log n}{\log n}\right)$.
\end{theorem}

\begin{proof}

As in the proof of Theorem \ref{thm:PrioritizedSpanner1}, we denote $\beta_1=\beta_1(\epsilon,\log n)$. Let $f:\{1,2,...,T\}\rightarrow\{1,2,...,n\}$ be the function
\[f(i)=\lfloor n^{1-\frac{1}{2^i}}\rfloor~,\]
where $T=\log\log n-\log(4\log\beta_1+1)$. For every $1\leq j\leq n$, note that $j\leq f(i)$ if and only if $j\leq n^{1-\frac{1}{2^i}}$, i.e., $\log\left(\frac{1}{1-\delta_j}\right)\leq i$. Therefore, we have $f^{-1}(j)=\left\lceil\log\left(\frac{1}{1-\delta_j}\right)\right\rceil$. As a result, we have 
\begin{equation} \label{eq:FOfFInv2}
f(f^{-1}(j))=\lfloor n^{1-2^{-f^{-1}(j)}}\rfloor\leq n^{1-2^{-\log\left(\frac{1}{1-\delta_j}\right)-1}}=n^{1-\frac{1-\delta_j}{2}}=n^{\frac{1+\delta_j}{2}}~.
\end{equation}

We now specify the source-wise spanner that we use for each subset $A\subseteq V$, in order to apply Lemma \ref{lemma:SourcewiseToPrioritized}. For $A\subseteq V$ with size $|A|\leq\frac{n}{2\beta_1^2}$, we use the source-wise spanner from Theorem \ref{thm:SubsetSpanner1}, with the same $k$ as in the proof of Theorem \ref{thm:PrioritizedSpanner1}:
\[k=k(|A|)=\left\lceil\frac{\log|A|}{\log n-2\log\beta_1-\log|A|}\right\rceil~.\]
Note that $|A|^{1+\frac{1}{k}}\leq\frac{n}{\beta_1^2}$, and in particular $q<1$, where $q=\log_n(|A|^{1+\frac{1}{k}})$ is the value from Theorem \ref{thm:SubsetSpanner1}. Then, this source-wise spanner has stretch $(4+\epsilon)k$, query time $O(\log k)$ and size
\[O_\epsilon\left(k|A|^{1+\frac{1}{k}}\cdot\beta_1+n\log\log n\right)=O_\epsilon\left(k\cdot\frac{n}{\beta_1^2}\cdot\beta_1+n\log\log n\right)=O_\epsilon\left(n\log\log n\right)~,\]
where in the last step we used the fact that $k\leq\log n\leq\beta_1$.

For $A\subseteq V$ with size $|A|>\frac{n}{2\beta_1^2}$, the choice of $k$ does not affect the resulting prioritized spanner, thus we choose it arbitrarily.

Lastly, for the role of the path-reporting spanner in Lemma \ref{lemma:SourcewiseToPrioritized}, we use the path-reporting spanner from Theorem \ref{thm:ES6}, with $k=\frac{\log n}{\log\log\log n}$. This path-reporting spanner has size\newline $O\left(\left\lceil\frac{k\log\log n\cdot\log\log\log n}{\epsilon\log n}\right\rceil\cdot n^{1+\frac{1}{k}}\right)=O_\epsilon\left(n(\log\log n)^2\right)$, stretch $(4+\epsilon)k=(4+\epsilon)\frac{\log n}{\log\log\log n}$ and query time $O(\log k)=O(\log\log n)$.

We now apply Lemma \ref{lemma:SourcewiseToPrioritized}. Let $(v_j,v_{j'})$ be a query such that $j<j'$. If $j\leq f(T)=\left\lfloor\frac{n}{2\beta_1^4}\right\rfloor$ - i.e., $j\leq\frac{1}{2}n^{1-\frac{4\log\beta_1}{\log n}}=\frac{1}{2}n^{1-\eta_1(n)}$ - then this query has the following stretch in the resulting path-reporting prioritized spanner:
\begin{eqnarray*}
(4+\epsilon)k\left(f(f^{-1}(j))\right)&\stackrel{(\ref{eq:FOfFInv2})}{\leq}&(4+\epsilon)k(n^{\frac{1+\delta_j}{2}})\\
&=&(4+\epsilon)\left\lceil\frac{\log(n^{\frac{1+\delta_j}{2}})}{\log n-2\log\beta_1-\log(n^{\frac{1+\delta_j}{2}})}\right\rceil\\
&=&(4+\epsilon)\left\lceil\frac{\frac{1+\delta_j}{2}}{1-\frac{2\log\beta_1}{\log n}-\frac{1+\delta_j}{2}}\right\rceil\\
&=&(4+\epsilon)\left\lceil\frac{1+\delta_j}{1-\delta_j-\frac{4\log\beta_1}{\log n}}\right\rceil=(4+\epsilon)\left\lceil\frac{1+\delta_j}{1-\delta_j-\eta_1(n)}\right\rceil~.
\end{eqnarray*}
The query time of such query $(v_j,v_{j'})$ is $O\left(\log k\left(f(f^{-1}(j))\right)\right)\leq O\left(\log\left\lceil\frac{1+\delta_j}{1-\delta_j-\eta_1(n)}\right\rceil\right)$.

For $j>f(T)=\left\lfloor\frac{n}{2\beta_1^4}\right\rfloor$, i.e., for $j>\frac{1}{2}n^{1-\eta_1(n)}$, the query $(v_j,v_{j'})$ (where $j<j'$) has stretch $(4+\epsilon)\frac{\log n}{\log\log\log n}$ and query time $O(\log\log n)$ (these are the same stretch and query time of the path-reporting spanner from Theorem \ref{thm:ES6}).

Finally, since $T<\log\log n$, the size of the resulting path-reporting prioritized spanner is
\[O_\epsilon\left(n+n(\log\log n)^2+T\cdot n\log\log n\right)=O_\epsilon\left(n(\log\log n)^2\right)~.\]
This completes the proof of the theorem.
\end{proof}

%When producing path-reporting prioritized spanners, the approach of Lemma \ref{lemma:SourcewiseToPrioritized} was as follows. Consider a monotonically increasing prefixes of the priority ranking $(v_1,v_2,...,v_n)$: $A_1\subseteq A_2\subseteq ...\subseteq A_T$, and for each one of them to apply a source-wise spanner. For the rest of the vertices, $V\setminus A_T$, we used a non-prioritized path-reporting spanner. Note that we must have $V\setminus A_T\neq\empty$ to keep the size of the resulting prioritized spanner small. Otherwise, if $A_T=V$ and we use our source-wise spanners for $A_T$ as well, its size would be at least $O(n\log^bn)$ (this is the size of the path-reporting source-wise spanner from Theorem \ref{thm:SubsetSpanner3} for $A=V$).

%Suppose that we are interested in path-reporting prioritized spanners that have query time $O(1)$. Recall that by "query time $O(1)$" we mean that the running time of the query algorithm is proportional to the length of the output path. Then, for $V\setminus A_T$, we must use a path-reporting spanner that has query time $O(1)$.

\bibliography{hopset}

\appendix

\section{Proof of Theorem \ref{thm:HopsetsExtendedVersion}} \label{sec:ProofExtendedVersion}

To prove Theorem \ref{thm:HopsetsExtendedVersion}, we dive into the details of Theorem \ref{thm:HopsetsWithLargeStretch} from \cite{NS22}. In \cite{NS22}, a hopset $H$ was constructed, in a way that generalizes other constructions of hopsets, that appear in the papers \cite{EN19,HP17,BP20}, and many others. 

In all of these constructions, given an undirected weighted graph $G=(V,E)$ on $n$ vertices, a certain choice of a hierarchy of sets $V=A_0\supseteq A_1\supseteq\cdots A_F=\emptyset$ is made. This hierarchy of subsets is a concept that appears in many of the constructions regarding spanners, hopsets, preservers, distance oracles and more. It was originated in the seminal paper by Thorup and Zwick \cite{TZ01}. In all the results that rely on this technique, (a variation of) the notions of \textit{pivots} and \textit{bunches} also appear (here $i\in[0,F-1]$):
\begin{enumerate}
    \item The $i$-th {\em pivot} of a vertex $u\in V$ is the closest vertex to $u$ from the set $A_i$ (breaking ties in some consistent manner). It is denoted by $p_i(u)$.
    \item The $i$-th {\em bunch} of a vertex $u\in V$ is the set of all vertices $v\in A_i$ that are closer to $u$ than $p_{i+1}(u)$. It is denoted by $B_i(u)=\{v\in A_i\;|\;d_G(u,v)<d_G(u,p_{i+1}(u))\}$. For $i=F-1$, since $A_F=\emptyset$, the pivot $p_F(u)$ is not defined, therefore in this case we simply define $d_G(u,p_F(u))=\infty$ and thus $B_{F-1}(u)=A_{F-1}$.
\end{enumerate}

The general idea of constructing a hopset out of these definitions is to connect each vertex $u\in V$ by a hop-edge to any pivot $p_i(u)$, for every $i\in[0,F-1]$. In addition, every $u\in V$ connects by a hop-edge to every $v\in B_j(u)$, for indices $j$ in some range: in \cite{EN19,HP17}, a vertex $u\in A_i\setminus A_{i+1}$ connects only to the vertices of $B_i(u)$. The distance oracle of Thorup and Zwick \cite{TZ01} can be viewed as a hopset that is constructed in the very same way (as was observed by \cite{HP17}), where every $u\in A_i\setminus A_{i+1}$ connects to the vertices in the bunch $B_j(u)$, for every $j\in[i,F-1]$.

To generalize these constructions, \cite{NS22} sets a function $f$ that controls the range of the indices $j$ such that every $u\in A_i\setminus A_{i+1}$ connects by a hop-edge to any vertex in $B_j(u)$. That is, for every $i\in[0,F-1]$, every $u\in A_i\setminus A_{i+1}$ connects to all of the vertices in $B_j(u)$, for every $j\in[i,f(i)]$. The specific choice of the function $f$ is
\[f(i)=\lfloor\frac{i}{c}\rfloor\cdot c+c-1~,\]
where $c>0$ is an integer parameter. This function rounds $i$ to the minimal multiple of $c$ that is larger than $i$, then subtract $1$.

At this point, we can define the hopset from \cite{NS22}:
\begin{equation} \label{eq:HopsetDef}
H=\bigcup_{u\in V}\bigcup_{i=0}^{F-1}\{(u,p_i(u))\}\cup\bigcup_{i=0}^{F-1}\bigcup_{u\in A_i\setminus A_{i+1}}\bigcup_{j=i}^{f(i)}\bigcup_{v\in B_j(u)}\{(u,v)\}~.
\end{equation}

We also denote
\[f^{-1}(i)=\min\{j\in\mathbb{N}\;|\;f(j)\geq i\}\]
(in the case of our specific $f$, it is easy to see that $f^{-1}(i)=\left\lfloor\frac{i}{c}\right\rfloor\cdot c$). Note that $j\in[i,f(i)]$ if and only if $i\in[f^{-1}(j),j]$. Indeed, by the definition of $f^{-1}$, if $j\leq f(i)$ then $f^{-1}(j)\leq i$. In the other direction, we use the fact that $f$ is non-decreasing, and again the definition of $f^{-1}$, to conclude that if $f^{-1}(j)\leq i$, then $j\leq f(f^{-1}(j))\leq f(i)$. As a result, we can write the definition of $H$ also as
\begin{eqnarray*}
H&=&\bigcup_{u\in V}\bigcup_{i=0}^{F-1}\{(u,p_i(u))\}\cup\bigcup_{j=0}^{F-1}\bigcup_{i=f^{-1}(j)}^{j}\bigcup_{u\in A_i\setminus A_{i+1}}\bigcup_{v\in B_j(u)}\{(u,v)\}\\
&=&\bigcup_{u\in V}\bigcup_{i=0}^{F-1}\{(u,p_i(u))\}\cup\bigcup_{j=0}^{F-1}\bigcup_{u\in A_{f^{-1}(j)}}\bigcup_{v\in B_j(u)}\{(u,v)\}~.    
\end{eqnarray*}

A crucial observation regarding the stretch and hopbound in \cite{NS22}, is that it {\em does not rely on the choice of the hierarchy} $A_0\supseteq A_1\supseteq\cdots A_F$, and it is actually true for \textit{any} choice of such hierarchy. This means that the following theorem is true, even without specifying the choice of the sets $A_i$. The proof of this theorem can be found in \cite{NS22}.

\begin{theorem} \label{thm:StretchAndHopbound}
The hopset $H$ has stretch $8c+3$ and hopbound at most $(1+\frac{1}{c})^F\cdot(2c+2)^{2\lfloor\frac{F}{c}\rfloor}$.
\end{theorem}

Next, we specify the way that the hierarchy $A_0\supseteq A_1\supseteq\cdots A_F$ is chosen. Fix some integer parameter $k>0$ and let $\{\lambda_i\}$ be a sequence that is defined by the following recurrence relation
\begin{equation} \label{eq:LambdaDef}
    \lambda_0=1,\;\;\;\;\;\lambda_i=1+\sum_{l<f^{-1}(i)}\lambda_i~.
\end{equation}
The set $A_0$ is defined to be $V$. For every $i\geq0$, the set $A_{i+1}$ is obtained by sampling each vertex of $A_i$ independently with probability $\delta\cdot n^{-\frac{\lambda_i}{k}}$. Here, $0<\delta\leq\frac{1}{2}$ is a real parameter that will be determined later. This probability choice is slightly different from the one in \cite{NS22}. In particular, the parameter $\delta$ was not used in \cite{NS22}. This, however, does not change the result stated in Theorem \ref{thm:StretchAndHopbound}.

\begin{lemma} \label{lemma:HopsetSize}
The expected size of $H$ is $O\left(Fn+\frac{c}{\delta}\cdot n^{1+\frac{1}{k}}\right)$.
%\[O\left(Fn+n^{1+\frac{1}{k}}\cdot\frac{1}{\delta}\sum_{j=0}^{F-1}\delta^{f^{-1}(j)}\right)~.\]
\end{lemma}

\begin{proof}

We denote 
\[H_P=\bigcup_{u\in V}\bigcup_{i=0}^{F-1}\{(u,p_i(u))\},\;\;\;\;\;H_B=\bigcup_{j=0}^{F-1}\bigcup_{u\in A_{f^{-1}(j)}\setminus A_{j+1}}\bigcup_{v\in B_j(u)}\{(u,v)\}~,\]
so that $H=H_P\cup H_B$. The size of $H_P$ is at most $\sum_{u\in V}\sum_{i=0}^{F-1}1=Fn$.

The expected size of the set $A_i$ is $|A_i|=\delta^i n^{1-\frac{1}{k}\sum_{l<i}\lambda_l}$. To bound the size of the bunch $B_j(u)$, for some $u\in V$, we use the following standard argument. This argument appears in \cite{TZ01} and in most of the constructions that rely on the same technique. Let $u=u_1,u_2,u_3,u_4,...$ be a list of the vertices in $A_j$, ordered by their distance from $u$. Since each of them is sampled into $A_{j+1}$ independently with probability $\delta n^{-\frac{\lambda_j}{k}}$, the index of the first vertex on this list to be in $A_{j+1}$ is bounded by a geometric random variable with expectation $\frac{1}{\delta}n^{\frac{\lambda_j}{k}}$. This index equals to $|B_j(u)|+1$, by the definition of bunches. Hence, in expectation,
\begin{equation} \label{eq:BunchSize}
    |B_j(u)|\leq\frac{1}{\delta}n^{\frac{\lambda_j}{k}}-1
\end{equation}

We use these bounds on the expected size of $A_i$ and $B_j(u)$, to bound the expected size of $H_B$.
\begin{eqnarray*}
\mathbb{E}[|H_B|]&\leq&\sum_{j=0}^{F-1}\sum_{u\in A_{f^{-1}(j)}}\mathbb{E}[|B_j(u)|]
\leq\sum_{j=0}^{F-1}\mathbb{E}[|A_{f^{-1}(j)}|]\cdot\frac{1}{\delta}n^{\frac{\lambda_j}{k}}\\
&\leq&\sum_{j=0}^{F-1}\delta^{f^{-1}(j)}n^{1-\frac{1}{k}\sum_{l<f^{-1}(j)}\lambda_l}\cdot\frac{1}{\delta}n^{\frac{\lambda_j}{k}}
\leq\frac{1}{\delta}\sum_{j=0}^{F-1}\delta^{f^{-1}(j)}n^{1+\frac{1}{k}(\lambda_j-\sum_{l<f^{-1}(j)}\lambda_l)}~.
\end{eqnarray*}
We use the fact that $\lambda_j=1+\sum_{l<f^{-1}(j)}\lambda_l$ from Equation (\ref{eq:LambdaDef}) to bound the term in the exponent (in fact, this is exactly the reason for this choice of the sequence $\{\lambda_i\}$).
\[\mathbb{E}[|H_B|]\leq\frac{1}{\delta}\sum_{j=0}^{F-1}\delta^{f^{-1}(j)}n^{1+\frac{1}{k}(\lambda_j-\sum_{l<f^{-1}(j)}\lambda_l)}= n^{1+\frac{1}{k}}\cdot\frac{1}{\delta}\sum_{j=0}^{F-1}\delta^{f^{-1}(j)}~.\]

Recall our choice of the function $f(i)=\left\lfloor\frac{i}{c}\right\rfloor\cdot c+c-1$. For this $f$, we have $f^{-1}(j)=\left\lfloor\frac{j}{c}\right\rfloor\cdot c$. Therefore, we can write
\[\mathbb{E}[|H_B|]\leq n^{1+\frac{1}{k}}\cdot\frac{1}{\delta}\sum_{j=0}^{F-1}\delta^{f^{-1}(j)}<n^{1+\frac{1}{k}}\cdot\frac{c}{\delta}\sum_{l=0}^{\infty}\delta^{l\cdot c}=O\left(\frac{c}{\delta}\cdot n^{1+\frac{1}{k}}\right)~,\]
where in the last step we used the assumption that $\delta\leq\frac{1}{2}$.

Finally, we conclude that the expected size of $H$ is at most
\[\mathbb{E}[|H_P|+|H_B|]=O\left(Fn+\frac{c}{\delta}\cdot n^{1+\frac{1}{k}}\right)~,\]
as desired.

\end{proof}

To prove Theorem \ref{thm:HopsetsExtendedVersion}, we need to show the partition of $H$ to three subsets. The first subset, $H_1$, is in fact the set $H_P$ that was defined in the proof of Lemma \ref{lemma:HopsetSize}:
\[H_1=\bigcup_{u\in V}\bigcup_{i=0}^{F-1}\{(u,p_i(u))\}~.\]

To define the sets $H_2,H_3$, we divide the set $H_B$ from the proof of Lemma \ref{lemma:HopsetSize} to two disjoint subsets:
\[H_2=\bigcup_{j=0}^{c-1}\bigcup_{u\in A_{f^{-1}(j)}}\bigcup_{v\in B_j(u)}\{(u,v)\},\;\;\;H_3=\bigcup_{j=c}^{F-1}\bigcup_{u\in A_{f^{-1}(j)}}\bigcup_{v\in B_j(u)}\{(u,v)\}~.\]
That is, $H_2$ captures all the hop-edges between a vertex and its bunch members, for bunches of level $j\in[0,c-1]$, while $H_3$ does the same for levels $j\in[c,F-1]$. %The proof of the following lemma is identical to the proof of Lemma \ref{lemma:HopsetSize}, where the index $j$ runs from $c$ to $F-1$, instead of from $0$ to $F-1$.

\begin{lemma} \label{lemma:H3Size}
In expectation, $|H_3|=O\left(c\delta^{c-1}\cdot n^{1+\frac{1}{k}}\right)$.
%$|H_3|=O\left(n^{1+\frac{1}{k}}\cdot\frac{1}{\delta}\sum_{j=c}^{F-1}\delta^{f^{-1}(j)}\right)$.
\end{lemma}

\begin{proof}

Repeating the proof of Lemma \ref{lemma:HopsetSize} proves that $|H_3|=O\left(n^{1+\frac{1}{k}}\cdot\frac{1}{\delta}\sum_{j=c}^{F-1}\delta^{f^{-1}(j)}\right)$. The only difference is that now the index $j$ runs from $c$ to $F-1$, instead of from $0$ to $F-1$. 

Using the fact that $f^{-1}(j)=\left\lfloor\frac{j}{c}\right\rfloor\cdot c$, we conclude that
\[|H_3|=O\left(n^{1+\frac{1}{k}}\cdot\frac{1}{\delta}\sum_{j=c}^{F-1}\delta^{f^{-1}(j)}\right)<O\left(n^{1+\frac{1}{k}}\cdot\frac{c}{\delta}\sum_{l=1}^{\infty}\delta^{l\cdot c}\right)=O\left(c\delta^{c-1}\cdot n^{1+\frac{1}{k}}\right)~,\]
where in the last step we used the assumption that $\delta\leq\frac{1}{2}$.

\end{proof}

The following lemma proves the existence of a path-reporting preserver for the set $H_1$. It can be easily derived from Lemma 1 in \cite{ES23}.
\begin{lemma} \label{lemma:H1Preserver}
The set $H_1$ has a path-reporting preserver with size $O(Fn)$.
\end{lemma}

\begin{proof}

For every $u,v\in V$, denote by $P_{u,v}$ a shortest path in $G$ between $u$ and $v$ that was chosen in a consistent way. Denote $E_1^i=\bigcup_{u\in V}P_{u,p_i(u)}$ and $E_1=\bigcup_{i=0}^{F-1}E_1^i$.

Fix some $i\in[0,F-1]$. For every vertex $u\in V$, let $q_i(u)$ be a pointer to the next vertex on the path $P_{u,p_i(u)}$ towards $p_i(u)$. Note that if $x$ is a vertex on $P_{u,p_i(u)}$, then $x$ has the same $i$-th pivot as $u$: seeking contradiction, if $d_G(x,p_i(x))<d_G(x,p_i(u))$, then also 
\[d_G(u,p_i(x))\leq d_G(u,x)+d_G(x,p_i(x))<d_G(u,x)+d_G(x,p_i(u))=d_G(u,p_i(u))~.\]
This contradicts the fact that $p_i(u)$ is the $i$-th pivot of $u$. Therefore, $d_G(x,p_i(x))=d_G(x,p_i(u))$, and since the pivots are chosen in a consistent manner, this means that $p_i(x)=p_i(u)$. Hence, the edge $(x,q_i(x))$ is the next edge after $x$ on the path $P_{u,p_i(u)}$. This implies that every edge in $E_1^i$ is of the form $(x,q_i(x))$, and thus $|E_1^i|\leq n$. By the union bound, $|E_1|\leq\sum_{i=0}^{F-1}n=Fn$.

We also define an oracle $D_1$, that given a pair of the form $(u,p_i(u))$, uses the pointers $q_i(x)$, starting with $x=u$ and reporting the edges $(x,q_i(x))$, until it gets to $x=p_i(u)$. The storage of the oracle $D_1$ is the total number of the pointers $q_i(u)$, for every $u\in V$ and every $i\in[0,F-1]$. That is, $|D_1|\leq Fn$.

We conclude that $(E_1,D_1)$ is a path-reporting preserver for $G,H_1$, with size $O(Fn)$.

\end{proof}

In \cite{TZ01}, Thorup and Zwick provided a distance oracle that is constructed in a similar way to the hopset $H$. There, for every vertex $u\in A_i\setminus A_{i+1}$, for every $i\in[0,F-1]$, the distances between $u$ and every $v\in B_j(u)$, for every $j\in[i,F-1]$, were stored. Then, to construct a spanner, rather than a distance oracle, they added a shortest path between any such pair of vertices $u,v$. To claim that the size of the resulting spanner is small, they used an argument that is based on the notion of \textit{clusters}.

Let $v\in V$, and suppose that $v\in A_j\setminus A_{j+1}$. We use the same definition from \cite{TZ01} for the cluster of $v$, but with a small modification - we only include in this cluster vertices $u\in A_f^{-1}(j)$, since these are the only vertices that might be connected by a hop-edge to $v$. The formal definition is
\[C(v)=\{u\in V\;|\;v\in B_j(u)\}=\{u\in A_{f^{-1}(j)}\;|\;d_G(u,v)<d_G(u,p_{j+1}(u)\}~.\]
Note that now we may rewrite the definition of $H_2$ as
\[H_2=\bigcup_{j=0}^{c-1}\bigcup_{u\in A_{f^{-1}(j)}\setminus A_{j+1}}\bigcup_{v\in B_j(u)}\{(u,v)\}=\bigcup_{j=0}^{c-1}\bigcup_{v\in A_j}\bigcup_{u\in C(v)}\{(u,v)\}~.\]
The definition for $H_3$ is the same but with the index $j$ runs from $c$ to $F-1$, instead of from $0$ to $c-1$. %Notice that for every $j\in[0,c-1]$, we have $f^{-1}(j)=0$, because of the definition of the function $f$. Therefore, in the case of $H_2$, we can actually write
%\[H_2=\bigcup_{j=0}^{c-1}\bigcup_{v\in A_j}\bigcup_{u\in C(v)}\{(u,v)\}=~.\]

We now make the same argument as in \cite{TZ01} regarding clusters, to prove the following lemma.

\begin{lemma} \label{lemma:H2Preserver}
The set $H_2$ has a path-reporting preserver with expected size $O\left(\frac{c}{\delta}\cdot n^{1+\frac{1}{k}}\right)$.
%$O\left(n^{1+\frac{1}{k}}\cdot\frac{1}{\delta}\sum_{j=0}^{c-1}\delta^{f^{-1}(j)}\right)$.

\end{lemma}

\begin{proof}

For every $u,v\in V$, denote by $P_{u,v}$ a shortest path in $G$ between $u$ and $v$ that was chosen in a consistent way. For every 
$j\in[0,c-1]$ and $v\in A_j$, denote $E_2^j(v)=\bigcup_{u\in C(v)}P_{u,v}$, and $E_2=\bigcup_{j=0}^{c-1}\bigcup_{v\in A_j}E_2^j(v)$.

Fix some $j\in[0,c-1]$ and $v\in A_j$. For every vertex $u\in C(v)$, let $q_{j,v}(u)$ be a pointer to the next vertex on the path $P_{u,v}$ towards $v$. We claim that if $x$ is a vertex on $P_{u,v}$, then also $x\in C(v)$: first, notice that for every $j\in[0,c-1]$, we have $f^{-1}(j)=0$, due to the definition of the function $f$. Therefore, $x\in A_0=V$ is trivial. Seeking contradiction, we assume that $d_G(x,p_{j+1}(x))\leq d_G(x,v)$, then also 
\[d_G(u,p_{j+1}(u))\leq d_G(u,p_{j+1}(x))\leq d_G(u,x)+d_G(x,p_{j+1}(x))\leq d_G(u,x)+d_G(x,v)=d_G(u,v)~.\]
Thus, $v\notin B_j(u)$. This contradicts the fact that $u\in C(v)$. Therefore, $x\in A_0$ and $d_G(x,v)<d_G(x,p_{j+1}(x))$ - i.e., $x\in C(v)$. 

As a result, the edge $(x,q_{j,v}(x))$ is the next edge after $x$ on the path $P_{u,v}$. This implies that every edge in $E_2^j(v)$ is of the form $(x,q_{j,v}(x))$ for some $x\in C(v)$, and thus $|E_2^j(v)|\leq|C(v)|$. By the union bound, 
\begin{eqnarray*}
\mathbb{E}[|E_2|]&\leq&\sum_{j=0}^{c-1}\sum_{v\in A_j}\mathbb{E}[|C(v)|]
=\sum_{j=0}^{c-1}\sum_{v\in B_j(u)}\sum_{u\in C(v)}1
=\sum_{j=0}^{c-1}\sum_{u\in V}\sum_{v\in B_j(u)}1\\
&=&\sum_{j=0}^{c-1}\sum_{u\in V}\mathbb{E}[|B_j(u)|]
\stackrel{(\ref{eq:BunchSize})}{\leq}\sum_{j=0}^{c-1}\sum_{u\in V}\frac{1}{\delta}n^{\frac{\lambda_j}{k}}~.
\end{eqnarray*}

Recall the definition of the sequence $\{\lambda_i\}$ from Equation (\ref{eq:LambdaDef}). Since $f^{-1}(j)=0$ for every $j\in[0,c-1]$, we simply get that $\lambda_j=1$ for every $j\in[0,c-1]$. We conclude that the expected size of $E_2$ is at most
\[\sum_{j=0}^{c-1}\sum_{u\in V}\frac{1}{\delta}n^{\frac{\lambda_j}{k}}=\sum_{j=0}^{c-1}\sum_{u\in V}\frac{1}{\delta}n^{\frac{1}{k}}=\frac{c}{\delta}\cdot n^{1+\frac{1}{k}}~.\]

We also define an oracle $D_2$, that given a pair of the form $(u,v)$, where $j\in[0,c-1],\;v\in A_j$ and $u\in C(v)$, uses the pointers $q_{j,v}(x)$, starting with $x=u$ and reporting the edges $(x,q_{j,v}(x))$, until it gets to $x=v$. The storage of the oracle $D_2$ is the total number of the pointers $q_{j,v}(u)$, for every $j\in[0,c-1],\;v\in A_j$ and $u\in C(v)$. By the same computation as above, the expectation of this number is $|E_2|\leq\frac{c}{\delta}\cdot n^{1+\frac{1}{k}}$.

We conclude that $(E_2,D_2)$ is a path-reporting preserver for $G,H_2$, with size $O(\frac{c}{\delta}\cdot n^{1+\frac{1}{k}})$.

\end{proof}

We are now ready to prove Theorem \ref{thm:HopsetsExtendedVersion}.

\begin{theorem*}[Theorem \ref{thm:HopsetsExtendedVersion}]
Let $G=(V,E)$ be an undirected weighted graph on $n$ vertices, and let $1<c\leq k$ be two integer parameters. There is a hopset $H$ for $G$ with stretch $8c+3$, hopbound $O\left(\frac{1}{c}\cdot k^{1+\frac{2}{\ln c}}\right)$, and size $O\left(cn\log_ck+ck^{\frac{9}{c-1}}n^{1+\frac{1}{k}}\right)$.

Moreover, the hoposet $H$ can be divided into three disjoint subsets $H=H_1\cup H_2\cup H_3$, such that $H_1$ has a path-reporting preserver with size $O(cn\log_ck)$, $H_2$ has a path-reporting preserver with size $O\left(ck^{\frac{9}{c-1}}n^{1+\frac{1}{k}}\right)$, and 
\[|H_3|=O\left(ck^{-9}n^{1+\frac{1}{k}}\right)~.\]
\end{theorem*}

\begin{proof}

Recall the definition of the sequence $\{\lambda_i\}$ in Equation (\ref{eq:LambdaDef}). Using the fact that $f^{-1}(j)=\left\lfloor\frac{j}{c}\right\rfloor\cdot c$, we prove by induction that $\lambda_{ac+b}=(c+1)^a$, for every integer $a\geq0$ and $b\in[0,c-1]$. For $a=b=0$, $\lambda_0=1=(c+1)^0$ by definition. Fix some $j=ac+b$, and assume by induction that for every $l=a'c+b'\leq j$ we have $\lambda_j=(c+1)^{a'}$. If $b<c-1$, then $b+1\leq c-1$, therefore $f^{-1}(j+1)=\left\lfloor\frac{ac+b+1}{c}\right\rfloor\cdot c=ac$, and then
\begin{eqnarray*}
\lambda_{j+1}&=&1+\sum_{l<f^{-1}(j+1)}\lambda_l=1+\sum_{l<ac}\lambda_l=1+\sum_{a'<a}\sum_{b'=0}^{c-1}(c+1)^{a'}\\
&=&1+c\cdot\frac{(c+1)^a-1}{(c+1)-1}=(c+1)^a~.
\end{eqnarray*}

If $b=c-1$, then $j+1=(a+1)c$, and $f^{-1}(j+1)=\left\lfloor\frac{j+1}{c}\right\rfloor\cdot c=(a+1)c$. Then,
\begin{eqnarray*}
\lambda_{j+1}&=&1+\sum_{l<f^{-1}(j+1)}\lambda_l=1+\sum_{l<(a+1)c}\lambda_l=1+\sum_{a'<a+1}\sum_{b'=0}^{c-1}(c+1)^{a'}\\
&=&1+c\cdot\frac{(c+1)^{a+1}-1}{(c+1)-1}=(c+1)^{a+1}~.
\end{eqnarray*}
This completes the inductive proof.

Recall from the proof of Lemma \ref{lemma:HopsetSize} that the expected size of $A_i$, for every $i\in[0,F-1]$ is $\delta^i n^{1-\frac{1}{k}\sum_{l<i}\lambda_l}$. Assume that $i=ac+b$ for integers $a\geq0$ and $b\in[0,c-1]$. To compute the sum that appear in the exponent of $n$, we first assume that $b=0$. Then, 
\[\sum_{l<i}\lambda_l=\sum_{a'<a}\sum_{b'=0}^{c-1}(c+1)^{a'}=(c+1)^a-1~.\]
Therefore, the expected size of $A_i$ is
\[|A_i|=\delta^i n^{1-\frac{1}{k}\sum_{l<i}\lambda_l}=\delta^i n^{1-\frac{(c+1)^a-1}{k}}~.\]
For $a\geq\log_{c+1}(k+1)$, this value is at most $\delta^i\leq\left(\frac{1}{2}\right)^{c\log_{c+1}(k+1)}=(k+1)^{-\frac{c}{\log_2(c+1)}}\leq\frac{1}{k+1}$. Thus, with high probability, $A_i$ is empty for $i=c\lceil\log_{c+1}(k+1)\rceil$. This means that we can always assume that $F=c\lceil\log_{c+1}(k+1)\rceil$.

For proving the theorem, let $H$ be the hopset that was defined in Equation (\ref{eq:HopsetDef}). By Theorem \ref{thm:StretchAndHopbound} and by Lemma \ref{lemma:HopsetSize}, $H$ has the desired stretch, hopbound and size. Recall that $H$ can be partitioned into three disjoint subsets $H_1,H_2,H_3$ defined above. We saw in Lemma \ref{lemma:H1Preserver} that $H_1$ has a path-reporting preserver with size $O(Fn)$. We saw in Lemma \ref{lemma:H2Preserver} that $H_2$ has a path-reporting preserver with size $O\left(\frac{c}{\delta}\cdot n^{1+\frac{1}{k}}\right)$. Lastly, we saw in Lemma \ref{lemma:H3Size} that $|H_3|=O\left(c\delta^{c-1}\cdot n^{1+\frac{1}{k}}\right)$.

Choose $\delta=k^{-\frac{9}{c-1}}$. We substitute this value of $\delta$ and $F=c\lceil\log_{c+1}(k+1)\rceil$, and conclude that
\begin{itemize}
    \item $H_1$ has a path-reporting preserver with size $O(Fn)=O(cn\log_ck)$,
    \item $H_2$ has a path-reporting preserver with size $O\left(\frac{c}{\delta}\cdot n^{1+\frac{1}{k}}\right)=O\left(ck^{\frac{9}{c-1}}n^{1+\frac{1}{k}}\right)$,
    \item $|H_3|=O\left(c\delta^{c-1}\cdot n^{1+\frac{1}{k}}\right)=O\left(ck^{-9}n^{1+\frac{1}{k}}\right)$,
\end{itemize}
as desired.

\end{proof}

\end{document}